\documentclass{article}

\usepackage{mathtools}
\usepackage{graphicx}
\usepackage{amsmath}
\usepackage{amssymb}
\usepackage{amsthm} 
\usepackage{tikz}
\usetikzlibrary{arrows.meta, positioning}
\usetikzlibrary{decorations.pathmorphing}

\usepackage{pxfonts}
\usepackage{enumerate}
\usepackage{color}
\usepackage{mathdots}
\usepackage{sectsty}
\usepackage[hidelinks]{hyperref}
\usepackage{tikz-cd}
\usepackage{framed}
\usepackage{caption}
\usepackage{adjustbox}
\usepackage{bbold}
\allowdisplaybreaks
\usepackage[nottoc]{tocbibind}

\expandafter\def\expandafter\normalsize\expandafter{%
  \normalsize  
  \setlength\abovedisplayskip{2ex}
  \setlength\belowdisplayskip{2ex}
  \setlength\abovedisplayshortskip{1ex}
  \setlength\belowdisplayshortskip{1ex}
}
\usepackage{titlesec}
\titlespacing*{\section}
{0pt}{2.5ex plus 1ex minus .2ex}{3.3ex plus .2ex}
\titlespacing*{\subsection}
{0pt}{2.5ex plus 1ex minus .2ex}{3.3ex plus .2ex}

\sectionfont{\scshape\centering\fontsize{11}{14}\selectfont}
\subsectionfont{\scshape\fontsize{11}{14}\selectfont}
\usepackage{fancyhdr}

\newenvironment{keywords}{}{}
\newtheorem{thm}{Theorem}[section]
\newtheorem{cor}[thm]{Corollary}
\newtheorem{lem}[thm]{Lemma}
\newtheorem{defn}[thm]{Definition}

\newtheorem{qtn}[thm]{Question}
\newtheorem{prop}[thm]{Proposition}
\newtheorem{conjecture}[thm]{Conjecture}
\newtheorem{problem}[thm]{Problem}

\newtheorem*{claim*}{Claim}

\newcommand\longvdots[1]{\raisebox{1em}{\rotatebox{-90}{\hbox to #1 {\dotfill}}}}

\newcommand{\E}{\mathbb{E}}
\newcommand{\R}{\mathbb{R}}

\newcommand{\defeq}{\stackrel{\mathrm{def}}{=} }
\newcommand{\convd}{\stackrel{\mathrm{d}}{\longrightarrow} }
\newcommand{\di}{\mathrm{d}}

\newcommand{\be}{\begin{equation}}
\newcommand{\ee}{\end{equation}}
\newcommand{\bea}{\begin{eqnarray}}
\newcommand{\eea}{\end{eqnarray}}
\newcommand{\bes}{\begin{equation*}}
\newcommand{\ees}{\end{equation*}}
\newcommand{\beas}{\begin{eqnarray*}}
\newcommand{\eeas}{\end{eqnarray*}}

\newcommand{\fxn}{\phi}
\newcommand{\Hfxn}{\mathbf{\Phi}}
\newcommand{\hfxn}{\boldsymbol{\phi}}

\usepackage{authblk}

\def\i{{\textnormal{i}}}

\newcommand{\blambda}{\boldsymbol{\lambda}}

\newcommand\shorttitle{Joint moments orthogonal and symplectic}
\newcommand\authors{T. Assiotis, M. A. Gunes, J. P. Keating and F. Wei}

\begin{document}

\title{ \sc Joint moments of characteristic polynomials from the orthogonal and unitary symplectic groups}

\author{\small THEODOROS ASSIOTIS, MUSTAFA ALPER GUNES, JONATHAN P. KEATING and FEI WEI}

\date{}
\maketitle

\providecommand{\keywords}[1]
{ \textbf{Keywords:} #1 }


\begin{abstract}
We establish asymptotic formulae for general joint moments of characteristic polynomials and their higher-order derivatives associated with matrices drawn randomly from the groups $\mathrm{USp}(2N)$ and $\mathrm{SO}(2N)$ in the limit as $N\to\infty$. This relates the leading-order asymptotic contribution in each case to averages over the Laguerre ensemble of random matrices. We uncover an exact connection between these joint moments and a solution of the $\sigma$-Painlev\'{e} V equation, valid for finite matrix size, as well as a connection between the leading-order asymptotic term and a solution of the $\sigma$-Painlev\'{e} III$'$ equation in the limit as $N \rightarrow \infty$. These connections enable us to derive exact formulae for joint moments for finite matrix size and for the joint moments of certain random variables arising from the Bessel point process in a recursive way. As an application, we provide a positive answer to a question proposed by Altuğ et al. [Quarterly Journal of Mathematics, 65 (2014), 1111–1125].
\end{abstract}

\tableofcontents

\section{Introduction}

\subsection{Background and motivation}

In this paper we address the problems of the convergence of joint moments of the characteristic polynomials with positive real exponents, along with their higher-order derivatives, of random matrices drawn from the unitary symplectic $\mathrm{USp}(2N)$ and even orthogonal $\mathrm{SO}(2N)$ classical compact groups after suitable scalings, and of establishing connections with Painlev\'{e}  equations for both finite matrix size $N$ and in the large-$N$ limit. The general question can be phrased as follows.

\begin{problem}\label{problem:ProblemIntro}
Given a sequence of $N$-indexed compact matrix groups $(\mathrm{G}(N))_{N=1}^\infty$ of increasing dimension endowed with Haar measure $\mu_{\mathrm{Haar}}$, denote the characteristic polynomial of $\mathbf{A}\in \mathrm{G}(N)$ by $\phi_\mathbf{A}(z)=\det(\mathbf{I}-z\mathbf{A})$. Let $h_1,\dots,h_m \in \mathbb{R}_{\geq 0}$. Then, the task is to establish the asymptotics of the matrix integral
\begin{equation}\label{eq:IntegralIntro}
\int_{\mathrm{G}(N)} \prod_{k=0}^m \left|\phi_\mathbf{A}^{(k)}(1)\right|^{h_k} \mathrm{d}\mu_{\mathrm{Haar}}(\mathbf{A}) 
,
\end{equation}
as $N \to \infty$.
\end{problem}

We show that the leading order asymptotics of this integral takes the form
\begin{equation*}\label{eq:AsympIntro}
\mathfrak{c}^{(G)}_{h_0,\dots,h_m} N^{\delta^{(G)}_{h_0,\dots,h_m}}
\end{equation*}
and determine expressions for $\mathfrak{c}^{(G)}_{h_0,\dots,h_m}$ and $\delta^{(G)}_{h_0,\dots,h_m}$.  We also find exact expressions for the integrals in question in terms of solutions of certain Painlev\'{e} equations.

These considerations, particularly the choice of evaluating $\phi_{\mathbf{A}}(z)$ at $z=1$, are motivated by conjectural connections to moments of $\mathrm{L}$-functions with different symmetry types, as studied in the work of Katz-Sarnak \cite{KatzSarnak} and Keating-Snaith \cite{KeatingSnaithLfunctions}. We discuss this further at the end of Section \ref{main results}.

The analogous problem of joint moments over $\mathrm{U}(N)$, which relates to joint moments of the Riemann zeta and Hardy's function \cite{Conrey, ConreyGhosh, Ingham}, has a long history starting from the work and conjectures of Hughes \cite{Hughes} from 2001, see \cite{barhoumi2020new, keatingwei, assiotis2022joint, assiotis2024exchangeablearraysintegrablesystems, Bailey_2019, Basor_2019, conreyetal, Dehaye2008, Dehaye2010note, ForresterWittePainleveIII, forrester2022joint, simm2024moments, winn2012derivative} for part of the growing literature. The question for $\mathrm{USp}(2N)$ and $\mathrm{SO}(2N)$ has also received attention, see \cite{altugetal, andrade2024joint,  gunes2022characteristic} and also \cite{AlvarezSnaith,alvarez2023asymptotics} for the related problem of moments of the logarithmic derivative, but it has been overall much less well studied. In this paper we first settle the convergence aspect of the general question, see Theorem \ref{thm:jointmom}, and give an expression for the limit in terms of concrete finite-dimensional averages. Secondly, we establish an exact connection between these joint moments and a solution of the $\sigma$-Painlevé V equation for finite matrix size, as well as a connection between the leading-order coefficients and a solution of the $\sigma$-Painlevé III$'$ equation in the limit as $N \rightarrow \infty$. Finally, as an application, we give a positive answer to a question posed by Altuğ et al. in \cite{altugetal}. 

We remark that for the odd orthogonal group $O^{-}(2N)$, analogous results to those above can similarly be obtained by almost identical methods. To avoid excessive length, we omit a detailed discussion of this case.

Our starting point is that the matrix integral \eqref{eq:IntegralIntro} can, in the cases when $\mathrm{G}(N)=\mathrm{USp}(2N)$ or $\mathrm{G}(N)=\mathrm{SO}(2N)$, be written in terms of averages involving singular negative power sum statistics of the Jacobi unitary ensemble of random matrices with certain parameters. Various such singular statistics of classical random matrix ensembles have been studied in the literature. For previous work and connections to integrable systems, orthogonal polynomials and combinatorics, see \cite{MezzadriSimm, SingularWeightPainleve1,SingularWeightPainleve2, Hurwitz,MomentsHypergeometric,forrester2022joint,ChenDai,ChenChenFan}. For physical motivations to study such quantities in relation to quantum transport in chaotic cavities, see \cite{Wigner-Smith,BallisticChaotic}. The papers closest to our setup, which study the special case of the first negative power sum of the Jacobi ensemble, are the work of Forrester \cite{forrester2022joint}, whose impetus is connections to multivariate hypergeometric functions, and \cite{ChenDai,ChenChenFan}, which prove connections to Painlev\'{e} equations.

Our approach to proving convergence of moments in \cite{assiotis2024exchangeablearraysintegrablesystems} was based on ideas from the theory of exchangeable arrays. In some sense the role of the Jacobi ensemble in \cite{assiotis2024exchangeablearraysintegrablesystems} is played by the Cauchy ensemble \cite{ForresterWitteCauchy,ForresterBook}, also known as the Hua-Pickrell ensemble \cite{Borodin_2001}. However, these exchangeability techniques do not seem to apply to the Jacobi ensemble. Fortunately, mainly due to the positivity of eigenvalues in this setting (this is not present in the Cauchy case for $\mathrm{U}(N)$), which we exploit at several places, we can take a more direct approach which uses as input symmetric function theory computations of Forrester \cite{forrester2022joint}. Nevertheless, after we get a probabilistic expression for the limit it is possible to obtain some explicit evaluations by appealing again to exchangeability in a non-obvious way, see Proposition \ref{thm:laguerre} and Corollary \ref{someformulae}.

 Regarding establishing the connection to Painlev\'{e} equations, we adapt our previous approach from \cite{assiotis2024exchangeablearraysintegrablesystems} to the present setting. This approach is based on certain ``partition-shifted” Hankel determinants, with the main difference being that the Hankel entries considered herein are truncated variants of the multivariate generating functions appearing in \cite{assiotis2024exchangeablearraysintegrablesystems}.
A complication arising from this truncation is that the ``partition-shifted” Hankel determinants can no longer be expressed entirely in terms of the partial derivatives of the original Hankel determinants, as was the case in \cite{assiotis2024exchangeablearraysintegrablesystems}. Instead, additional terms emerge, which are related to the truncation length. To ensure that this issue does not obstruct the inductive procedure central to our analysis, we proceed in two steps. First, we derive explicit expressions for the additional terms, expressed in terms of the truncation lengths. Second, since our inductive process  proceeds with the evaluations of the partial derivatives of the ``shifted Hankel determinants" at the origin, we choose truncation lengths that depend appropriately on the orders of the partial derivatives. This ensures that, after evaluation, the contribution from the aforementioned extra terms is zero. A more precise sketch of these strategies is provided in Section~\ref{ideas}.

Finally, in a different direction, the limits of rescaled negative power sums and corresponding elementary symmetric polynomials are closely related to the Taylor coefficients of the secular functions associated to stochastic operators for $\beta$-ensembles \cite{StochasticAiry,RamirezRider,ValkoViragOperator,ValkoViragSecular,valkoli}. Here,  we are concerned with the stochastic ($\beta=2$) Bessel operator $\mathfrak{G}_{a}$ and its inverse $\mathfrak{G}_{a}^{-1}$ from \cite{RamirezRider} in which case the secular function is given by:
\begin{equation*}
\det\left(\mathbf{I}-z\mathfrak{G}_{a}^{-1}\right)=1+\sum_{k=1}^\infty \mathsf{t}_{a}^{(k)}z^k.
\end{equation*}
Expressions for the coefficients $(\mathsf{t}_{a}^{(k)})_{k=1}^\infty$ in terms of multiple integrals involving Brownian motions have been given in \cite{valkoli}. Below Proposition \ref{thm:laguerre}, we explain a connection between these coefficients and the joint moments considered in the following section.

\subsection{Asymptotics of the joint moments and the Bessel point process}\label{main results}

Let $\mathbf{A}\in \mathrm{USp}(2N)$ or $\mathbf{A}\in \mathrm{SO}(2N)$, and define its characteristic polynomial on the unit circle:
\[
\varphi_\mathbf{A}(\theta)\defeq\det\Bigl(\mathbf{I}-\mathrm{e}^{-\mathrm{i}\theta}\mathbf{A}\Bigr).
\]
Define the associated joint moments, for $h_{0}, h_1,\dots,h_m\in \mathbb{R}_{\geq 0}$ :
\begin{align*}
\mathfrak{J}_N^{\mathrm{USp}}(h_0,h_1,\ldots,h_m)&\defeq \int_{\mathrm{USp}(2N)} \prod_{k=0}^m \left|\varphi_\mathbf{A}^{(k)}(0)\right|^{h_k} \mathrm{d}\mu_{\mathrm{Haar}}(\mathbf{A}),\\
\mathfrak{J}_N^{\mathrm{SO}}(h_0,h_1,\ldots,h_m)&\defeq \int_{\mathrm{SO}(2N)} \prod_{k=0}^m \left|\varphi_\mathbf{A}^{(k)}(0)\right|^{h_k} \mathrm{d}\mu_{\mathrm{Haar}}(\mathbf{A}).
\end{align*}
In this paper, the following Bessel point process, the universal scaling limit of random matrix eigenvalues at a hard edge, will play a key role \cite{ForresterBessel,ForresterBook}.

\begin{defn}\label{def:Bessel}Let $a>-1$. We define the Bessel point process on $\mathbb{R}_+$ with parameter $a$  to be the determinantal point process  \cite{BorodinDet,JohanssonDet,ForresterBessel,ForresterBook} with correlation kernel $\mathsf{K}_a^{\mathrm{Bes}}$,
\begin{equation*}
    \mathsf{K}_a^{\mathrm{Bes}}(x,y) =\frac{\sqrt{x}J_{a+1}(\sqrt{x})J_{a}(\sqrt{y})-\sqrt{y}J_{a+1}(\sqrt{y})J_{a}(\sqrt{x})}{2(x-y)},
\end{equation*}
where $J_a$ denotes the Bessel function of the first kind with parameter $a$. It almost surely consists of a strictly increasing sequence\footnote{We note that the superscript $l$ in $\mathfrak{B}_a^l$ does not correspond to raising to the $l$-th power but rather denotes the $l$-th ordered point in the point process.} $0<\mathfrak{B}_a^1<\mathfrak{B}_a^2<\mathfrak{B}_a^3<\cdots$.
\end{defn}
\noindent Define the family of random variables $(\mathfrak{e}_m(a))_{m=0}^\infty$ by $\mathfrak{e}_0(a)\equiv 1$ and,
\begin{equation*}
    \mathfrak{e}_m(a)\defeq\sum_{1\leq i_1<\cdots<i_m}\frac{1}{\mathfrak{B}_{a}^{i_1}\cdots\mathfrak{B}_{a}^{i_m}},
\end{equation*}
which are simply elementary symmetric polynomials in the infinitely many random points $\{\mathfrak{B}_{a}^{l}\}_{l=1}^\infty$.
These are almost surely finite. 
Correspondingly, we define the power sums, for $m\geq 1$ (clearly $\mathfrak{e}_1(a)=\mathfrak{p}_1(a)$),
\beas
\mathfrak{p}_m(a)\defeq\sum_{j=1}^{\infty}\frac{1}{\left(\mathfrak{B}_{a}^{j}\right)^m}.
\eeas
In what follows, we define $\mathcal{P}_k^{(2)}$ to be the set of integer partitions of $k$ with each part equal to $1$ or $2$. For each such partition $\boldsymbol{\mu}$, we also define $l(\boldsymbol{\mu})$ to be its length, and $\theta(\boldsymbol{\mu})$ to be the number of parts that are equal to $ 2$. The following is our first main result.   

\begin{thm}
\label{thm:jointmom}
    Suppose $h_0,\ldots,h_m\in [0,\infty)$. Let $s=\sum_{i=0}^{m}h_{i}$. Then, we have convergence of joint moments:
    \begin{equation*}
        \lim_{N\to \infty}  \frac{\mathfrak{J}^{\mathrm{USp}}_N(h_0,h_1,\ldots,h_m)}{N^{\frac{s(s+1)}{2}+\sum_{k=1}^{m} kh_k}}=\frac{2^{\frac{s^2}{2}}G(1+s)\sqrt{\Gamma(1+s)}}{\sqrt{G(1+2s)\Gamma(1+2s)}} \mathbb{E}\left[\prod_{k=1}^m \mathfrak{R}_{k}\bigg(s+\frac{1}{2}\bigg)^{h_k}\right]<\infty
    \end{equation*}
    where $G$ is the Barnes $G$-function and
    \begin{equation*}
        \mathfrak{R}_k(a)\defeq \sum_{\boldsymbol{\mu} \in \mathcal{P}_k^{(2)}}\frac{2^{-2\theta(\boldsymbol{\mu})}k!}{(l(\boldsymbol{\mu})-\theta(\boldsymbol{\mu}))!} \mathfrak{e}_{\theta(\boldsymbol{\mu})}(a).
    \end{equation*}
    Moreover, whenever $h_0>\frac{1}{2}$, we have
      \begin{equation}
      \label{jointmomorth}
        \lim_{N\to \infty}  \frac{\mathfrak{J}_N^{\mathrm{SO}}(h_0,h_1,\ldots,h_m)}{N^{\frac{s(s-1)}{2}+\sum_{k=1}^{m} kh_k}}=\frac{2^{\frac{s^2}{2}}G(1+s)\sqrt{\Gamma(1+2s)}}{\sqrt{G(1+2s)\Gamma(1+s)}} \mathbb{E}\left[\prod_{k=1}^m \mathfrak{R}_{k}\bigg(s-\frac{1}{2}\bigg)^{h_k}\right]<\infty.
    \end{equation}
    \end{thm}

Theorem \ref{thm:jointmom} completely solves Problem  \ref{problem:ProblemIntro} for $\mathrm{USp}(2N)$, and essentially solves it, modulo the technical restriction $h_0>1/2$ for $\mathrm{SO}(2N)$. In both cases, this is achieved by determining completely explicitly the leading order $\delta^{(G)}_{h_0,\dots,h_m}$
and giving a natural probabilistic expression for $\mathfrak{c}^{(G)}_{h_0,\dots,h_m}$. We expect that the result should hold in full generality beyond the technical restriction, and in fact, in the special case where $h_1,\ldots,h_m \in \mathbb{N}$ the restriction can already be removed using the arguments developed in this paper with minimal additional effort, see Corollary \ref{cor:integerhj}.\\
\indent We now turn to the task of computing the moments of $\mathfrak{R}_k$ that define $\mathfrak{c}_{h_0,\dots,h_m}^{(G)}$ in the two cases $G(N)=\mathrm{USp}(2N)$ and $G(N)=\mathrm{SO}(2N)$.
It is obvious that the integer joint moments of the random variables $\mathfrak{R}_k$ can be recovered by evaluating the integer joint moments of  the $\mathfrak{e}_k$'s. In turn, as we show in Proposition \ref{thm:laguerre}, this is achievable by representing such averages as a linear combination of averages over finite-dimensional matrix eigenvalue ensembles which are explicitly computable. To this end, we define the Weyl chamber
    \begin{equation*}
    \mathbb{W}_N\defeq \{\mathbf{x}=(x_1,\dots,x_N)\in \R^N \, \mid \, x_1\leq x_2\leq \ldots \leq x_N\}
\end{equation*}
    and introduce the Laguerre probability measures $\nu_N^{(a)}$ on $\mathbb{W}_N$, with $a>-1$:
    \begin{equation*}
        \nu_N^{(a)}(\mathrm{d}\mathbf{x})= \frac{1}{\mathcal{Z}_{N}^{(a)}} \, \prod_{j=1}^N x_j^a \mathrm{e}^{-x_j} \mathbf{1}_{x_j\in \mathbb{R}_+}\prod_{1\leq i<j\leq N} (x_i-x_j)^2  \mathrm{d}\mathbf{x},
    \end{equation*}
  where $\mathcal{Z}_N^{(a)}$ is an explicit normalization constant, see \cite{ForresterBook}.   We denote by $\E_N^{(a)}$ the averages taken with respect to $\nu_N^{(a)}$. 
  
  Using ideas developed in the authors' previous paper \cite{assiotis2024exchangeablearraysintegrablesystems}, we show that one can in fact represent integer joint moments of the $\mathfrak{e}_k$ in terms of finite-dimensional averages of the Laguerre ensemble. The fact that this is possible is non-obvious. Introduce the elementary symmetric polynomials $\mathrm{e}_k$ in the variables $(x_1,\dots,x_N)$ that are defined for each $1\leq k\leq N$ as:
  \begin{equation}
     \mathrm{e}_{k} \big(x_1,\ldots, x_N\big) \defeq \sum_{1\leq i_1<\cdots<i_k\leq N} x_{i_1}\cdots x_{i_k}.
  \end{equation}
\noindent Our result on evaluating the leading order coefficient reads as follows.
\begin{prop}
\label{thm:laguerre}
Suppose $m\in \mathbb{N}$, $a>0$, and $h_1,\ldots, h_m\in \mathbb{N}\cup \{0\}$ with $h_{m}\neq 0$ and $\sum_k h_k<a+1$, and define $Q=\sum_k kh_k$. Then, we have
\begin{equation*}
\mathbb{E}\left[\prod_{k=1}^m \left(\mathfrak{e}_{k}(a)\right)^{h_k}\right]=\frac{(1!)^{h_1}\cdots(m!)^{h_m}}{Q!}\sum_{j=m}^{Q} (-1)^{Q+j} {Q \choose j} \mathbb{E}_{j}^{(a)}\left[\prod_{k=1}^m \left(\mathrm{e}_{k} \left(\frac{1}{\mathsf{x}_1^{(j)}},\frac{1}{\mathsf{x}_2^{(j)}},\ldots, \frac{1}{\mathsf{x}_j^{(j)}}\right)\right)^{h_k}\right].
\end{equation*}
\end{prop}



As stated earlier, the $\mathfrak{e}_{k}(a)$ are simply elementary symmetric polynomials in the (infinitely many) variables $\{\mathfrak{B}_{a}^i\}_{i\geq 1}$.  These are basically the Taylor coefficients of the random entire function $\mathfrak{f}_{a}$ obtained in the work of Li and Valk\'{o} \cite{valkoli} as the limit of characteristic polynomials of the real orthogonal ($\beta=2$) ensemble or as the secular function \cite{ValkoViragSecular,valkoli} of the stochastic ($\beta=2$) Bessel operator $\mathfrak{G}_{a}^{-1}$, see \cite{RamirezRider}, also \cite{RandomAnalytic}: 
\begin{equation}
\mathfrak{f}_{a}(z)=\det\left(\mathbf{I}-z\mathfrak{G}_{a}^{-1}\right)=\prod_{i=1}^\infty\left(1-z/\mathfrak{B}_{a}^i\right)=1+\sum_{k=1}^{\infty}(-1)^k \mathfrak{e}_k(a) z^k.
\end{equation}
In their work Li and Valk\'{o} \cite{valkoli} give nice expressions for the $\mathfrak{e}_k(a)$ in terms of multiple integrals of exponential Brownian motions. Joint moments of these multiple integrals are then given by the formula in Proposition \ref{thm:laguerre}. It would be interesting if there is an alternative derivation of this formula using Brownian motions.\\
\indent Finally, we use Theorem \ref{thm:jointmom} to formulate conjectures for the mean values of high-order derivatives of $L$-functions at the center of the critical strip, in particular conjecturing that the leading term has a probabilistic representation.  
This direction was initiated by the results of Katz and Sarnak \cite{KatzSarnak}, who introduced the concept of symmetry types for families of $L$-functions and suggested that the zero distribution of such a family should be modeled by the eigenvalue distribution of the corresponding classical groups. These classical groups are determined by the symmetries inherent in each family of $L$-functions. Later, Keating and Snaith (see \cite{keatingsnaith,KeatingSnaithLfunctions}) investigated the moments of the characteristic polynomials of random matrices from the groups $\mathrm{U}(N)$, $\mathrm{USp}(2N)$, and $\mathrm{SO}(2N)$. They conjectured that the leading-order asymptotics of these moments coincide with the leading terms in the corresponding moments of families of $L$-functions at the central point $s=1/2$, corresponding respectively to the unitary, unitary symplectic, and orthogonal symmetry types. Furthermore, as shown in~\cite{keatingsnaith} and~\cite{KeatingSnaithLfunctions}, these conjectures are consistent with the known results for $L$-functions in number theory.

Following these heuristics and using the results of this paper, we can give conjectures regarding the asymptotic mean value distribution of derivatives of arbitrary order of families of $L$-functions at the central point $s = 1/2$. In the following, we present an example conjecture in the even orthogonal case. Theorem \ref{thm:jointmom}, together with the recipe for formulating conjectures about the moments of any primitive family of $L$-functions, enables similar conjectures to be made for other families of $L$-functions. We refer the readers to \cite[Sections 1.3]{CFKRS2} and \cite[Section 4.3]{CFKRS1} for some examples without derivatives. Consider the mean values of 
\beas
L_f(s) = \sum_{n=1}^{\infty} \lambda_f(n) n^{-s}, 
\eeas
at the critical point $s = \frac{1}{2}$, averaged over $f \in H_m^*(N)$. Here, $ H_m^*(N)$ denotes the set of primitive Hecke eigenforms of
weight $m$ relative to the subgroup $\Gamma_0(N)$,
and $\lambda_f(n)$ are the normalized Fourier coefficients of the eigenform with $\lambda_f(1)=1$. For simplicity, we restrict our attention to the case $m = 2$ and $ N = q $, where $ q$ is a prime number. The low-lying zeros of this family near the critical point exhibit orthogonal symmetry, and the functional equation for the $L$-function can be written as,
\[
\Lambda_{f}(s) \defeq \Bigg(\frac{\sqrt{q}}{2\pi}\Bigg)^s \Gamma\left(s + \frac{1}{2}\right) L_{f}(s) = \varepsilon_f \Lambda_{f}(1 - s),
\]
where $\varepsilon_f = -\sqrt{q} \lambda_f(q) = \pm 1 $. We say that $f$ is even (respectively, odd) if and only if $ \varepsilon_f = +1 $(respectively, $\varepsilon_f = -1 )$. We define the harmonic average as
\[
\sum_{f} \raisebox{0.7ex}{\hspace{-0.3em}$^{h}$} 
A_f \defeq \sum_{f \in H_2^*(q)} \frac{A_f}{(f, f)},
\]
where $(f, g)$ is the Petersson inner product on the space $\Gamma_0(q) \backslash \mathbb{H}$.

\begin{conjecture} 
Let $k,s\geq 0$ be non-negative integers. Then, as $q\to \infty$,
\begin{equation*}
\sum_{ f~even} \raisebox{0.7ex}{\hspace{-0.5em}$^{h}$}
\left(L^{(k)}_{f}\Big(\frac{1}{2}\Big)\right)^{s}\sim a_{s}\big(\log(q^{1/2})\Big)^{^{s(s-1)/2+sk} } g(k;s),
\end{equation*}
where
\beas
g(k;s)=\frac{2^{\frac{s^2}{2}-1}G(1+s)\sqrt{\Gamma(1+2s)}}{\sqrt{G(1+2s)\Gamma(1+s)}} \sum_{l_{1},\dots,l_{s}=0}^{k}\left(\prod_{j=1}^{s}\binom{k}{l_{j}}\right)(-2)^{sk-\sum_{j=1}^{s}l_{j}}\mathbb{E}\left[\prod_{j=1}^s \mathfrak{R}_{l_{j}}\bigg(s-\frac{1}{2}\bigg)\right],
\eeas
and 
\beas
a_s = \prod_{p \nmid q} \left( 1 - \frac{1}{p} \right)^{\frac{s(s-1)}{2}} \frac{2}{\pi} \int_0^\pi \sin^2 \theta
\Bigg(
\frac{
\mathrm{e}^{\i \theta} \Big(1 - \frac{\mathrm{e}^{\i \theta}}{\sqrt{p}}\Big)^{-1} - \mathrm{e}^{-\i \theta} \Big(1 - \frac{\mathrm{e}^{-\i \theta}}{\sqrt{p}}\Big)^{-1}
}{
\mathrm{e}^{\i \theta} - \mathrm{e}^{-\i \theta}
}
\Bigg)^s \mathrm{d}\theta.
\eeas
\end{conjecture} 

It is possible to state analogous conjectures for general joint moments, but for brevity we only state the simplest case for moments of a single derivative of arbitrary order.

 \subsection{Characteristic polynomials and Painlevé equations}
 
 In various works in literature, it was shown that solutions of certain Painlevé equations are closely related to the limiting objects appearing in the problem of joint moments over various classical groups, see for example \cite{altugetal, ABGS, Basor_2019, ForresterWittePainleve1, ForresterWittePainleve2, ForresterWittePainleveIII, gunes2022characteristic}. 
The work of most relevance to the setting herein is \cite{gunes2022characteristic} which studies the special case of joint moments with the second derivative for $\mathrm{USp}(2N)$. It is proven there that a random variable appearing in those asymptotics admits a representation in terms of a solution to a certain Painlev\'{e} equation. This random variable was not explicitly identified in \cite{gunes2022characteristic} and denoted therein by $M(a,b)$. The connection to Painlev\'{e} equations then reads as follows. The function $\tau_{a,b}$ given by:
\begin{equation*}
 \tau_{a,b}(t)=t \frac{\di}{\di t} \log \mathbb{E}\left[\mathrm{e}^{-tM(a,b)}\right]
\end{equation*}
solves the $\sigma$-Painlev\'{e} III$'$ equation
    \begin{equation}\label{eq:Painleve}
        \left(t \frac{\di^2 }{\di t^2}\tau_{a,b}(t)\right)^2+4\left(\frac{\di}{\di t}\tau_{a,b}(t)\right)^2\left(t \frac{\di}{\di t}\tau_{a,b}(t)-\tau_{a,b}(t)\right)-\left(a\frac{\di}{\di t}\tau_{a,b}(t)+1\right)^2=0,
    \end{equation}
for any $t\in \mathbb{R}_{\geq 0}\setminus \{0\}$, along with the initial conditions:
\begin{equation*}
\begin{cases}
      \tau_{a,b}(0)=0, & \text{for}  \ a>0,\label{eq:boundary1}\\
      & \\
   \frac{\mathrm{d}}{\mathrm{d}t}\tau_{a,b}(t)\Bigr\rvert_{t = 0}=-\frac{1}{a}, & \text{for} \ a>1. 
\end{cases}
\end{equation*}
Here, we note that for the derivatives to exist at $t=0$, one requires the finiteness of certain moments, which is then guaranteed by imposing further conditions on the parameter $a$. Observe that, the equation \eqref{eq:Painleve} does not depend on $b$ and had in fact appeared previously in work on $\mathfrak{e}_1(a)$, see \cite{ABGS}. Based on this, it was predicted in \cite{gunes2022characteristic}, see Remark 1.2 therein, that $M(a,b)$ is in fact equal to $\mathfrak{e}_1(a)$. Our working below identifies, in distribution, $$M(a,b)=\mathfrak{e}_1(a)$$ and confirms this prediction. Moreover, by virtue of this identification and the results of \cite{ABGS} we obtain an alternative proof of \eqref{eq:Painleve}.

\subsubsection{Connections between higher joint moments and Painlev\'{e} equations}
Given the above, it is natural to ask if there is a connection between higher joint moments of the $(\mathfrak{e}_1(a),\dots,\mathfrak{e}_k(a))$, or equivalently, ($\mathfrak{p}_1(a),\mathfrak{p}_2(a)\dots,\mathfrak{p}_k(a))$, and integrable systems. The equivalence here refers to the fact that one vector of random variables can be expressed as a linear combination of power sums of the other, which follows from Newton's identities and Lemma \ref{lem:skorokhod} below. The following result answers this natural question positively.


\begin{thm}\label{mainpainlevethm}
Let $k\geq 2$, $n_{2},\ldots,n_{k}\in \mathbb{N}\cup\{0\}$ and $\sum_{j=2}^{k} n_j >0$. Suppose that $a\in \mathbb{R}$ with $a>\sum_{\ell=2}^{k}\ell n_{\ell}-1$. Then, for $t_{1}\geq 0$,
\begin{align}\label{formula1}
 \mathbb{E}\left[\mathrm{e}^{- t_{1}\mathfrak{e}_1(a)}\prod_{j=2}^{k}\mathfrak{p}_j(a)^{n_{j}}\right]=\frac{1}{t_{1}^{\sum_{\ell=2}^{k}\ell n_{\ell}-1}}
\sum_{m=0}^{\sum_{\ell=2}^{k}(\ell-1) n_{\ell}} t_1^{m-1} \mathcal{P}_{m}^{(a)}(n_{2},\ldots,n_{k};t_{1})\frac{\mathrm{d}^{m}}{\mathrm{d}t_{1}^{m}}\mathbb{E}\Big[\mathrm{e}^{-t_{1}\mathfrak{e}_{1}(a)}\Big],   
\end{align}
where $\mathcal{P}_{m}^{(a)}(n_{2},\ldots,n_{k};t_{1})$ are polynomials in $t_{1}$ whose coefficients are polynomials in $a$ of degree at most $\sum_{\ell=2}^{k}(\ell-1)n_{\ell}$. Moreover, the degree of $t_1^{m-1}\mathcal{P}_{m}^{(a)}(n_{2},\ldots,n_{k};t_{1})$ in $t_1$ is at most $\sum_{\ell=2}^{k}(\ell-1)n_{\ell}-1$.

\end{thm}
In the proof of the above theorem, we provide a recursive method to obtain explicit formulae for the polynomials $\mathcal{P}_{m}^{(a)}(n_{2},\ldots,n_{k};t_1)$. Furthermore, as mentioned earlier, the function $\mathbb{E}\left[\mathrm{e}^{-t_{1}\mathfrak{e}_1(a)}\right]$ can be written in terms of a solution to the $\sigma$-Painlev\'{e} III$'$ equation (\ref{eq:Painleve}). This observation then enables an asymptotic analysis of the Taylor expansion of the solution at $t_{1}=0$, which, when combined with Theorem \ref{mainpainlevethm}, allows for the recursive computation of all integer joint moments of $\mathfrak{e}_{j},j=1,\ldots,k$, also allowing one to express them efficiently as rational functions of $a$.
We give a few examples of these formulae in the corollary below. These expressions, in spite of their brevity, appear to be new, and at least to our knowledge, are not mere consequences of the probabilistic representations the random variables $\{\mathfrak{p}_m(a)\}_{m=1}^\infty$ admit.
\begin{cor}\label{some examples}
Let $t\geq 0$. For $a>1$,
\beas
\E[\mathrm{e}^{-t\mathfrak{e}_1(a)} \mathfrak{p}_2(a)]
= \frac{a}{t} \frac{\mathrm{d}}{\mathrm{d}t} \E[\mathrm{e}^{-t\mathfrak{e}_1(a)}]
+ \frac{1}{t} \E[\mathrm{e}^{-t\mathfrak{e}_1(a)}]
\eeas
and for $a>3$,
\beas
\E[\mathrm{e}^{-t\mathfrak{e}_1(a)} \mathfrak{p}_2(a)^2]
= \frac{a^2+2}{t^2} \frac{\mathrm{d}^2}{\mathrm{d}t^2} \E[\mathrm{e}^{-t\mathfrak{e}_1(a)}]
+ \frac{-3a^2+2at}{t^3} \frac{\mathrm{d}}{\mathrm{d}t} \E[\mathrm{e}^{-t\mathfrak{e}_1(a)}] + \frac{t-3a}{t^3}  \E[\mathrm{e}^{-t\mathfrak{e}_1(a)} ].
\eeas
\end{cor}
The formulae above may appear singular at $t= 0$, however 
the moments of $\mathfrak{e}_1(a)=\mathfrak{p}_1(a)$
and $\mathfrak{p}_2(a)$ 
satisfy relations that ensure cancellations of the singular terms in the denominator.

As an application of Corollary \ref{some examples}, we present explicit formulae for the leading terms of certain joint moments of the second and fourth derivatives over the groups $\mathrm{USp}(2N)$ and $\mathrm{SO}(2N)$.
\begin{cor}\label{someformulae}
Whenever $s>1/2$, we have
\begin{align*}
 \lim_{N\rightarrow \infty}\frac{\int_{\mathrm{USp}(2N)} \left|\varphi_{\mathbf{A}}^{(2)}(0)\right|^2 \left|\varphi_{\mathbf{A}}(0)\right|^{s-2} \mathrm{d}\mu_{\mathrm{Haar}}(\mathbf{A})}{N^{\frac{s(s+1)}{2}+4}\int_{\mathrm{USp}(2N)} \left(\varphi_{\mathbf{A}}(0)\right)^{s} \mathrm{d}\mu_{\mathrm{Haar}}(\mathbf{A})}=&4{\frac {2\,{s}^{3}+5\,{s}^{2}+2\,s-2}{ \left( 2\,s+1 \right) 
 \left( 2\,s-1 \right)  \left( 2\,s+3 \right) }}
\end{align*}
and 
\beas
&&\lim_{N\rightarrow \infty}\frac{\int_{\mathrm{USp}(2N)} \left|\varphi_{\mathbf{A}}^{(4)}(0)\right|^2 \left|\varphi_{\mathbf{A}}(0)\right|^{s-2} \mathrm{d}\mu_{\mathrm{Haar}}(\mathbf{A})}{N^{\frac{s(s+1)}{2}+8}\int_{\mathrm{USp}(2N)} \left(\varphi_{\mathbf{A}}(0)\right)^{s} \mathrm{d}\mu_{\mathrm{Haar}}({\mathbf{A}})}\\
&&=16{\frac {4\,{s}^{6}+56\,{s}^{5}+307\,{s}^{4}+806\,{s}^{3}+967\,{s}^
{2}+380\,s+36}{ \left( 2\,s+1 \right) ^{2} \left( 2\,s+7 \right) 
 \left( 2\,s+5 \right)  \left( 2\,s+3 \right)  \left( 2\,s-1 \right) }
}.
\eeas
Whenever $s>3/2$,
\beas
\lim_{N\rightarrow \infty}\frac{\int_{\mathrm{SO}(2N)} \left|\varphi_{\mathbf{A}}^{(2)}(0)\right|^2 \left|\varphi_{\mathbf{A}}(0)\right|^{s-2} \mathrm{d}\mu_{\mathrm{Haar}}(\mathbf{A})}{N^{\frac{s(s-1)}{2}+4}\int_{\mathrm{SO}(2N)} \left(\varphi_\mathbf{A}(0)\right)^s\mathrm{d}\mu_{\mathrm{Haar}}(\mathbf{A})}=4\,{\frac {2\,{s}^{3}-{s}^{2}-2\,s-1}{ \left( 2\,s-1 \right)  \left( 2
\,s-3 \right)  \left( 2\,s+1 \right) }}
\eeas
and
\beas
&&\lim_{N\rightarrow \infty}\frac{\int_{\mathrm{SO}(2N)} \left|\varphi_{\mathbf{A}}^{(4)}(0)\right|^2 \left|\varphi_\mathbf{A}(0)\right|^{s-2} \mathrm{d}\mu_{\mathrm{Haar}}(\mathbf{A})}{N^{\frac{s(s-1)}{2}+8}\int_{\mathrm{SO}(2N)} \left(\varphi_\mathbf{A}(0)\right)^s \mathrm{d}\mu_{\mathrm{Haar}}(\mathbf{A})}\\
&&=16{\frac {4\,{s}^{6}+32\,{s}^{5}+87\,{s}^{4}+58\,{s}^{3}-109\,{s}^{2
}-108\,s+72}{ \left( 2\,s-1 \right) ^{2} \left( 2\,s+5 \right) 
 \left( 2\,s+3 \right)  \left( 2\,s+1 \right)  \left( 2\,s-3 \right) }
}.
\eeas
\end{cor}
The explicit formulae for moments not involving any derivatives, which appear in the denominator of the left-hand side of the identities above are given in \cite{KeatingSnaithLfunctions}. In particular, modulo this factor, we see from Corollary \ref{someformulae} that the leading coefficients of the joint moments of the derivatives are rational functions of the exponent $s$.

Theorem \ref{mainpainlevethm} is proved by taking the large $N$ limit of the corresponding expression that holds at the finite-$N$ level. This, on the other hand, requires establishing a similar relation for the higher order joint moments for finite matrix size $N$ and $\sigma$-Painlev\'e V transcendents, which is what we state next. In order to do so in a more compact form, let us define the random variables (we sometimes drop $a,b$ from the notation if it is clear from context)
\begin{equation*}
    \mathfrak{p}_{m,N}^{(a,b)}\defeq \sum_{j=1}^N \frac{1}{\big(\mathsf{x}_j^{(N)}\big)^m}
\end{equation*}
where $(\mathsf{x}_1^{(N)},\ldots, \mathsf{x}_N^{(N)})$ are distributed according to the Jacobi ensemble $\mu_{N}^{(a,b)}$, see Definition \ref{def:jacobi}. The following result for singular statistics of the much-studied Jacobi ensemble $\mu_{N}^{(a,b)}$ is new and may be of independent interest.
\begin{thm}\label{structureforgeneralfinitesizeforfiniteN}
Let $k\geq 2$, $n_{2},\ldots,n_{k}$ be non-negative integers and $\sum_{j=2}^{k} n_j >0$. Let $b>-1$. Let $a>\sum_{q=2}^{k}qn_{q}-1$. Then, for all $t_{1}\geq 0$,
\begin{align}\label{general structure1}
    \frac{1}{N^{2\sum_{q=2}^{k}n_{q}}}\mathbb{E}_{N}^{(a,b)}\Bigg[\mathrm{e}^{-\frac{t_1}{N^2}\mathfrak{p}_{1,N}}\prod_{q=2}^{k}(\mathfrak{p}_{q,N})^{n_{q}}\Bigg]=\frac{1}{t_{1}^{\sum_{q=2}^{k}qn_{q}-1}}
\sum_{m=0}^{\sum_{q=2}^{k}(q-1)n_{q}}t_{1}^{m-1}P_{m}^{(a,b)}(N;t_{1})\frac{\mathrm{d}^{m}}{\mathrm{d}t_{1}^{m}}\mathbb{E}_{N}^{(a,b)}\Bigg[\mathrm{e}^{-\frac{t_1}{N^2}\mathfrak{p}_{1,N}}\Bigg],
\end{align}
where $P_{m}^{(a,b)}(N;t_{1})$ are polynomials of $t_{1}$. 
Moreover, $t_1^{m-1} P_{m}^{(a,b)}(N;t_{1})$ are polynomials of $t_{1}$ of degree at most $\sum_{q=2}^{k}qn_{q}-1$, whose coefficients are polynomials in $N, a, b$, with degrees in $N$ and in $a,b$ no more than $\sum_{q=2}^{k}2(q-1)n_{q}$ and $\sum_{q=2}^{k}(q-1)n_{q}$, respectively.

\end{thm}
\noindent 
For reasons that will become clear in Section \ref{section3}, joint moments over $\mathrm{SO}(2N)$ and $\mathrm{USp}(2N)$ admit representations in terms of joint moments of the random variables $\{\mathfrak{p}_{m,N}^{(a,b)}\}_{m=1}^\infty$. Moreover, it is known that 
(see e.g., \cite[Section 3]{ForresterWittePainleve2}) the function
\beas
\sigma_N(t) = t \frac{\mathrm{d}}{\mathrm{d}t} \log \E_{N}^{(a,b)}\left[ \mathrm{e}^{-t \sum_{j=1}^N \frac{1}{\mathsf{x}_j^{(N)}  } } \right] 
\eeas
satisfies the following $\sigma$-Painlev\'e-V equation:
\bea
\left( t\frac{\mathrm{d}^2 \sigma_N}{\mathrm{d}t^2} \right)^2
&=& -4 \left( \frac{\mathrm{d}\sigma_N}{\mathrm{d}t} \right)^3 t + \big(a^2+2at+4bt+t^2+4\sigma_N\big) \left( \frac{\mathrm{d} \sigma_N}{\mathrm{d}t} \right)^2\nonumber \\
&& +\big(-2(a+2b+t) \sigma_N +2N(a+t)(a+b+N) \big) \frac{\mathrm{d} \sigma_N}{\mathrm{d}t} \nonumber\\
&& + \big(\sigma_N - N (a+b+N)\big)^2.\label{equationforfiniteN}
\eea
Combining these, we see that joint moments at the finite-$N$ level also admit expressions in terms of Painlevé transcendents.
Moreover, as one would expect at this point, taking the $N\to \infty$ limit of the above equation upon rescaling $t \rightarrow \frac{t}{N^2}$, one precisely obtains (\ref{eq:Painleve}). 

Another consequence of Theorem \ref{structureforgeneralfinitesizeforfiniteN} is that, when combined with the probabilistic techniques used in the proof Theorem \ref{thm:jointmom}, it allows us to remove the restriction $h_0>\frac{1}{2}$ in the orthogonal case, whenever $h_1,\ldots,h_m\in \mathbb{N}$.
\begin{cor}
\label{cor:integerhj}
    The convergence asserted in Theorem \ref{thm:jointmom} for $\mathrm{SO}(2N)$ holds for any $h_0\in [0,\infty)$, provided $h_1,\ldots,h_m\in \mathbb{N}\cup\{0\}.$
\end{cor}

It is also worth noting that Theorem \ref{structureforgeneralfinitesizeforfiniteN} provides a recursive method to explicitly compute all integer joint moments of derivatives of any order for finite matrix size $N$, which also has number theoretic consequences due to the work of Katz and Sarnak \cite{KatzSarnak}. In particular, they prove rigorously that  joint moments of $L$-function families corresponding to curves over $\mathbb{F}_q$, in the large $q$ regime, can be modeled by the corresponding averages over the classical groups. Consequently, Theorem \ref{structureforgeneralfinitesizeforfiniteN} provides yet another example of objects in integrable systems appearing in number theoretic considerations.

\subsubsection{A question of  Altuğ, Bettin, Petrow,  Rishikesh, and Whitehead}

Finally, as an application of the above results concerning  the connections between the joint moments of the random variables $\mathfrak{e}_k(a)$ and integrable systems, we provide an affirmative answer to a question of Altuğ, Bettin, Petrow,  Rishikesh, and Whitehead, as addressed in the paragraph below \cite[formula (5)]{altugetal}.

In the following, we give the necessary background and explicitly formulate the question, which one can also ask for various other compact matrix groups $(G(N))_{N=1}^\infty$. In \cite{altugetal}, the authors investigated the asymptotic formulae for the moments of the derivatives of characteristic polynomials over the groups $\mathrm{USp}(2N)$, $\mathrm{SO}(2N)$ and $\mathrm{O}^{-}(2N)$. To be more precise,  they showed that for any $s \in \mathbb{N}$, the leading-order coefficients of
\bea\label{previousone}
M_{s}(G(N),k)\defeq \int_{G(N)}\left(\phi_\mathbf{A}^{(k)}(1)\right)^{s} \mathrm{d}\mu_{\mathrm{Haar}}(\mathbf{A}),
\eea
admit expressions in terms of the derivatives of 
\bea\label{the previous determinant}
\det_{1\leq i,j\leq s} \Big(f_{2i-j+l}(t)\Big)
\eea
with respect to $t$ at zero for the following cases: for $l=0$, when $G(N)=\mathrm{USp}(2N), k=2$ and $G(N)=\mathrm{O^{-}}(2N), k=3$; and for $l=-1$ when $G(N)=\mathrm{SO}(2N), k=2$. 
Here, for $n\in \mathbb{Z}$,
\beas
f_{n}(t) &=& \frac{1}{2\pi \i} \oint_{|w|=1} \frac{\mathrm{e}^{w+\frac{t}{w^2}}}{w^{n+1}} \mathrm{d}w = \frac{1}{\Gamma(n+1)} {}_0F_{2}\left(; \frac{n}{2}+1, \frac{n+1}{2}; \frac{t}{4}\right)
\eeas
where ${}_0F_{2}$ is the hypergeometric function. These results were generalized by Andrade and Best in \cite{andrade2024joint} to arbitrary $k$ (or, more generally, for joint moments of $k_{1}$-th and $k_{2}$-th derivatives), in terms of the derivatives (with respect to $t$ at zero) of 
\bea\label{the previous determinant2}
\det_{1\leq i,j\leq s} \Big(f_{2i-j-1+l_{i}}(t)\Big)
\eea
for certain non-negative integers $l_{i}$. The approaches used in \cite{altugetal} and \cite{andrade2024joint} are based on the shifted moment formulae of Conrey, Farmer, Keating, Rubinstein, and Snaith in \cite{CFKRS1}, as well as techniques similar to those in \cite{conreyetal} by Conrey, Rubinstein, and Snaith. As explained in \cite{altugetal}, the authors were motivated by the result in \cite{conreyetal}, which expressed the leading coefficients of $M_{s}(\mathrm{U}(N),1)$ in terms of the derivatives (with respect to $t$ at zero) of 
\bea\label{inthepreviousCUE}
\det_{1\leq i,j\leq s} \Big(I_{i+j-1}(2\sqrt{t})\Big),
\eea
where $I_{n}(x)$ is the modified Bessel function of the first kind. In  \cite{keating-fei}, the third and fourth-named authors of this paper proved that, for any $k \geq 2$, $M_{s}(\mathrm{U}(N),k)$ can also be expressed in terms of the derivatives of (\ref{inthepreviousCUE}) evaluated at $t = 0$. The Hankel determinant (\ref{inthepreviousCUE}) was shown by Forrester and Witte \cite{ForresterWittePainleveIII}  to be an Okamoto $\tau$-function associated with a $\sigma$-Painlev\'{e} III$'$. This connection was later established alternatively by Basor et al. \cite{Basor_2019} using the Riemann–Hilbert method. In light of all these developments, it becomes natural to consider the following question that was initially posed in \cite{altugetal}.

\begin{qtn}[Altuğ et al. \cite{altugetal}]
\label{questionrelatedtoPainleveequation}
Does there exist a differential equation that plays the analogous role for symplectic and orthogonal types as the $\sigma$-Painlev\'{e} III$'$  equation is known to do for unitary symmetry?
\end{qtn}
As a direct consequence of our results we provide a positive answer to this question by confirming that a suitable choice of differential equation is the $\sigma$-Painlevé III$'$ equation itself, as given in (\ref{eq:Painleve}).
More specifically, our results confirm that, in the case considered here, 
$
\mathbb{E}\!\left[\mathrm{e}^{-t\mathfrak{e}_1(a)}\right]$ plays a role analogous to that of (\ref{inthepreviousCUE}) in the unitary case, for the following two reasons. On the one hand, 
\(\mathbb{E}\!\left[\mathrm{e}^{-t\mathfrak{e}_1(a)}\right]\) is the \(\tau\)-function associated with the \(\sigma\)-Painlev\'{e} III$'$ equation(\ref{eq:Painleve}). On the other hand, in what follows we give a precise statement: for any integer \(k \geq 1\) and for \(G(N) = \mathrm{USp}(2N)\) or \(\mathrm{SO}(2N)\), the leading coefficients of the moments \(M_{s}(G(N),k)\) can be expressed in terms of the derivatives at \(t=0\) of 
$
\mathbb{E}\!\left[\mathrm{e}^{-t\mathfrak{e}_1(a)}\right].
$
By a similar argument, the same conclusion also holds for \(G(N) = \mathrm{O}^{-}(2N)\). For the sake of brevity, we omit the details.
\begin{cor}\label{answeraquestion}
Let $s,k\geq 1$ be integers. Let $\mathcal{P}_k^{(2)}$, $l(\boldsymbol{\mu})$ and $\theta(\boldsymbol{\mu})$ be as defined in Theorem \ref{thm:jointmom}. 
Then 
\beas
\lim_{N\to \infty}  \frac{M_{s}(\mathrm{USp}(2N),k)}{N^{\frac{s^2+s}{2}+ks}}=\frac{2^{\frac{s^2}{2}}G(1+s)\sqrt{\Gamma(1+s)}}{\sqrt{G(1+2s)\Gamma(1+2s)}}\sum_{\boldsymbol{\mu_{1}},\ldots, \boldsymbol{\mu_{s}}\in \mathcal{P}_k^{(2)}}\frac{(k!)^{s}2^{-2\sum_{j=1}^{s}\theta(\boldsymbol{\mu}_{j})}}{\prod_{j=1}^{s}(l(\boldsymbol{\mu}_{j})-\theta(\boldsymbol{\mu}_{j}))!}\prod_{j=1}^{s}b_{\theta(\boldsymbol{\mu}_{j})}\left(s+\frac{1}{2}\right)
\eeas
and 
\beas
\lim_{N\to \infty}  \frac{M_{s}(\mathrm{SO}(2N),k)}{N^{\frac{s^2-s}{2}+ks}}=\frac{2^{\frac{s^2}{2}}G(1+s)\sqrt{\Gamma(1+2s)}}{\sqrt{G(1+2s)\Gamma(1+s)}} \sum_{\boldsymbol{\mu_{1}},\ldots, \boldsymbol{\mu_{s}}\in \mathcal{P}_k^{(2)}}\frac{(k!)^{s}2^{-2\sum_{j=1}^{s}\theta(\boldsymbol{\mu}_{j})}}{\prod_{j=1}^{s}(l(\boldsymbol{\mu}_{j})-\theta(\boldsymbol{\mu}_{j}))!}\prod_{j=1}^{s}b_{\theta(\boldsymbol{\mu}_{j})}\left(s-\frac{1}{2}\right),
\eeas
where $b_{0}(a)=1$ and for $m\geq1$, $b_{m}(a)$
is given by the following formula:
\beas
&&\Bigg((-1)^{m}\frac{1}{m!}\frac{\mathrm{d}^{m}}{\mathrm{d}t^{m}}\mathbb{E}\Big[\mathrm{e}^{-t\mathfrak{e}_{1}(a)}\Big]+(-1)^{m}\sum_{\substack{j_{-1}+j_{0}+j_{1}+2j_{2}+\cdots+mj_{m}=m\\j_{2}+j_{3}+\cdots+j_{m}\neq 0}}\frac{(-1)^{j_{-1}}}{\prod_{i=2}^{m}(-i)^{j_{i}}\prod_{i=-1}^{m}j_{i}!}\\
&&\sum_{n=0}^{m-\sum_{i=-1}^{m}j_{i}}\frac{(m-n-j_{0}-j_{1}-1)!}{(m-n-j_{0}-j_{-1}-j_{1}-1)!t^{m-n-j_{0}-j_{1}}}\left(\frac{\mathrm{d}^{j_{0}}}{\mathrm{d}t^{j_{0}}}\mathcal{P}_{n}^{(a)}(j_{2},\ldots,j_{m};t)\right)\frac{\mathrm{d}^{n+j_{1}}}{\mathrm{d}t^{n+j_{1}}}\mathbb{E}\Big[\mathrm{e}^{-t\mathfrak{e}_{1}(a)}\Big]\Bigg)\Bigg|_{t=0},
\eeas
with $\mathcal{P}_{n}^{(a)}(j_{2},\ldots,j_{m};t)$ given as in Theorem \ref{mainpainlevethm} and
\beas
t\mapsto t \frac{\di}{\di t} \log \mathbb{E}\left[\mathrm{e}^{-t\mathfrak{e}_1(a)}\right]
\eeas
satisfies the $\sigma$-Painlev\'{e} III$'$  equation (\ref{eq:Painleve}).
\end{cor}
The above result is an immediate consequence of combining Theorems \ref{thm:jointmom} and \ref{mainpainlevethm}. In fact,
note that $\phi_{\mathbf{A}}(\mathrm{e}^{-\i\theta})=\varphi_{\mathbf{A}}(\theta)$, so it is not hard to check that the leading coefficient of the moment $M_{s}(G(N),k)$ coincides with that of $\int_{G(N)}\left|\varphi_\mathbf{A}^{(k)}(0)\right|^{s} \mathrm{d}\mu_{\mathrm{Haar}}(\mathbf{A})$ when $G(N)=\mathrm{USp}(2N)$ or $\mathrm{SO}(2N)$. Moreover, these results yield a more general version of Corollary \ref{answeraquestion}, describing the joint moments of the $k_i$-th derivatives for $i = 1, \ldots, m$. To help the reader better understand the result, we present it here for the case $m = 1$ only. The results for general $m$ are analogous but have more complicated forms. As a consequence of these results, we can recursively compute all \( M_{s}(G(N), k_{1}, \ldots, k_{m}) \) by using an asymptotic analysis of the Taylor expansion of \( \mathbb{E}\left[\mathrm{e}^{-t \mathfrak{e}_1(a)}\right] \) at \( t = 0 \), due to the properties of  solutions to the corresponding Painlev\'{e} equation. 

We now briefly explain how the approach we use herein that leads to the proof of Corollary \ref{answeraquestion} is different from those used in \cite{altugetal} and \cite{andrade2024joint}. Initially, in exploring connections between the asymptotic behavior of $M_{s}(G(N),k)$ for unitary  symplectic and orthogonal groups and the Painlev\'{e}
equations, we began with the determinant representation $\det_{1\leq i,j\leq s} \Big(f_{2i-j-1+l_{i}}(t)\Big)$
as in (\ref{the previous determinant}). See \cite{gharakhloo2023modulated} for some investigations in this direction. However, the absence of a Hankel structure or a Toeplitz property in this determinant poses a challenge for further analysis through Okamoto's $\tau$-function theory or the Riemann-Hilbert method.

Therefore, instead of working directly with (\ref{the previous determinant}) or (\ref{the previous determinant2}), we turned to a hidden Hankel structure in $M_{s}(G(N),2)$ for finite $N$, whose entries involve the Laplace transform of certain Jacobi-like weights. For general $k$, we expressed the moment in terms of a Hankel determinant shifted by partitions.  We then used combinatorial methods to analyze the connections between these shifted Hankel determinants and the original one, which satisfies a $\sigma$-Painlev\'{e} V equation. Finally, after appropriate rescaling and normalization (by dividing by a suitable power of $N$), we first establish the existence of the limit using our results from Section \ref{section2},
and then by taking the limit as 
$N\to \infty$ in the expression that connects the moments for finite matrix size $N$ with the $\sigma$-Painlev\'{e} V equation, we obtain the desired relation between the leading coefficients of $M_{s}(G(N),k)$
and a solution of the $\sigma$-Painlev\'{e} III$'$ equation.
A more detailed sketch of the proof strategy is provided in Section \ref{ideas}. 

\paragraph{Acknowledgements.} 
F.W. acknowledges support from the Royal Society, grant URF$\backslash$R$\backslash$

231028.
\paragraph{Open Access.} For the purpose of Open Access, the authors have applied a CC BY public copyright license to any Author Accepted Manuscript (AAM) version arising from this submission.




    

\section{Convergence of joint moments of power sums}\label{section2}

We present a quick argument that takes as input convergence of points along with some uniform integrability and gives joint convergence in distribution and of appropriate joint moments for an arbitrary number of power sums. We  only apply this to the hard-edge rescaled Jacobi ensemble in Section \ref{sec:Jacobi} but it could be of independent interest. 

The general setting is as follows. Suppose we are given sequences of random variables $(\lambda_{j,N})_{1\leq j\leq N}$, for $N \in \mathbb{N}$, and $(\gamma_j)_{j=1}^\infty$  such that:
\begin{itemize}
    \item For every $N\in \mathbb{N}$, we have, almost surely,
    \begin{equation*}
        0<\lambda_{1,N}<\lambda_{2,N}<\cdots<\lambda_{N,N},
    \end{equation*}
    and almost surely $0<\gamma_1<\gamma_2<\gamma_3<\cdots$.
    \item For every fixed $l\in \mathbb{N}$, as $N\to \infty$, we have,
\begin{equation}
\label{eq:convergenceassump}
    \left(\lambda_{1,N},\ldots, \lambda_{l,N}\right)\convd (\gamma_1,\ldots, \gamma_l).
\end{equation}
\end{itemize}
Abusing notation we will always write $\mathbb{P}$ for the underlying laws of these random variables, even if not necessarily defined on the same probability space, and $\mathbb{E}$ for the corresponding expectations. Finally define, for $k \in \mathbb{N}$, the negative power sums:
\begin{align*}
    M_N^{(k)} = \sum_{j=1}^N \frac{1}{\lambda_{j,N}^k}, \ \ M^{(k)}= \sum_{j=1}^\infty \frac{1}{\gamma_j^k}.
\end{align*}
Observe that, the random variables $M^{(k)}$ are well-defined but without further assumptions, in principle, could be infinite.

\begin{prop}
\label{thm:generalmain}
   In the setting described above, suppose 
    \begin{equation}
    \label{eq:convergeofmom1assump}
        \lim_{N\to \infty} \mathbb{E}\left[M_N^{(1)}\right]= \mathbb{E}\left[M^{(1)}\right]<\infty.
    \end{equation}
    Then, for a fixed $m\in \mathbb{N}$, as $N\to \infty$, we have:
     \begin{equation*}
        \left(M_N^{(1)},\ldots, M_N^{(m)}\right)\convd \left(M^{(1)},\ldots, M^{(m)}\right).
    \end{equation*}
    Moreover, if we further have that
    \begin{equation}
    \label{eq:mombdassump}
       \sup_N \mathbb{E}\left[\big(M_N^{(1)}\big)^{r}\right]<\infty,
    \end{equation}
    for some $r\geq1$, and $h_1,\ldots,h_m \in [0,\infty)$ are such that $\sum_{k=1}^{m}k h_k<r$, then we also have convergence of joint moments as $N\rightarrow \infty$:
    \begin{equation}
\label{eq:generaljointmomconv}
\mathbb{E}\left[\prod_{k=1}^m \left(M_N^{(k)}\right)^{h_k}\right] \longrightarrow \mathbb{E}\left[\prod_{k=1}^m \left(M^{(k)}\right)^{h_k}\right]<\infty.
    \end{equation}  
\end{prop}
\noindent The strategy to show convergence of joint moments will be to split each $M_N^{(k)}$ into separate contributions, with the first, and also dominating term having the desired convergence properties, so that the final result is obtained by a limiting procedure. With this goal in mind, for fixed $l\in \mathbb{N}$, we define:
\begin{equation*}
    \mathsf{Z}_{N,l}^{(k)} = \sum_{j=1}^l \frac{1}{\lambda_{j,N}^k}, \quad 
\mathsf{Y}_{N,l}^{(k)} = \sum_{j=l+1}^N \frac{1}{\lambda_{j,N}^k},
\end{equation*}
and their infinite counterparts:
\begin{equation*}
    \mathsf{Z}_{l}^{(k)} = \sum_{j=1}^l \frac{1}{\gamma_{j}^k}, \quad 
\mathsf{Y}_{l}^{(k)} = \sum_{j=l+1}^\infty \frac{1}{\gamma_{j}^k},
\end{equation*}
so that $M_N^{(k)}=\mathsf{Z}_{N,l}^{(k)}+\mathsf{Y}_{N,l}^{(k)}$ and $M^{(k)}= \mathsf{Z}_{l}^{(k)}+\mathsf{Y}_{l}^{(k)}$ for each $l,N\in \mathbb{N}$. We begin with the following lemma.
 \begin{lem}
 \label{lem:Yconv}
     Under the assumptions of Proposition \ref{thm:generalmain}, we have that for each fixed $l\in \mathbb{N}$,
     \begin{equation*}
         \mathsf{Y}_{N,l}^{(1)}\longrightarrow \mathsf{Y}_l^{(1)}
     \end{equation*}
     in distribution as $N\to \infty$.
     
 \end{lem}
 \begin{proof}
First observe that for each $l\in \mathbb{N}$, the sequence $\{\mathsf{Y}_{N,l}^{(1)}\}_{N\geq 1}$ is tight due to the assumption in \eqref{eq:convergeofmom1assump}. In particular, we have that every subsequence has a further subsequence $\{\mathsf{Y}_{N_j,l}^{(1)}\}_{j\geq 1}$ that converges in distribution to a random variable $\mathsf{T}_l$.
Hence, by virtue of the Portmanteau lemma, we have that for each $l,m\in \mathbb{N}$ with $l \leq m$ and $t\in (0,\infty)$, 
\begin{align*}
\mathbb{P}\left(\mathsf{
T}_l\geq t\right)\geq\limsup_{j\to \infty} \mathbb{P}\left(\mathsf{Y}_{N_j,l}^{(1)}\geq t\right)\geq\liminf_{j\to \infty}\mathbb{P}\left(\sum_{i=l+1}^m \frac{1}{\lambda_{i,N_j}}> t\right)= \mathbb{P}\left(\sum_{i=l+1}^m \frac{1}{\gamma_{i}}> t\right).
\end{align*}
Next, taking $m\to \infty$ by means of monotone convergence theorem, we see that
\begin{equation}
\label{eq:ineq1}
    \mathbb{P}(\mathsf{T}_l\geq t)\geq \mathbb{P}(\mathsf{Y}_l^{(1)}> t).
\end{equation}
Now, observe that immediately from the assumptions, we have $\mathsf{Z}_{N,l}^{(1)}\longrightarrow \mathsf{Z}_l^{(1)}$ in distribution as $N\to \infty$. Hence, using \eqref{eq:convergeofmom1assump} and Fatou's lemma, we obtain
\begin{equation*}
    \limsup_{N\to \infty} \mathbb{E}\left[\mathsf{Y}_{N,l}^{(1)}\right]=\limsup_{N\to \infty} \bigg\{\mathbb{E}\left[M_N^{(1)}\right]-\mathbb{E}\left[\mathsf{Z}_{N,l}^{(1)}\right]\bigg\}\leq\mathbb{E}\left[M^{(1)}\right]-\mathbb{E}\left[\mathsf{Z}_{l}^{(1)}\right]=\mathbb{E}\left[\mathsf{Y}_{l}^{(1)}\right].
\end{equation*}
In particular, using the above gives, upon integrating both sides of \eqref{eq:ineq1} over $t\in (0,\infty)$, 
\begin{equation*}
    \int_{0}^\infty \mathbb{P}(\mathsf{T}_l\geq t) \mathrm{d}t = \mathbb{E}[\mathsf{T}_l]\leq\liminf_{j\to \infty} \mathbb{E}\left[\mathsf{Y}_{N_j,l}^{(1)}\right]\leq \mathbb{E}[\mathsf{Y}_l^{(1)}]= \int_0^\infty \mathbb{P}(\mathsf{Y}_l^{(1)}> t) \mathrm{d}t,
 \end{equation*}
 where for the first inequality we again use Fatou's lemma. Therefore, we therefore must have that
\begin{equation}
\label{eq:TlYldist}
   \mathbb{P}(\mathsf{T}_l\geq t)=\mathbb{P}(\mathsf{Y}_l^{(1)}> t) 
\end{equation}
for almost every $t\in (0,\infty)$. But then the function $t \mapsto \mathbb{P}(\mathsf{T}_l\geq t)$ can have countably many discontinuities, so that $ \mathbb{P}(\mathsf{T}_l>t)=\mathbb{P}(\mathsf{Y}_l^{(1)}> t) $ for almost every $t\in (0,\infty)$, hence showing that $\mathsf{T}_l$ must have the same law as $\mathsf{Y}_l^{(1)}$. In particular, we conclude that every subsequence of $\{\mathsf{Y}_{N,l}^{(1)}\}_{N\geq 1}$ has a further subsequence that converges in law to $\mathsf{Y}_l^{(1)}$, which then implies the desired convergence for the full sequence. This completes the proof of the lemma.
\end{proof}
\noindent Observe that immediately from the lemma above with $l=1$, we have the desired distributional convergence for $M_N^{(1)}$. To deal with $M_N^{(k)}$ with $k\geq 2$, we proceed with the following lemma.

\begin{lem}
\label{ineq209}
    With $M_N^{(k)}$, $\mathsf{Z}_{N,l}^{(k)}$, and $\mathsf{Y}_{N,l}^{(k)}$ defined as above, we have
    \begin{equation}
    \label{eq:jmbd1}
        \left( \mathbb{E}\left[\prod_{k=1}^m \left(M_N^{(k)}\right)^{h_k}\right]\right)^{1/S} \le \sum_{J\subseteq\{1,\dots,m\}} \left( \mathbb{E}\left[ \prod_{k\in J} \left(\mathsf{Z}_{N,l}^{(k)}\right)^{h_k} \prod_{k\notin J} \left(\mathsf{Y}_{N,l}^{(k)}\right)^{h_k} \right] \right)^{1/S},
    \end{equation}
    where $S\defeq \max(\sum_{k=1}^m h_k,1)$ with $h_{1},\ldots,h_{m}\in [0,\infty)$, and the sum is over all subsets of $\{1,\ldots,m\}$ including $\emptyset$.
\end{lem}
\begin{proof}
  First, consider the case $\sum_{k=1}^m h_k\geq1$. Define, for each $k\in \{1,\ldots,m\}$, $p_k=\frac{h_k}{S}$. It is easy to see that 
\[
\left(M_N^{(k)}\right)^{h_k} = \Bigl[\bigl(\mathsf{Z}_{N,l}^{(k)}+\mathsf{Y}_{N,l}^{(k)}\bigr)^{p_k}\Bigr]^{S} \le \Bigl[\bigl(\mathsf{Z}_{N,l}^{(k)}\bigr)^{p_k} + \bigl(\mathsf{Y}_{N,l}^{(k)}\bigr)^{p_k}\Bigr]^{S}.
\]
The inequality now follows by taking products over \(k=1,\dots,m\), expanding the product, and using the triangle inequality for $\mathrm{L}^S$. If, on the other hand, $\sum_{k=1}^m h_k<1$, then $h_j\in[0,1)$ for $j=1,\ldots,m$, so that the inequality follows immediately by using the sub-additivity of the map $t\mapsto t^{h_k}$.
\end{proof}
\noindent In what follows, for ease of notation, we define, for each $J\subset \{1,\ldots,m\}$ and $l\in \mathbb{N}$, the averages
\begin{equation*}
\mathsf{T}_{N,J}^{(l)}=\mathbb{E}\left[ \prod_{k\in J} \left(\mathsf{Z}_{N,l}^{(k)}\right)^{h_k} \prod_{k\notin J} \left(\mathsf{Y}_{N,l}^{(k)}\right)^{h_k} \right] .
\end{equation*}
\begin{proof}[Proof of Proposition \ref{thm:generalmain}] We start by proving the convergence of joint moments. That is, we fix $r\geq 1$ as in the statement of the result (note that the first moments are always uniformly bounded under the assumption \eqref{eq:convergeofmom1assump}), and show that whenever $h_1,\ldots,h_m\in [0,\infty)$ are such that $Q \defeq \sum_{k=1}^m k h_k<r$, we have the convergence given in \eqref{eq:generaljointmomconv}. We will then show that this implies distributional convergence, again heavily exploiting positivity. To begin the proof, we first claim that
    \begin{equation}
    \label{eq:claim1}
        \mathsf{T}_{N,\{1,\ldots,m\}}^{(l)}\longrightarrow \mathbb{E}\left[\prod_{k=1}^m \left(\sum_{i=1}^l \frac{1}{\gamma_i^{k}}\right)^{h_k}\right]
    \end{equation}
    as $N\to \infty$. Indeed, by the assumption in \eqref{eq:convergenceassump}, we have
    \begin{equation*}
      \prod_{k=1}^m \left(\sum_{i=1}^l \frac{1}{\lambda_{i,N}^{k}}\right)^{h_k} \convd \prod_{k=1}^m \left(\sum_{i=1}^l \frac{1}{\gamma_i^{k}}\right)^{h_k}.
    \end{equation*}
    Combining this with the uniform integrability given by applying the subadditivity of $t\mapsto t^{1/k}$ on $[0,\infty)$ to each multiplicand, and making use of the assumption \eqref{eq:mombdassump}, we obtain the convergence of averages. Moving on, note once again utilizing positivity, we have
    \begin{equation*}
   \liminf_{N\to \infty}\mathbb{E}\left[\prod_{j=1}^m (M_N^{(k)})^{h_k}\right] \geq \liminf_{N\to \infty }\E\left[\prod_{k=1}^m \left(\sum_{i=1}^l \frac{1}{\lambda_{i,N}^{k}}\right)^{h_k}\right]=\E\left[\prod_{k=1}^m \left(\sum_{i=1}^l \frac{1}{\gamma_i^{k}}\right)^{h_k}\right].
\end{equation*}
    Now, taking $l\to \infty$ on the right hand side gives that
    \begin{equation}\label{inequality109}
        \liminf_{N\to \infty}  \mathbb{E}\left[\prod_{k=1}^m (M_N^{(k)})^{h_k}\right]\geq\mathbb{E}\left[\prod_{k=1}^m (M^{(k)})^{h_k}\right].
    \end{equation}
    For the upper bound, we first claim that whenever $J\neq \{1,\ldots,m\}$,
    \begin{equation}\label{ineq20910}
        \limsup_{l\to \infty} \limsup_{N\to \infty} \mathsf{T}_{N,J}^{(l)}=0.
    \end{equation}
   Indeed, applying Holder's inequality with exponents $\frac{Q}{kh_k}$ on the term corresponding to $k\in J$, we see that
   \begin{equation*}
    \mathsf{T}_{N,J}^{(l)}\leq \prod_{k\in J} \mathbb{E}\bigg[(\mathsf{Z}_{N,l}^{(k)})^{Q/k}\bigg]^{\frac{kh_k}{Q}}  \prod_{k\notin J} \mathbb{E}\bigg[(\mathsf{Y}_{N,l}^{(k)})^{Q/k}\bigg]^{\frac{kh_k}{Q}}.
   \end{equation*}
   The claim (\ref{ineq209}) now follows by observing that for each $k$,  we have $(\mathsf{Z}_{N,l}^{(k)})^{Q/k}\leq (\mathsf{Z}_{N,l}^{(1)})^Q$ and $(\mathsf{Y}_{N,l}^{(k)})^{Q/k}\leq (\mathsf{Y}_{N,l}^{(1)})^Q$, and also that by means of Lemma \ref{lem:Yconv} along with assumption \eqref{eq:mombdassump},
   \begin{equation*}
     \limsup_{l\to \infty}  \limsup_{N\to \infty}  \mathbb{E}\bigg[(\mathsf{Y}_{N,l}^{(k)})^{Q/k}\bigg]\leq \limsup_{l\to \infty}  \limsup_{N\to \infty} \mathbb{E}[(\mathsf{Y}_{N,l}^{(1)})^{Q}]=\limsup_{l\to \infty}\mathbb{E}\left[(\mathsf{Y}_l^{(1)})^Q\right]=0.
   \end{equation*}
Then, combining with \eqref{eq:claim1}, (\ref{ineq20910}) and Lemma \ref{ineq209},
\begin{equation}\label{eq30910}
 \limsup_{N\to \infty}\mathbb{E}\left[\prod_{j=1}^m (M_N^{(k)})^{h_k}\right] \leq \E\left[\prod_{k=1}^m \left(M^{(k)}\right)^{h_k}\right].
\end{equation}
By (\ref{inequality109}) and (\ref{eq30910}), we have the convergence of joint moments. 

For the joint convergence in law, pick, under the given assumptions, $h_1,\ldots,h_m\in \mathbb{R}_+$ such that $Q=\sum_k kh_k <r$.  
    First, note that by a combination of Markov's inequality and Prokhorov's theorem, the random vectors
    \begin{equation*}
        \mathsf{X}_N\defeq \left((M_N^{(1)})^{h_1},\ldots, (M_N^{(m)})^{h_m}\right)
    \end{equation*}
    form a tight family as $N$ ranges over $\mathbb{N}$. In particular, every subsequence of $\{\mathsf{X}_N\}_{N\geq 1}$ has a further subsequence which converges in distribution. Next, we claim that all of these subsequential limits are identical. Indeed, consider an arbitrary subsequence $\{\mathsf{X}_{N_j}\}_{j\geq 1}$, and let $(Y_1,\ldots,Y_m)$ be the limiting random vector. Then, for each $l\in \mathbb{N}$, by the Portmanteau lemma, one has: 
    \begin{align*}
      \mathbb{P}\left(Y_1\geq t_1,\ldots, Y_m\geq t_m\right)&\geq\limsup_{j\to \infty} \mathbb{P}\left((M_{N_j}^{(1)})^{h_1}\geq t_1,\ldots,(M_{N_j}^{(m)})^{h_m}\geq t_m\right)\\&\geq \liminf_{j\to \infty} \mathbb{P}\left((\mathsf{Z}_{{N_j},l}^{(1)})^{h_1}> t_1,\ldots,(\mathsf{Z}_{{N_j},l}^{(m)})^{h_m}
      >t_m\right)\\&\geq\mathbb{P}\left((\mathsf{Z}_{l}^{(1)})^{h_1}> t_1,\ldots,(\mathsf{Z}_{l}^{(m)})^{h_m}> t_m\right).
    \end{align*}
    Next, noting that each $\mathsf{Z}_l^{(k)}$ is monotonically increasing in $l$, we see that by the monotone convergence theorem,
    \begin{equation*}
        \lim_{l\to \infty} \mathbb{P}\left((\mathsf{Z}_l^{(1)})^{h_1}>t_1,\ldots, (\mathsf{Z}_l^{(m)})^{h_m}>t_m\right)=\mathbb{P}\left((M^{(1)})^{h_1}> t_1,\ldots, (M^{(m)})^{h_m}>t_m\right).
    \end{equation*}
    Hence, we have shown that
    \begin{equation}
    \label{eq:distineqvector}
        \mathbb{P}(Y_1\geq t_1,\ldots, Y_m\geq t_m)\geq\mathbb{P}\left((M^{(1)})^{h_1}> t_1,\ldots, (M^{(m)})^{h_m}> t_m\right).
    \end{equation}
    
   Next, we show that integrating both sides of this inequality leads to the same value. Indeed, using the first part of the proof, we have
   \begin{align*}
       \int_{\R_+^m} \mathbb{P}(Y_1\geq t_1,\ldots, &Y_m\geq t_m) \mathrm{d}\mathbf{t}=\mathbb{E}\left[Y_1\cdots Y_m\right]=\lim_{j\to \infty} \mathbb{E}\left[(M_{N_{j}}^{(1)})^{h_1}\cdots (M_{N_{j}}^{(m)})^{h_m}\right]\\ &=\mathbb{E}\left[(M^{(1)})^{h_1}\cdots (M^{(m)})^{h_m}\right]= \int_{\R_+^m} \mathbb{P}\left((M^{(1)})^{h_1}>t_1,\ldots, (M^{(m)})^{h_m}> t_m\right) \mathrm{d}\mathbf{t}.
   \end{align*}
In particular, it must be the case that equality is achieved in \eqref{eq:distineqvector} for almost every
$\mathbf{t}=(t_1,\ldots, t_m)\in \R_+^m$. Hence, arguing as in the final paragraph of the proof of Lemma \ref{lem:Yconv}, we arrive at the distributional equality
\begin{equation*}
    (Y_1,\ldots, Y_m)\stackrel{\mathrm{d}}{=}((M^{(1)})^{h_1},\ldots, (M^{(m)})^{h_m}).
\end{equation*}
Thus, we have shown that every subsequence of $\{\mathsf{X}_N\}_{N\geq 1}$ has a further subsequence that converges in law to $((M^{(1)})^{h_1},\ldots, (M^{(m)})^{h_m})$, and hence the whole sequence must converge in distribution to the same vector, which completes the proof.
\end{proof}

\subsection{Application to Jacobi ensemble} \label{sec:Jacobi}

We apply Proposition \ref{thm:generalmain} to the hard-edge rescaled Jacobi ensemble. The Jacobi ensemble is defined as follows \cite{ForresterBook}.

\begin{defn}
\label{def:jacobi}
The Jacobi ensemble $\mu_{N}^{(a,b)}$ with parameters $a,b>-1$ is the probability measure on $\mathbb{W}_N$ given by
\begin{equation*}
    \mu_{N}^{(a,b)}(\di \mathbf{x})=\frac{1}{\mathcal{Z}_{N}^{a,b}}\prod_{j=1}^Nx_j^{a}(1-x_j)^{b}\mathbf{1}_{x_j\in [0,1]} \prod_{i<j}|x_i-x_j|^2 \mathrm{d}\mathbf{x},
\end{equation*}
where the normalisation constant is given by the following formula, 
\beas
\mathcal{Z}_{N}^{a,b}=\frac{1}{N!}\prod_{j=1}^{N}\frac{\Gamma(a+j)\Gamma(b+j)\Gamma(1+j)}{\Gamma(a+b+N+j)}.
\eeas
\end{defn}

\noindent For $k \in \mathbb{N}$,recall that we defined earlier the negative power sums
\begin{equation*}
    \mathfrak{p}_{m,N}^{(a,b)}\defeq \sum_{j=1}^N \frac{1}{\big(\mathsf{x}_j^{(N)}\big)^m}
\end{equation*}
where $\lambda_{j,N}=N^2\mathsf{x}_{j}^{(N)}$, and $\big(\mathsf{x}_{1}^{(N)},\ldots,\mathsf{x}_{N}^{(N)}\big)$ are distributed according to $\mu_{N}^{(a,b)}$. The $N^2$-scaling here is the so-called hard-edge scaling for the Jacobi ensemble. In what follows, we suppress the parameters $a,b$ from the notation whenever their values are clear. Also recall that we defined their Bessel analogues:
\begin{equation*}
   \mathfrak{p}_m(a)\defeq \sum_{j=1}^\infty \frac{1}{\big(\mathfrak{B}_{a}^{j}\big)^{m}}.
\end{equation*}
We then have the following result, where here, and throughout the rest of the paper, for any positive real number $x$, we define 
 {\[
\lfloor x \rfloor' \defeq
\begin{cases}
\lfloor x \rfloor, & x \notin \mathbb{Z}, \\
x - 1, & x \in \mathbb{Z}.
\end{cases}
\]
}

\begin{prop}
\label{thm:main}
    Let $a>0$, $b>-1$, and $m\in \mathbb{N}$. Then, we have
    \begin{equation*}
       \big( \mathfrak{p}_{1,N}^{(a,b)},\ldots,  \mathfrak{p}_{m,N}^{(a,b)}\big) \convd \big(\mathfrak{p}_1(a),\ldots,\mathfrak{p}_m(a)\big),
    \end{equation*}
    as $N\to \infty$, where $\convd$ denotes convergence in distribution. Moreover,  for $h_1,\ldots,h_m \in \mathbb{R}_{\geq 0}$, with $\sum_{k=1}^m k h_k <\lfloor a+1\rfloor'$ we have convergence of moments:
    \begin{equation*}
        \mathbb{E}_{N}^{(a,b)}\left[\prod_{k=1}^m (\mathfrak{p}_{k,N})^{h_k}\right]\longrightarrow \mathbb{E}\left[\prod_{k=1}^m (\mathfrak{p}_k(a))^{h_k}\right]< \infty,  \ \ \mathrm{as} \  N \to \infty.
    \end{equation*}
\end{prop}
\begin{proof}
    The result will follow as a special case of Proposition \ref{thm:generalmain}, once the conditions therein are verified. Firstly, the fact that convergence in \eqref{eq:convergenceassump} is guaranteed for $(\lambda_{1,N},\ldots, \lambda_{l,N})=(\mathsf{x}_1^{(N)},\ldots,\mathsf{x}_l^{(N)})$ is well-known, even for the Jacobi $\beta$-ensemble, see \cite{dianejacobi, valkoli}. As for the moment assumptions, we first note that under the given conditions on the parameters, we have
   \begin{equation*}
        \lim_{N\to \infty}\mathbb{E}_N^{(a,b)}  \Big[\mathfrak{p}_{1,N}\Big]=\frac{1}{a}=\mathbb{E} \big[\mathfrak{p}_1(a)\big],
    \end{equation*}
where the first equality follows from Aomoto's integral formula (see, for instance, \cite{forrester_importanceofselberg}), and the second equality is seen by first writing the average as the one-dimensional integral $\int_0^\infty x^{-1}\mathsf{K}_a^{\mathrm{Bes}}(x,x)\mathrm{d}x$, thanks to the determinantal property of the process $\{\mathfrak{B}_a^k\}_{k=1}^\infty$, and computing this integral explicitly using formulae for Bessel functions (a roundabout way to obtain the desired formula would be to combine results of \cite{CircJacobiBeta,assiotiswishart}). Finally, as a consequence of \cite[Proposition 4.1]{forrester2022joint}, we obtain that whenever $\alpha \in \mathbb{N}$ with $\alpha<a+1$, we have
    \begin{equation*}
        \sup_N \mathbb{E}_N^{(a,b)} \Big[\big(\mathfrak{p}_{1,N}\big)^{\alpha}\Big] <\infty,
    \end{equation*}
 so that the result now follows from an application of Proposition \ref{thm:generalmain} with $r$ being the largest such $\alpha$.
\end{proof}

\section{Proofs of Theorem \ref{thm:jointmom} and Proposition \ref{thm:laguerre}}\label{section3}

The main goal of this section will be to prove Theorem \ref{thm:jointmom}. Towards this end recall that, for a Haar-distributed random matrix $\mathbf{A}\in \mathrm{USp}(2N)$ with eigenvalues $\mathrm{e}^{\pm \i\theta_1},\ldots, \mathrm{e}^{\pm \i\theta_N}$, the density of the eigenangles $(\theta_1,\ldots, \theta_N)\in [0,\pi)$ is given by
\begin{equation*}
    \frac{1}{\mathcal{Z}_N^{\mathrm{USp}}}\prod_{1\leq j<k\leq N} \left(\cos(\theta_k)-\cos(\theta_j)\right)^2 \prod_{j=1}^N \sin^2(\theta_j) \mathrm{d}\theta_j,
\end{equation*}
for an explicit normalisation constant $\mathcal{Z}_N^{\mathrm{USp}}$, see \cite{ForresterBook}. Similarly, the density of the eigenangles $(\theta_1,\ldots, \theta_N)\in [0,\pi)$ for $\mathbf{A}\in \mathrm{SO}(2N)$ is given by:
\begin{equation*}
   \frac{1}{\mathcal{Z}_N^{\mathrm{SO}}} \prod\limits_{1 \leq k <j \leq N} \left(\cos\left(\theta_k\right)-\cos\left(\theta_j\right)\right)^2 \prod_{j=1}^N\mathrm{d} {\theta}_j,
\end{equation*}
for an explicit normalisation constant $\mathcal{Z}_N^{\mathrm{SO}}$, see \cite{ForresterBook}. Our starting point will be the following formula.

\begin{prop}\label{lem:jointmomSym}
Suppose $h_0,\ldots,h_m\in [0,\infty)$. Let $s=\sum_{i=0}^{m}h_{i}$.  We have 
\begin{align*}
    \mathfrak{J}^{\mathrm{USp}}_N(h_{0},h_1,\ldots,h_m)&= \mathfrak{J}^{\mathrm{USp}}_N(s,0,\ldots,0) \mathbb{E}_{N}^{(s+\frac{1}{2},\frac{1}{2})}\left[\prod_{k=1}^m (\mathfrak{R}_{N,k}(\mathsf{x}_1^{(N)},\ldots,\mathsf{x}_N^{(N)}))^{h_k}\right],\\
        \mathfrak{J}^{\mathrm{SO}}_N(h_{0},h_1,\ldots,h_m)&= \mathfrak{J}^{\mathrm{SO}}_N(s,0,\ldots,0) \mathbb{E}_{N}^{(s-\frac{1}{2},-\frac{1}{2})}\left[\prod_{k=1}^m (\mathfrak{R}_{N,k}(\mathsf{x}_1^{(N)},\ldots,\mathsf{x}_N^{(N)}))^{h_k}\right],
    \end{align*}
where for any $k\in \mathbb{N}$, $\mathfrak{R}_{N,k}$ is defined by
\begin{equation*}
    \mathfrak{R}_{N,k}(x_1,\ldots,x_N)=\sum_{\substack{n_1,\dots,n_N\ge0 \\ n_1+\cdots+n_N=k}}
\binom{k}{n_1,\dots,n_N}\,\prod_{j:\, n_j\ge2} \left(1+\frac{2^{n_j}-2}{4x_j}\right).
\quad
\end{equation*}
 
\end{prop}

\begin{proof}
    We first claim that for any $k\geq 1$, we have
    \[
\frac{\varphi_\mathbf{A}^{(k)}(0)}{\varphi_{\mathbf{A}}(0)}
=(-\mathrm{i})^{k}\sum_{\substack{n_1,\dots,n_N\ge0 \\ n_1+\cdots+n_N=k}}
\binom{k}{n_1,\dots,n_N}\,\prod_{j:\, n_j\ge2} \left(1+\frac{2^{n_j}-2}{4x_j}\right),
\quad\textrm{with } x_j=\frac{1}{2}(1-\cos\theta_j).
\]
    Indeed, writing
    \[
\varphi_\mathbf{A}(\theta)=\prod_{j=1}^N \varphi_j(\theta),
\]
where we define
\[
\varphi_j(\theta)\defeq \big(1-\mathrm{e}^{-\mathrm{i}(\theta-\theta_j)}\big)\big(1-\mathrm{e}^{-\mathrm{i}(\theta+\theta_j)}\big)=1-2\mathrm{e}^{-\mathrm{i}\theta}\cos\theta_j+\mathrm{e}^{-2\mathrm{i}\theta},
\]
we easily see that
\[
\varphi_j^{(n)}(0)=
\begin{cases}
2\,(1-\cos\theta_j), & n=0,\\[1mm]
(-\mathrm{i})^n\Bigl[2^n-2\cos\theta_j\Bigr], & n\ge1.
\end{cases}
\]
Hence, using the change of variables $x_j=(1-\cos(\theta_j))/2$, we see that for any $n\geq 1$,
\[
\frac{\varphi_j^{(n)}(0)}{\varphi_j(0)}
=(-\mathrm{i})^n\,\frac{2^n-2+4x_j}{4x_j}.
\]
Next, one may differentiate repeatedly to compute:
\[
\frac{\varphi^{(k)}_{\mathbf{A}}(0)}{\varphi_{\mathbf{A}}(0)}
=\sum_{\substack{n_1,\dots,n_N\ge0 \\ n_1+\cdots+n_N=k}}
\binom{k}{n_1,\dots,n_N}\prod_{j=1}^N\frac{\varphi_j^{(n_j)}(0)}{\varphi_j(0)}.
\]
From this the result follows by making the change of variables $x_j=\frac{1}{2}(1-\cos \theta_j)$ in the corresponding matrix integral.
\end{proof}

The main point in the formula above is that the highest contribution will come from sums of terms involving elementary symmetric polynomials in inverse points of $x_j$ sampled according to $\mu_{N}^{(a,b)}$. To establish this, we first need the following proposition. Here we make use of the determinantal property of $\mu_{N}^{(a,b)}$. One could obtain some uniform boundedness for elementary symmetric polynomials by writing them in terms of power sums and using our previous arguments but this restricts the parameters for which such bounds exist and it is quite suboptimal. 

Throughout the rest of this section we write $[N]=\{1,\ldots, N\}$ for each $N\in \mathbb{N}$, and let $[N]^{(l)}$ denote, for each $l,N\in \mathbb{N}$, the set of subsets of $[N]$ of size $l$. 

\begin{prop}
\label{prop:UIbdforelemsym}
    For each $a>0$, $b\geq -\frac{1}{2}$ we have
    \begin{equation*}
        \sup_N \frac{1}{N^{2l\alpha}}\E_{N}^{(a,b)}\left[\mathrm{e}_l\left(\frac{1}{\mathsf{x}_1^{(N)}},\ldots, \frac{1}{\mathsf{x}_N^{(N)}}\right)^{\alpha}\right]<\infty,
    \end{equation*}
    where $\alpha=\lfloor a+1\rfloor'$.

\end{prop}
\begin{proof} 
We first recall the determinantal property we need. Consider the so-called $l$-point density function:
\begin{equation*}
   \boldsymbol{\rho}_{N;l}^{(a,b)}(x_1,\ldots,x_l) \defeq  \frac{1}{\mathcal{Z}_{N}^{a,b}}\int_{(x_{l+1},\ldots, x_N)\in \R^{N-l}} \prod_{j=1}^N (1-x_j)^a x_j^b \prod_{1\leq j<k\leq N} (x_j-x_k)^2 \di x_{l+1}\ldots \di x_{N}.
\end{equation*}
Using standard arguments, see, for instance, \cite[Proposition 5.1.2]{ForresterBook}, this function can be written as:
\begin{equation}
    \boldsymbol{\rho}_{N;l}^{(a,b)}\left(x_1,\ldots, x_l\right)= \frac{(N-l)!}{N!} \det\left[K_N^{(a,b)}(x_i,x_j)\right]_{i,j=1}^l.
    \label{eq:dens}
\end{equation}
The so-called correlation kernel, $K_N^{(a,b)}(x,y)$, appearing in the determinant is given by:
\begin{equation}
\label{eq:corker}
    K_N^{(a,b)}(x,y) = \sqrt{w^{(a,b)}(x) w^{(a,b)}(y)} \sum_{j=1}^N \mathcal{P}_j^{(a,b)}(x) \mathcal{P}_j^{(a,b)}(y),
\end{equation}
where the weight $w^{(a,b)}$ is given by
\begin{equation*}
    w^{(a,b)}(x) = x^a(1-x)^b, \ \ x\in [0,1],
\end{equation*}
and the $\big\{\mathcal{P}_j^{(a,b)}\big\}_{j=1}^\infty$ are polynomials indexed by their degrees defined by the orthonormality condition:
\begin{equation*}
    \int_0^1 \mathcal{P}_j^{(a,b)}(x) \mathcal{P}_k^{(a,b)} (x) w^{(a,b)} (x) \di x = \mathbf{1}_{j=k}.
\end{equation*}
Now, since $\alpha$ is an integer, expanding, we see that 
    \begin{align*}
    \E_{N}^{(a,b)}\left[\mathrm{e}_l\left(\frac{1}{\mathsf{x}_1^{(N)}},\ldots, \frac{1}{\mathsf{x}_N^{(N)}}\right)^{\alpha}\right]\leq C(\alpha)\sum_{\sum_{J\in [N]^{(l)}} m_J=\alpha} \mathbb{E}_{N}^{(a,b)}&\left[\prod_{J\in [N]^{(l)}} \prod_{j\in J} \frac{1}{{\big(\mathsf{x}_j^{(N)}\big)}^{m_J}}\right]\\&=C(\alpha)\sum_{\sum_{J\in [N]^{(l)}} m_J=\alpha} \mathbb{E}_{N}^{(a,b)}\left[\prod_{j=1}^N \frac{1}{\big(\mathsf{x}_j^{(N)}\big)^{\sum_{J : j\in J}m_J}}\right],
    \end{align*}
where $C(\alpha)$ is a constant only depending on $\alpha$.
    Next, define, for a fixed $(m_J)_{J\in [N]^{(l)}}$, the sets
    \begin{align*}
        A=&\,\{j\in [N] \mid j\in J \textrm{ for some } J \textrm{ with } m_J\neq 0\},\\
        A_j = &\, \{J\in [N]^{(l)}\, \mid j\in J, \, m_J\neq 0\},
    \end{align*}
    and let $\kappa=|A|$. Then, observe that by the determinantal property \eqref{eq:dens}, we have that
    \begin{equation}
    \label{eq:ekUIeq1}
       \mathbb{E}_{N}^{(a,b)}\left[\prod_{j=1}^N \frac{1}{\big(\mathsf{x}_j^{(N)}\big)^{\sum_{J : j\in J}m_J}}\right]=\frac{(N-\kappa)!}{N!}\int_{(y_j)_{j\in A}\in \mathbb{R}^\kappa} \prod_{j\in A} \frac{1}{y_j^{\sum_{J\in A_j}m_J}} \det_{l,k\in A}\left[K_N^{(a,b)}(y_l,y_k)\right] \mathrm{d}\mathbf{y}.
    \end{equation}
    Next, note that applying Cauchy-Schwarz to \eqref{eq:corker} immediately gives 
    \begin{equation*}
        K_N^{(a,b)}(z,w)\leq \sqrt{K_N^{(a,b)}(z,z)}\sqrt{K_N^{(a,b)}(w,w)}.
    \end{equation*}
    Hence, expanding the determinant as a sum over the symmetric group and applying this inequality
    we obtain that the average in \eqref{eq:ekUIeq1} is bounded above by
    \begin{equation*}
        \frac{C(\alpha)}{\binom{N}{\kappa}} \prod_{j\in A} \int \frac{1}{y^{\sum_{J\in A_j}m_J}} K_N^{(a,b)}(y,y) \mathrm{d}y\leq \frac{\tilde{C}(a,b,\alpha)}{\binom{N}{\kappa}} \prod_{j\in A} N^{2\sum_{j\in J} m_J}= \frac{\tilde{C}(a,b,\alpha)}{\binom{N}{\kappa}}N^{2\alpha l}
    \end{equation*}
    where for the inequality we used the bound
    \begin{equation*}
        \int_0^1 x^{-\theta} K_N^{(a,b)}(x,x) \mathrm{d}x\leq C(a,b,\theta) N^{2\theta}
    \end{equation*}
    that holds for all $1\leq \theta<a+1$ (see the proof of Proposition 2.2 in \cite{gunes2022characteristic}), and $\tilde{C}(a,b,\alpha)=C(\alpha)\sup_{1\leq\theta<a+1}C(a,b,\theta)$. Combining these observations, one then has 
    \begin{equation*}
         \E_{N}^{(a,b)}\left[\mathrm{e}_l\bigg(\frac{1}{\mathsf{x}_1^{(N)}},\ldots, \frac{1}{\mathsf{x}_N^{(N)}}\bigg)^{\alpha}\right]\leq \tilde{C}(a,b,\alpha)\sum_{\sum_{J\in [N]^{(l)}} m_J=\alpha} N^{2\alpha l}\binom{N}{\kappa((m_J)_{J_\in [N]^{(l)}})}^{-1}.
    \end{equation*}
    Note, then that for each fixed $\kappa=l,\ldots,\alpha l$, there are at most 
    \begin{equation*}
        \tilde{C}(\alpha,l)\, \binom{N}{\kappa}
    \end{equation*}
    choices of $(m_J)_{J\in [N]^{(l)}}$ that gives $\kappa=\kappa\big((m_J)_{J_\in [N]^{(l)}}\big)$. This then gives the desired bound.
\end{proof}
Along with the bound given by the lemma above, we will need to understand the distribution in the $N\to \infty$ limit of the elementary symmetric polynomials. As expected, the objects in the limit will be the elementary symmetric functions in the inverse points of the Bessel point process. Hence, in order to ease the notation, we define:
\begin{equation}\label{definitionofeNk}
\mathfrak{e}_{N,k}^{(a,b)}\defeq\sum_{1\leq i_1<\cdots<i_k\leq N} \frac{1}{\mathsf{x}_{i_1}^{(N)}\cdots \mathsf{x}_{i_k}^{(N)}}
\end{equation}
where $(\mathsf{x}_i^{(N)})_{1\leq i \leq N}$ is distributed according to $\mu_{N}^{(a,b)}$. We then have the following lemma, a consequence of Proposition \ref{thm:main}.

\begin{lem}
\label{lem:skorokhod}
    Let $a>0$, $b>-1$, and $m\in \mathbb{N}$. Then there exist random vectors $(\mathfrak{J}_{N,1}^{(a,b)},\ldots,\mathfrak{J}_{N,m}^{(a,b)})_{N\geq 1}$ and $(\mathfrak{J}_{1}(a),\ldots,\mathfrak{J}_{m}(a))$ defined on a common probability space such that
    \begin{equation*}
        (\mathfrak{e}_{N,1}^{(a,b)},\ldots,\mathfrak{e}_{N,m}^{(a,b)})\stackrel{\mathrm{d}}{=}(\mathfrak{J}_{N,1}^{(a,b)},\ldots,\mathfrak{J}_{N,m}^{(a,b)})
    \end{equation*}
    for all $N\in \mathbb{N}$, and
    \begin{equation*}
(\mathfrak{e}_{1}(a),\ldots,\mathfrak{e}_{m}(a))\stackrel{\mathrm{d}}{=}(\mathfrak{J}_{1}(a),\ldots,\mathfrak{J}_{m}(a)).
    \end{equation*}
    Furthermore, on the probability space they are defined, we have the almost sure convergence
    \begin{equation*}
   (N^{-2}\mathfrak{J}_{N,1},\ldots,N^{-2m}\mathfrak{J}_{N,m})\longrightarrow (\mathfrak{J}_{1},\ldots,\mathfrak{J}_{m}), \ \ \textnormal{as } N \to \infty.
    \end{equation*}
\end{lem}
\begin{proof}
    First, note that applying Skorkohod's representation theorem with the distributional convergence given in Proposition \ref{thm:main}, the statement holds when elementary symmetric polynomials are replaced by power sums. The analogous result for elementary symmetric polynomials immediately follows by virtue of Newton's identities.
\end{proof}
\noindent We are now in position to prove Theorem \ref{thm:jointmom}.

\begin{proof}[Proof of Theorem \ref{thm:jointmom}]

We will first rewrite the sum that defines each $\mathfrak{R}_{N,k}$ as a sum of $O(1)$ terms each of which has controllable growth as $N\to \infty$. To this end, let $\mathcal{P}_k$ denote the set of integer partitions $\boldsymbol{\mu}=(\mu_1\geq \mu_2\geq \ldots \geq \mu_l)$ of $k$, and for each $\boldsymbol{\mu}\in \mathcal{P}_k$, let $\mathcal{Q}_{\boldsymbol{\mu}}$ be the set of integer tuples $(n_1,\ldots, n_N)$ such that upon re-ordering if necessary, $(n_1,\ldots, n_N)$ becomes $(\mu_1,\mu_2,\ldots, \mu_l,0,\ldots,0)$. It will turn out that the main contribution within each $\mathfrak{R}_{N,k}$ will come from partitions $\boldsymbol{\mu}$ such that $\mu_j\in \{1,2\}$ for each $j=1,\ldots,l$. Hence, denote the set of such integer partitions of $k$ by $\mathcal{P}_k^{(2)}$. Next, we write each $\mathfrak{R}_{N,k}$ as: 

    \begin{equation*}
        \mathfrak{R}_{N,k}\big(\mathsf{x}_1^{(N)},\ldots, \mathsf{x}_N^{(N)}\big)=\sum_{\boldsymbol{\mu}\in \mathcal{P}_k} \underbrace{\sum_{(n_1,\ldots, n_N)\in \mathcal{Q}_{\boldsymbol{\mu}}}
\binom{k}{\mu_1,\dots,\mu_l}\,\prod_{j:\, n_j\ge2} \left(1+\frac{2^{n_j}-2}{4\mathsf{x}_j^{(N)}}\right).}_{\defeq \mathsf{T}^{(N)}_{\mu}}
    \end{equation*}
   Next, let $\mathfrak{S}_{l,N}$ denote the set of injections $\sigma : [l]\to [N]$, and define, for each $r\in [k]$,
   \begin{align}
    m_r(\boldsymbol{\mu})=\# \left\{j : \mu_j=r\right\}.
   \end{align}
   Also letting $\theta=\theta(\boldsymbol{\mu})$ be $l(\boldsymbol{\mu})-m_1(\boldsymbol{\mu})$, one sees that 
   \begin{equation*}
      \mathsf{T}_{\boldsymbol{\mu}}^{(N)}= \sum_{\sigma\in \mathfrak{S}_{l,N}}
\binom{k}{\mu_1,\ldots, \mu_l}\,\prod_{r=1}^k \frac{1}{m_r(\boldsymbol{\mu})!}\prod_{j=1}^{\theta(\boldsymbol{\mu})} \left(1+\frac{2^{\mu_j}-2}{4\mathsf{x}_{\sigma(j)}^{(N)}}\right).
   \end{equation*}
   Now, observe that for $\boldsymbol{\mu}\notin \mathcal{P}_k^{(2)}$, using that the Jacobi ensembles are supported in $[0,1]^N$, we may bound:
\begin{equation*}
    \mathsf{T}_{\boldsymbol{\mu}}^{(N)}\leq C(\boldsymbol{\mu})\sum_{\sigma\in \mathfrak{S}_{l,N}} \prod_{j=1}^{\theta(\boldsymbol{\mu})}\frac{1}{x^{(N)}_{\sigma(j)}}\leq C(\boldsymbol{\mu}) N^{l-\theta} \mathfrak{e}_{N,\theta}.
\end{equation*}
In particular, one sees using Lemma \ref{lem:skorokhod} that since $l(\boldsymbol{\mu})+\theta(\boldsymbol{\mu})<k$ whenever $\boldsymbol{\mu} \notin \mathcal{P}_k^{(2)}$, we have
\begin{equation*}
    \frac{1}{N^k}\mathsf{T}^{(N)}_{\boldsymbol{\mu}} \convd 0
\end{equation*}
as $N\to \infty$. Furthermore, for the remaining terms, one may compute:
\begin{align*}
   \mathsf{T}_{\boldsymbol{\mu}}^{(N)}&= \sum_{\sigma\in \mathfrak{S}_{l,N}}
\binom{k}{\mu_1,\ldots, \mu_k}\,\prod_{r=1}^k \frac{1}{m_r(\boldsymbol{\mu})!}\prod_{j=1}^\theta \left(1+\frac{1}{2\mathsf{x}^{(N)}_{\sigma(j)}}\right)\\&=\sum_{\sigma\in \mathfrak{S}_{l,N}}
\binom{k}{\mu_1,\ldots, \mu_k}\,\prod_{r=1}^k \frac{1}{m_r(\boldsymbol{\mu})!}
\sum_{I\subset [\theta]}\prod_{j\in I} \frac{1}{2\mathsf{x}^{(N)}_{\sigma(j)}},
\end{align*}
rewriting this sum as $h=|I|$ ranges between $0$ and $\theta$, we see that
\begin{align*}
    \mathsf{T}^{(N)}_{\boldsymbol{\mu}}&=\binom{k}{\mu_1,\ldots, \mu_k}\,\prod_{r=1}^k \frac{1}{m_r(\boldsymbol{\mu})!}\sum_{h=0}^\theta  \binom{\theta}{h}\,\frac{(N-h)!}{(N-l)!} \,h!\,2^{-h}\sum_{I\in [N]^{(h)}} \prod_{j\in I} \frac{1}{\mathsf{x}_j^{(N)}}\\&=\binom{k}{\mu_1,\ldots, \mu_k}\,\prod_{r=1}^k \frac{1}{m_r(\boldsymbol{\mu})!}\sum_{h=0}^\theta  \binom{\theta}{h}\,\frac{(N-h)!}{(N-l)!} \,h!2^{-h}\mathfrak{e}_{N,h}.
\end{align*}
Now, again arguing analogously to before, we see that since $l(\boldsymbol{\mu})+h<k$ for all $h<\theta(\boldsymbol{\mu})$, after dividing by $N^k$, all the terms in the sum except the $h=\theta$ term converge to $0$ in distribution as $N\to \infty$. Hence, using these observations in combination with Lemma \ref{lem:skorokhod}, one obtains the convergence in distribution
\begin{equation*}\label{convergencequantity}
    \frac{1}{N^{\sum_{k=1}^{m}kh_{k}}}\prod_{k=1}^m (\mathfrak{R}_{N,k})^{h_k}\longrightarrow \prod_{k=1}^m (\mathfrak{R}_{k}(a))^{h_k}.
\end{equation*}
Furthermore, using Skorokhod's representation theorem once more, we may assume without loss of generality that this convergence takes place almost surely as well.
Thus, to prove the convergence of moments, it suffices to show that the former random variables form a uniformly integrable sequence. But then this follows by first picking $q>1$ such that $q\sum_{k=1}^{m} h_k$ is an integer with $q\sum_{k=1}^{m} h_k<a+1$, and then applying Holder's inequality to the average of the $q$-th power of these random variables with exponents $\frac{\sum_{k=1}^{m}h_k}{h_k}$, so that taking the supremum over $N$ gives a bounded quantity thanks to Proposition \ref{prop:UIbdforelemsym}. The result now follows by a combination of Proposition \ref{lem:jointmomSym}, Proposition \ref{prop:UIbdforelemsym} (applied with \((a,b) = \left(s + \frac{1}{2}, \frac{1}{2}\right)\) for $\mathrm{USp}(2N)$ and \((a,b) = \left(s - \frac{1}{2}, -\frac{1}{2}\right)\) for $\mathrm{SO}(2N)$), and the explicit asymptotic expressions
given in \cite{KeatingSnaithLfunctions}:
\begin{align*}
    \mathfrak{J}^{\mathrm{Sp}}_{N}(s,0,\ldots,0)&\sim N^{\frac{s(s+1)}{2}}\frac{2^{\frac{s^2}{2}}G(1+s)\sqrt{\Gamma(1+s)}}{\sqrt{G(1+2s)\Gamma(1+2s)}},\\
    \mathfrak{J}^{\mathrm{SO}}_{N}(s,0,\ldots,0)&\sim N^{\frac{s(s-1)}{2}} \frac{2^{\frac{s^2}{2}}G(1+s)\sqrt{\Gamma(1+2s)}}{\sqrt{G(1+2s)\Gamma(1+s)}}.
\end{align*}

\end{proof}
To end this section, we prove Proposition \ref{thm:laguerre} by exploiting a connection to the inverse Laguerre ensemble and the exchangeability theory results of \cite{assiotis2024exchangeablearraysintegrablesystems}. We need some definitions. Let $\mathbb{H}_+(N)$ be the space of $N\times N$ non-negative definite matrices. Consider the maps (observe that these are correctly defined):
\begin{align*}
    \mathfrak{P}_N^{N+1} : 
    \mathbb{H}_+(N+1)&\to \mathbb{H}_+(N),\\
    [\mathbf{H}_{ij}]_{i,j=1}^{N+1}&\mapsto [\mathbf{H}_{ij}]_{i,j=1}^N.
\end{align*}
 Define $\mathbb{H}_+(\infty)= \varprojlim
\mathbb{H}_+(N)$, the projective limit of the $\mathbb{H}_+(N)$ under these maps, and
 \begin{align*}
    \mathfrak{P}_N^{\infty} : 
    \mathbb{H}_+(\infty)&\to \mathbb{H}_+(N),\\
    [\mathbf{H}_{ij}]_{i,j=1}^{\infty}&\mapsto [\mathbf{H}_{ij}]_{i,j=1}^N.
\end{align*}
Consider the space $\mathcal{M}_{\mathrm{p}}^{\mathrm{Perm}}$ of permutation invariant probability measures on $\mathbb{H}_+(\infty)$ defined as follows: $\mu \in \mathcal{M}_{\mathrm{p}}^{\mathrm{Perm}}$ if and only if with $\mathbf{H}$ distributed according to $\mu$, for all $N \in \mathbb{N}$, $\mathbf{H}_N=\mathfrak{P}_N^\infty(\mathbf{H})$ is invariant in law under conjugation by permutation matrices. We then have the following result, see \cite{assiotis2024exchangeablearraysintegrablesystems} for more details.
\begin{prop}[\cite{assiotis2024exchangeablearraysintegrablesystems}]
\label{thm:agkw} Let $\mu \in \mathcal{M}_\mathrm{p}^{\mathrm{Perm}}$ and $\mathbf{H}$
    having law $\mu$, such that for some $m\in \mathbb{N}$, $p\in [1,\infty)$ such that  $\det(\mathbf{H}_k)\in \mathrm{L}^p(\mu)$ for every $k\in [m]$, where $\mathbf{H}_k=\mathfrak{P}_k^\infty(\mathbf{H})$. Then, if we denote by $(\mathsf{x}_1^{(N)},\ldots, \mathsf{x}_N^{(N)})$ the eigenvalues of $\mathbf{H}_N$, we have
    \begin{equation*}
        \frac{1}{N^k} \mathrm{e}_k (\mathsf{x}_1^{(N)},\ldots, \mathsf{x}_N^{(N)})\longrightarrow \mathsf{a}_k
    \end{equation*}
    as $N\to \infty$, $\mu$-almost surely and in $\mathrm{L}^p(\mu)$, for some random variables $\{\mathsf{a}_k\}_{k=1}^m$. Furthermore, for $h_1,\ldots,h_m\in \mathbb{N}\cup\{0\}$ with $h_{m}\neq 0$, we have the formula, with $Q= \sum_k kh_k$,
   \begin{equation}\label{eq:exchformula}
\mathbb{E}\left[\prod_{k=1}^m \mathsf{a}_{k}^{h_k}\right]=\frac{(1!)^{h_1}\cdots(m!)^{h_m}}{Q!}\sum_{j=m}^{Q} (-1)^{j+Q} {Q \choose j} \mathbb{E}\left[\prod_{k=1}^m \left(\mathrm{e}_{k} \left({\mathsf{x}}_1^{(j)},{\mathsf{x}}_2^{(j)},\ldots, {\mathsf{x}}_j^{(j)}\right)\right)^{h_k}\right].
\end{equation}
\end{prop}
We now prove Proposition \ref{thm:laguerre}.

\begin{proof}[Proof of Proposition \ref{thm:laguerre}] 
First consider the matrix Laguerre probability measure on $\mathbb{H}_+(N)$, with $a>-1$, defined by
\begin{equation*}
\mathfrak{L}_N^{(a)}(\mathrm{d}\mathbf{H})=\frac{1}{\tilde{\mathcal{Z}}_N^{(a)}}\det\left(\mathbf{H}\right)^a\exp(-\mathrm{Tr}(\mathbf{H}))\mathbf{1}_{\mathbf{H}\in \mathbb{H}_+(N)}\mathrm{d}\mathbf{H},
\end{equation*}
for an explicit constant $\tilde{\mathcal{Z}}_N^{(a)}$, see \cite{ForresterBook}. By Weyl's integration formula the eigenvalues in $\mathbb{W}_N$ of an $\mathfrak{L}_N^{(a)}$-distributed random matrix have law $\nu_N^{(a)}$. Consider the inverse Laguerre probability measures $\mathfrak{IL}_N^{(a)}$ on $\mathbb{H}_+(N)$ to be defined as the pushforward of $\mathfrak{L}_N^{(a)}$ under the map $\mathbf{H}\mapsto \mathbf{H}^{-1}$. Observe that, the eigenvalues in $\mathbb{W}_N$ of an $\mathfrak{IL}_N^{(a)}$-distributed matrix have the same law $((\mathsf{y}_{N-i+1}^{(N)})^{-1})_{i=1}^N$ with $(\mathsf{y}_i^{(N)})_{i=1}^N$ being $\nu_N^{(a)}$-distributed. It is known, see \cite{assiotiswishart}, that $(\mathfrak{P}_N^{N+1})_*\mathfrak{IL}_{N+1}^{(a)}=\mathfrak{IL}_N^{(a)}$. Moreover, observe that for all $N\in \mathbb{N}$, $\mathfrak{IL}_N^{(a)}$ is permutation invariant as it is in fact unitarily invariant.
Hence, there exists a unique $\mathfrak{IL}_\infty^{(a)}\in \mathcal{M}_\mathrm{p}^{\mathrm{Perm}}$ such that $(\mathfrak{P}_N^\infty)_*\mathfrak{IL}_\infty^{(a)}=\mathfrak{IL}_N^{(a)}$, for all $N \in \mathbb{N}$. In particular we can apply Proposition \ref{thm:agkw} to $\mathfrak{IL}_\infty^{(a)}$.

On the other hand, by combining the results of \cite{assiotiswishart,assiotis2020boundary} we obtain, as $N \to \infty$,
\begin{equation*}
\left(N^{-1}\sum_{i=1}^N\frac{1}{\mathsf{y}_i^{(N)}}, N^{-2}\sum_{i=1}^N\frac{1}{(\mathsf{y}_i^{(N)})^{2}},\dots,N^{-m}\sum_{i=1}^N\frac{1}{(\mathsf{y}_i^{(N)})^{m}}\right) \convd \left(\sum_{i=1}^\infty\frac{1}{\mathfrak{B}_a^i}, \sum_{i=1}^\infty\frac{1}{(\mathfrak{B}_a^{i})^2},\dots,  \sum_{i=1}^\infty\frac{1}{(\mathfrak{B}_a^{i})^m}\right)
\end{equation*}
and in particular, from Newton's identities, as $N \to \infty$,
\begin{equation}
\label{eq:laguerreasconv}
  \left(N^{-1}\mathrm{e}_1 \left(\frac{1}{\mathsf{y}_1^{(N)}},\ldots, \frac{1}{\mathsf{y}_N^{(N)}}\right), \dots,N^{-m}\mathrm{e}_m\left(\frac{1}{\mathsf{y}_1^{(N)}},\ldots, \frac{1}{\mathsf{y}_N^{(N)}}\right)\right)   \convd (\mathfrak{e}_1(a),\dots,\mathfrak{e}_m(a)),
\end{equation}
with the $\mathfrak{e}_k(a)$ defined as before. We then obtain the desired result, namely the formula \eqref{eq:exchformula}, since by virtue of \eqref{eq:laguerreasconv}, one can identify the random variables $\{\mathsf{a}_k\}_{k=1}^m$ appearing in Proposition \ref{thm:agkw} when applied to $\mathfrak{IL}_\infty^{(a)}$, as $\{\mathfrak{e}_k(a)\}_{k=1}^m$.
\end{proof}
\section{Connections to integrable systems}

In this section, we establish a connection between the joint moments of $\mathfrak{e}_{N,1}^{(a,b)},\ldots, \mathfrak{e}_{N,k}^{(a,b)}$ and the $\sigma$-Painlev\'{e} V equation, along with a connection between joint moments of $\mathfrak{e}_{1}^{(a)},\ldots, \mathfrak{e}_{k}^{(a)}$ and the  $\sigma$-Painlev\'{e} III$'$ equation upon taking suitable large $N$ limits. These are the main contents of Theorem \ref{structureforgeneralfinitesizeforfiniteN} and Theorem \ref{mainpainlevethm}, respectively. In fact, these theorems address more general situations, involving the joint moments of $\mathrm{e}^{-t\mathfrak{e}_{N,1}^{(a,b)}},\mathfrak{e}_{N,2}^{(a,b)},\ldots, \mathfrak{e}_{N,k}^{(a,b)}$ and $\mathrm{e}^{-t\mathfrak{e}_{1}^{(a)}}, \mathfrak{e}_{2}^{(a)},\ldots, \mathfrak{e}_{k}^{(a)}$, respectively.

We begin by representing the left-hand side of formula (\ref{general structure1}) as a determinant, where the column indices of the corresponding Hankel matrix entries are shifted by integers.
\begin{prop}\label{proprepresented as a Hankel}
Let $N\geq 1$, $k\geq 2$ and $n_{2},\ldots,n_{k}\geq 0$ be integers. Suppose $a>\sum_{q=2}^{k}qn_{q}-1, b>-1$, and denote by $U(\cdot,\cdot;z)$ the confluent hypergeometric function of the second kind. Then, if we define
\begin{align}\label{definition of am}
g_{m}(t_{1})\defeq\exp(-t_{1})U(b+1,-a-2N+2+m;t_{1}),
\end{align} 
we have that for any $t_{1}\in [0,\infty)$,
\begin{align}
\label{combinatotrialexpression}
\mathbb{E}_{N}^{(a,b)}\Bigg[\mathrm{e}^{-\sum_{j=1}^{N}\frac{t_{1}}{\mathsf{x}_{j}^{(N)}}}\prod_{q=2}^{k}&\Big(\sum_{j=1}^{N}(\mathsf{x}_{j}^{(N)})^{-q}\Big)^{n_{q}}\Bigg]\nonumber
\\&=
C_{N}^{(a,b)}\sum_{\substack{\sum_{j=1}^{N}l_{m,j}=n_{m}\\m=2,\ldots,k}}\prod_{m=2}^{k}\binom{n_{m}}{l_{m,1},\ldots,l_{m,N}}
\det_{0\leq i,j\leq N-1}\left(g_{i+j+\sum_{m=2}^{k}ml_{m,j+1}}(t_{1})\right),
\end{align}
where
\begin{align*}
\label{definitioncnsnew}
C_{N}^{(a,b)}=\frac{(\Gamma(b+1))^{N}}{\mathcal{Z}_{N}^{a,b}}.
\end{align*}
\end{prop}
\begin{proof}
First, observe that by the definition of $\mathbb{E}_{N,2}^{(a,b)}$,
\begin{align*}
\mathbb{E}_{N}^{(a,b)}\Bigg[\mathrm{e}^{-\sum_{j=1}^{N}\frac{t_{1}}{\mathsf{x}_{j}^{(N)}}}\prod_{q=2}^{k}\Big(\sum_{j=1}^{N}&(\mathsf{x}_{j}^{(N)})^{-q}\Big)^{n_{q}}\Bigg]\\
=&
\frac{1}{N!\mathcal{Z}_{N}^{a,b}} \int_{[0,1]^N} \prod_{j=1}^N \mathrm{e}^{- \frac{t_{1}}{x_{j}}}\prod_{q=2}^k \Big(\sum_{j=1}^N (x_j)^{-q}\Big)^{n_{q}} \prod_{j=1}^{N}x_{j}^{a}(1-x_{j})^{b}\Delta^2(\mathbf{x})\mathrm{d}\mathbf{x}.
\end{align*}
We now replace the Vandermonde determinant above by the Vandemonde determinant of  differential operators. To be more precise, defining
\begin{align*}
\Delta\bigg(\frac{\mathrm{d}}{\mathrm{d}w_{1}},\ldots, \frac{\mathrm{d}}{\mathrm{d}w_{N}}\bigg)\defeq \prod_{1\leq i<j\leq N}\bigg(\frac{\mathrm{d}}{\mathrm{d}w_{i}}-\frac{\mathrm{d}}{\mathrm{d}w_{j}}\bigg)=\det_{1\leq i,j\leq N}\bigg(\frac{\mathrm{d}^{j-1}}{\mathrm{d}w_{i}^{j-1}}\bigg),
\end{align*}
we have
\bea \label{formulamay7th1}
&& \hspace{-.4cm} \mathbb{E}_{N}^{(a,b)}\Bigg[\mathrm{e}^{-t_{1}\sum_{j=1}^{N}\frac{1}{\mathsf{x}_{j}^{(N)}}}\prod_{q=2}^{k}\Big(\sum_{j=1}^{N}(\mathsf{x}_{j}^{(N)})^{-q}\Big)^{n_{q}}\Bigg] \\
&\hspace{-.6cm}=& \hspace{-.5cm} \frac{1}{N!\mathcal{Z}_{N}^{(a,b)}} \Delta^2\big(\frac{\di}{\di w_{1}},\ldots, \frac{\di}{\di w_{N}}\big)\underbrace{\int_{[0,1]^N} \prod_{j=1}^N \mathrm{e}^{- \frac{w_{j}}{x_j}}\prod_{q=2}^k \Big(\sum_{j=1}^N x_j^{-q}\Big)^{n_{q}}\prod_{j=1}^{N}x_{j}^{a+2N-2}(1-x_{j})^{b} \di \mathbf{x}}_{\defeq h(w_1,\ldots,w_N)}\Big|_{w_{1}=\cdots=w_{N}=t_{1}}  \nonumber
\eea
\noindent In particular, expanding the power in the product over $q$, we get
\begin{equation*}
\label{May7thformula2}
h(w_{1},\ldots,w_{N})=\sum_{\substack{\sum_{j=1}^{N}l_{m,j}=n_{m}\\m=2,\ldots,k}}
\prod_{m=2}^{k}\binom{n_{m}}{l_{m,1},\ldots,l_{m,N}}\int_{[0,1]^N} \prod_{j=1}^N \mathrm{e}^{-\frac{w_{j}}{x_j}}x_{j}^{-\sum_{m=2}^{k}ml_{m,j}}x_{j}^{a+2N-2}(1-x_{j})^{b}\mathrm{d} \mathbf{x}.
\end{equation*}
Hence, by the symmetry of $h(w_{1},\ldots,w_{N})$ and the definition of the Vandermonde determinant of differential operators, we have
\bea
&& \Delta^2\left(\frac{\mathrm{d}}{\mathrm{d}w_{1}},\ldots, \frac{\mathrm{d}}{\mathrm{d}w_{N}}\right)h(w_{1},\ldots,w_{N})\Big|_{w_{1}=\ldots=w_{N}=t_{1}}\label{May7thformula3} \\
&=&N!(-1)^{\sum_{m=2}^{k}mn_{m}}\sum_{\substack{\sum_{j=1}^{N}l_{m,j}=n_{m}\\m=2,\ldots,k}}
\prod_{m=2}^{k}\binom{n_{m}}{l_{m,1},\ldots,l_{m,N}}
\det_{0\leq i,j\leq N-1}\left(\frac{\mathrm{d}^{i+j+\sum_{m=2}^{k}ml_{m,j+1}}}{\mathrm{d}w^{i+j+\sum_{m=2}^{k}ml_{m,j+1}}}f(w)\Big|_{w=t_{1}}\right)
\nonumber,
\eea
where
\begin{align*}
f(w)=\int_{0}^{1}\mathrm{e}^{-\frac{w}{x}}x^{a+2N-2}(1-x)^{b}\mathrm{d}x.
\end{align*}
Note that for $b>-1$, making a change of variables one can compute
\begin{align*}
f(w)=\mathrm{e}^{-w}\int_{0}^{\infty}\mathrm{e}^{-wz}z^{b}(z+1)^{-a-2N-b}\mathrm{d}z=\Gamma(b+1)\mathrm{e}^{-w}U(b+1,-a-2N+2;w).
\end{align*}
Moreover, by \cite[p.258 formula (14)]{hammer1953higher},
\bea \label{differentialrelations}
\frac{\mathrm{d}^m}{\mathrm{d}w^m} \Big( \mathrm{e}^{-w}U(b+1, -a-2N+2; w) \Big)
=(-1)^m \mathrm{e}^{-w}U(b+1, -a-2N+2+m; w).
\eea
Combining these, the desired expression follows.
\end{proof}
The above proposition provides a Hankel determinant representation for $\mathbb{E}_{N}^{(a,b)}\Big[\mathrm{e}^{-\sum_{j=1}^{N}\frac{t_{1}}{\mathsf{x}_{j}^{(N)}}}\Big]$. As an application,  we deduce the Toda lattice equation satified by $\mathbb{E}\left[\mathrm{e}^{-t\mathfrak{e}_1(a)}\right]$, which is given as follows.
\begin{prop}\label{a recurrence relation}
Let $a\in \mathbb{R}$ with $a>1$. Then for $t\geq 0$,
\beas
\mathbb{E}\left[\mathrm{e}^{-t\mathfrak{e}_1(a-2)}\right]\mathbb{E}\left[\mathrm{e}^{-t\mathfrak{e}_1(a+2)}\right]
=a^2(a^2-1)\Bigg(\frac{\di^2\mathbb{E}\left[\mathrm{e}^{-t\mathfrak{e}_1(a)}\right]}{\di t^2} \mathbb{E}\left[\mathrm{e}^{-t\mathfrak{e}_1(a)}\right]
-\Bigg(\frac{\di \mathbb{E}\left[\mathrm{e}^{-t\mathfrak{e}_1(a)}\right]}{\di t} \Bigg)^2\Bigg).
\eeas
\end{prop}
\begin{proof}
Let $g_{m}(t_{1})$ be given as in (\ref{definition of am}). Let 
\beas
\mathcal{B}_{N}=\left(g_{i+j}(t_{1})\right)_{0\leq i,j\leq N-1}.
\eeas
By Proposition \ref{proprepresented as a Hankel}, we have for $a>-1$, $b>-1$,
\bea\label{equation772}
\mathbb{E}_{N}^{(a,b)}\Bigg[\mathrm{e}^{-\sum_{j=1}^{N}\frac{t_{1}}{\mathsf{x}_{j}^{(N)}}}\Bigg]=C_{N}^{(a,b)}(\det\mathcal{B}_{N})(t_{1}).
\eea
Note that, $g_{m}(t_{1})$ depends on the parameters $N,a,b$. Hence,
\bea\label{equation771}
\mathbb{E}_{N+1}^{(a-2,b)}\Bigg[\mathrm{e}^{-\sum_{j=1}^{N+1}\frac{t_{1}}{\mathsf{x}_{j}^{(N)}}}\Bigg]=C_{N+1}^{(a-2,b)}(\det\mathcal{B}_{N+1})(t_{1})
,\quad
\mathbb{E}_{N-1}^{(a+2,b)}\Bigg[\mathrm{e}^{-\sum_{j=1}^{N-1}\frac{t_{1}}{\mathsf{x}_{j}^{(N)}}}\Bigg]=C_{N-1}^{(a+2,b)}(\det\mathcal{B}_{N-1})(t_{1}),
\eea
hold for $a>1,b>-1$.
For $m\geq 1$, denote $\mathcal{B}_{N}^{(\{i_{1},\ldots,i_{m}\},\{j_{1},j_{2},\ldots,j_{m}\})}$ by the matrix $\mathcal{B}_{N}$ with rows $i_{1},\ldots,i_{m}$ and the columns $j_{1},\ldots,j_{m}$ deleted.
Using (\ref{differentialrelations}),
\beas
\frac{\mathrm{d}g_{m}(t_{1})}{\mathrm{d}t_{1}}=-g_{m+1}(t_1).
\eeas
Then, it is not hard to check that
\bea\label{77formula4}
\det \mathcal{B}_{N+1}^{(\{N\},\{N+1\})}=\det \mathcal{B}_{N+1}^{(\{N+1\},\{N\})}=-\frac{\mathrm{d}(\det \mathcal{B}_{N})(t_{1})}{\mathrm{d}t_{1}}
\eea
and
\bea\label{77formula5}
\det \mathcal{B}_{N+1}^{(\{N\},\{N\})}=\frac{\mathrm{d^2}(\det \mathcal{B}_{N})(t_{1})}{\mathrm{d}t_{1}^2}.
\eea
By the Desnanot-Jacobi identity we thus obtain:
\beas
\det \mathcal{B}_{N+1}\det \mathcal{B}_{N-1}=\det \mathcal{B}_{N+1}^{(\{N\},\{N\})}\det \mathcal{B}_{N}-\det \mathcal{B}_{N+1}^{(\{N\},\{N+1\})}\det \mathcal{B}_{N+1}^{(\{N+1\},\{N\})}.
\eeas
Together with (\ref{equation772})-(\ref{77formula5}), 
we have for $a>1$, $b>-1$,
\begin{align}\label{77formula7}
&\mathbb{E}_{N+1}^{(a-2,b)}\Bigg[\mathrm{e}^{-\sum_{j=1}^{N+1}\frac{t_{1}}{\mathsf{x}_{j}^{(N)}}}\Bigg]\mathbb{E}_{N-1}^{(a+2,b)}\Bigg[\mathrm{e}^{-\sum_{j=1}^{N-1}\frac{t_{1}}{\mathsf{x}_{j}^{(N)}}}\Bigg]\frac{(C_{N}^{(a,b)})^2}{C_{N+1}^{(a-2,b)}C_{N-1}^{(a+2,b)}}\nonumber\\
=&\frac{\mathrm{d^2}}{\mathrm{d}t_{1}^2}\mathbb{E}_{N}^{(a,b)}\Bigg[\mathrm{e}^{-\sum_{j=1}^{N}\frac{t_{1}}{\mathsf{x}_{j}^{(N)}}}\Bigg]\mathbb{E}_{N}^{(a,b)}\Bigg[\mathrm{e}^{-\sum_{j=1}^{N}\frac{t_{1}}{\mathsf{x}_{j}^{(N)}}}\Bigg]-\Bigg(\frac{\mathrm{d}}{\mathrm{d}t_{1}}\mathbb{E}_{N}^{(a,b)}\Bigg[\mathrm{e}^{-\sum_{j=1}^{N}\frac{t_{1}}{\mathsf{x}_{j}^{(N)}}}\Bigg]\Bigg)^2.
\end{align}
Now, after doing scaling $t_{1}\mapsto \frac{t_{1}}{N^2}$ in equation \eqref{77formula7}, and by 
\beas
\lim_{N\rightarrow \infty}\frac{\mathrm{d}^{m}}{\mathrm{d}t_{1}^{m}}\mathbb{E}_{N}^{(a,b)}\Bigg[\mathrm{e}^{-\sum_{j=1}^{N}\frac{t_{1}}{N^2\mathsf{x}_{j}^{(N)}}}\Bigg]=\frac{\mathrm{d}^{m}}{\mathrm{d}t_{1}^{m}}\mathbb{E}\Bigg[\mathrm{e}^{-t_{1}\mathfrak{e}_{1}(a)}\Bigg],
\eeas
for $a>1$, $b>-1$ and $m=0,1,2$,
we obtain the equation claimed in this proposition.
\end{proof}

\subsection{Ideas behind the proofs of Theorem \ref{structureforgeneralfinitesizeforfiniteN} and Theorem \ref{mainpainlevethm}}\label{ideas}

By Proposition \ref{proprepresented as a Hankel}, to prove Theorem \ref{structureforgeneralfinitesizeforfiniteN}, it suffices to show that (\ref{combinatotrialexpression}), the linear combinations of determinants of the form $\det_{0\leq i, j\leq N-1}\Big(g_{i+j+m_{j}}(t_{1})\Big)$ with the associated combinatorial coefficients, possess the required structure,
namely, that they can be expressed as linear combinations (with polynomial coefficients in $t_1$ 
of suitable-order derivatives) of the Hankel determinant $\det_{0\leq i, j\leq N-1}\Big(g_{i+j}(t_{1})\Big)$.  

Inspired by the idea used in the authors' previous work \cite[Section 4]{assiotis2024exchangeablearraysintegrablesystems}, we tried lifting formula (\ref{combinatotrialexpression}) to a setting involving partial derivatives of a ``large''
Hankel determinant, whose entries are constructed as generating functions of $g_{n}(t_{1})$ with respect to $t_{2},\ldots,t_{k}$, evaluated at $(0,\ldots,0)$. That is, 
\bea
(\ref{combinatotrialexpression})
=
C_{N}^{(a,b)} 
 \prod_{q=2}^k \frac{\partial^{n_q}}{\partial t_q^{n_{q}}}\Psi_N(t_1,\ldots,t_k)\Bigg|_{t_2=\cdots=t_k=0}\label{liftformaly},
\eea
where
\bea\label{lifttoalargeHankel}
\Psi_N(t_1,\ldots,t_k)=\det_{0\leq i,j\leq N-1}\left( \sum_{m_2=0}^{\infty}\cdots\sum_{m_{k}=0}^{\infty} \frac{t_2^{m_2}\ldots t_k^{m_k}}{m_2 !\ldots m_k!} g_{i+j+\sum_{q=2}^k q m_q}(t_1)
\right).
\eea
Unfortunately, the above construction fails due to the non-convergence of the infinite series defining the matrix entries. This divergence arises because, for fixed $t_1$, $g_m(t_1)$ grows as fast as $m!$ as $m\rightarrow \infty$, according to (\ref{definition of am}). So we need to change the strategy. 

To overcome this difficulty, we begin by studying a truncated version of the Hankel determinant appearing on the right-hand side of equation (\ref{lifttoalargeHankel}). Specifically, we consider
\bea\label{thedetermiantwanttostudy}
\Phi_{N}(R_{2},\ldots,R_{k};t_{1},\ldots,t_{k})\defeq\det_{0\leq i,j\leq N-1}\left( \sum_{m_2=0}^{R_{2}}\cdots\sum_{m_{k}=0}^{R_{k}} \frac{t_2^{m_2}\ldots t_k^{m_k}}{m_2 !\ldots m_k!} g_{i+j+\sum_{q=2}^k q m_q}(t_1),
\right)
\eea
where the key difference is that the summation in each matrix entry is finite rather than infinite, and therefore, there is no issue of convergence. 

The motivation behind this truncation approach is from the observation that, with appropriately chosen parameters $R_{2},\ldots,R_{k}$, the basic relation (\ref{liftformaly}) is preserved. Specifically, for any $R_{q}\geq n_{q}$  $q=2,\ldots,k$,
\bea\label{anidentitytr}
&&(\ref{combinatotrialexpression})=
C_{N}^{(a,b)} 
 \prod_{q=2}^k \frac{\partial^{n_q}}{\partial t_q^{n_{q}}}\Phi_N(R_{2},\ldots,R_{k};t_{1},\ldots,t_{k})\Big|_{t_2=\cdots=t_k=0}\label{lift}.
\eea

The price we pay for performing the truncation is that we lose the neat structure observed in the representation of the partial derivatives of the Hankel determinant~\eqref{thedetermiantwanttostudy} as shifted Hankel determinants, as was the case in our previous work~\cite[Section~4.3]{assiotis2024exchangeablearraysintegrablesystems}. Here, "neat" refers to the fact that in the aforementioned reference, the partial derivatives of~\eqref{lifttoalargeHankel} with respect to $t_2, \ldots, t_k $ could be fully expressed in terms of shifted Hankel determinants. In contrast, in our setting, the representation of the partial derivatives of~\eqref{thedetermiantwanttostudy} involves not only shifted Hankel determinants but also additional terms, which we refer to as \emph{error terms}. 

One of the main differences between this paper and~\cite[Section~4.3]{assiotis2024exchangeablearraysintegrablesystems} is that we explicitly compute these error terms. This is carried out in Subsection~\ref{truncation results}, where we express them as multivariate polynomials in $t_2, \ldots, t_k$, with exponents determined by the parameters $ R_2, \ldots, R_k $, and coefficients given in terms of traces of certain matrices.

Another key difference between this section and~\cite[Section~4.3]{assiotis2024exchangeablearraysintegrablesystems} lies in the evaluation of the representations obtained in Subsection \ref{truncation results} for the partial derivatives of the shifted version of \eqref{thedetermiantwanttostudy} at $ t_2 = \cdots = t_k = 0 $. In our case, we must carefully choose appropriate values for \( R_2, \ldots, R_k \), depending on the order of the partial derivatives. On one hand, to ensure that the error terms vanish upon evaluation at this point; on the other hand, to ensure that the remaining terms in the representation have the analogue form as (\ref{combinatotrialexpression}) after evaluation. These resulting quantities are the ones we shall process the induction. We provide the details of the above arguments in Subsection \ref{recursiveformula}.

Since the recursive formulas obtained in Subsection \ref{recursiveformula} share the same structure as those in \cite[Propositions 4.13 and 4.15]{assiotis2024exchangeablearraysintegrablesystems}, we can apply the inductive argument developed therein to prove Theorem \ref{structureforgeneralfinitesizeforfiniteN} in Subsection \ref{proofofmainTheorem}, as well as the limiting case (as $N\to \infty$) stated in Theorem \ref{mainpainlevethm}. It is worth noting that, although the expression for the coefficients in Theorem \ref{structureforgeneralfinitesizeforfiniteN} (the finite $N$ case) depends on both $a$ and $b$, in the limiting case, as one would expect, the coefficients depend only on $a$. This is because the terms involving $b$ do not contribute to the leading order as $N \to \infty$, details for which are provided in the proof of Theorem \ref{mainpainlevethm}.

\subsection{Hankel determinants shifted by partitions}\label{truncation results}
Let $g_{m}(t_{1})$ be given as (\ref{definition of am}). We now focus on the Hankel determinant (\ref{thedetermiantwanttostudy}). Before proceeding with the proof, for the sake of brevity, let us introduce some notation. First, define for any $n\in \mathbb{N}$, the functions
\begin{equation*}
    \fxn_{n}(R_{2},\ldots,R_{k}; t_1,\ldots,t_k)=\sum_{m_2=0}^{R_{2}}\cdots\sum_{m_{k}=0}^{R_{k}} \frac{t_2^{m_2}\ldots t_k^{m_k}}{m_2 !\ldots m_k!} g_{n+\sum_{q=2}^k q m_q}(t_1)
\end{equation*}
Given our specific focus on the large $N$ behavior of the Hankel determinants of these functions, let us also define the scaled versions of these functions as,
\bea\label{definition of scaling h}
\hfxn_{n}(R_{2},\ldots,R_{k};t_1,\ldots,t_k)\defeq\fxn_{n}\left(R_{2},\ldots,R_{k};\frac{t_1}{N^2},\ldots, \frac{t_k}{N^2}\right),
\eea
leading to the main Hankel determinant of consideration:
\bea\label{Hndef}
\Hfxn_{N}(R_{2},\ldots,R_{k};t_{1},\ldots,t_{k})=\det\left(\hfxn_{i+j}(R_{2},\ldots,R_{k};t_{1},\ldots,t_{k})\right)_{i,j=0,\ldots,N-1}.
\eea
Moving on, for a fixed integer partition $\boldsymbol{\lambda}=(\lambda_1,\ldots, \lambda_m)$ with $\lambda_1\geq \cdots \geq \lambda_m\geq 1$, we define
\begin{equation}\label{definitionofANlambda}
\mathbf{A}_{N,\boldsymbol{\lambda}}(R_{2},\ldots,R_{k};t_1,\ldots, t_k)\defeq  \left(\hfxn_{i+j+\lambda_{N-j}}(R_{2},\ldots,R_{k};t_1,\ldots, t_k)\right)_{i,j=0,\ldots,N-1}
\end{equation}
where we use the convention that $\lambda_{m+1}=\cdots=\lambda_N=0$ whenever $m\leq N-1$.
We then define the determinants of these functions as 
\bea\label{generaldefinition of mathcalG}
\Hfxn_{N, \boldsymbol{\lambda}}=\det\mathbf{A}_{N,\boldsymbol{\lambda}}.
\eea
In particular, when $\boldsymbol{\lambda}=\emptyset$, that is $\lambda_{1}=\dots=\lambda_{m}=0$, we have $\Hfxn_{N,\emptyset}=\Hfxn_{N}$. Throughout the rest of this section, the functions $\Hfxn_{N,\boldsymbol{\lambda}}$ will be referred to as a {\em{Hankel determinant shifted by the partition $\boldsymbol{\lambda}$}}. Also related to these functions, we define
\begin{equation*}
\mathbf{A}_{N,\boldsymbol{\lambda}}^{(w)}\defeq\left((i+j+\lambda_{N-j})\boldsymbol{\hfxn}_{i+j+\lambda_{N-j}}(R_{2},\ldots,R_{k};t_1,\ldots, t_k)\right)_{i,j=0,\ldots,N-1}.
\end{equation*}
For a more clear exposition of how Hankel determinants shifted by partitions are interrelated, we introduce a translation map that acts on the set of partitions: For each $\boldsymbol{\lambda}=(\lambda_1,\ldots, \lambda_m)$ with $m\leq N$, define:
\begin{equation*}
   \mathcal{S}_h\boldsymbol{\lambda}=  (\lambda_1+h,\ldots, \lambda_m+h,\underbrace{h,\ldots,h}_{N-m}).
\end{equation*}
As a natural extension of this translation map, for reasons which will become more clear throughout this section, we also define two transformations on Hankel determinants shifted by partitions:
\bea\label{defn of translation operator2}
\Hfxn_{N,\boldsymbol{\lambda},h}(R_{2},\ldots,R_{k};t_1,\ldots,t_k)\defeq \mathrm{Tr}(\mathrm{adj}(\mathbf{A}_{N,\lambda})\mathbf{A}_{N,\mathcal{S}_{h}\lambda}) 
\eea
and 
\bea\label{defn of translation operator3}
\Hfxn_{N,\boldsymbol{\lambda},h}^{(w)}(R_{2},\ldots,R_{k};t_1,\ldots,t_k)\defeq \mathrm{Tr}(\mathrm{adj}(\mathbf{A}_{N,\lambda}){\mathbf{A}}_{N,\mathcal{S}_h \lambda}^{(w)}) 
\eea
where for a matrix $\mathbf{M}$, $\mathrm{adj}(\mathbf{M})$ denotes its adjugate.


It will become clear later, that we will primarily deal with expressions featuring a specific class of partitions, which we define as follows.

\begin{defn}\label{defofhook}
For $1\le q \le n$, we define $\boldsymbol{\lambda}_{n,q}$ to be the integer partition $(n-q+1,1,\ldots,1)$, consisting of $q-1$ parts with value $1$.
\end{defn}

\subsubsection{Identities relating $\Hfxn_{N,\boldsymbol{\lambda}}$, $\Hfxn_{N,\boldsymbol{\lambda},h}$ and $\Hfxn_{N,\boldsymbol{\lambda},h}^{(w)}$}
With the earlier definitions in mind, in this section we aim to gather useful differential and recursive identities for the determinants $\Hfxn_{N}$, and utilize these to arrive at further identities relating the derivatives of $\Hfxn_{N}$ to the shifted determinants $\Hfxn_{N, \boldsymbol{\lambda}}$.

\begin{prop}\label{relationforhlem}
Let $n\in \mathbb{N}$ and let $\hfxn_{n}$ be given as (\ref{definition of scaling h}). Then, we have that
\bea\label{generalrecursiverelationfortextcrg2}
\frac{\partial \hfxn_{n}}{\partial t_1}(R_{2},\ldots,R_{k};t_{1},\ldots,t_{k})
&=&-\frac{1}{N^2}\hfxn_{n+1}(R_{2},\ldots,R_{k};t_{1},\ldots,t_{k}).
\eea
Moreover, whenever $R_{i} \in \mathbb{N}$ for $i=2,\ldots,k$, we have
\bea\label{genralrecursiverelationfortextcrg1}
\hfxn_{n+2}&&= \frac{N^2(a+2N-n+b-1)}{t_1+2t_2}\hfxn_{n}
+\left( \frac{(-a-2N+n+2)N^2}{t_1+2t_2} + \frac{t_1}{t_1+2t_2} \right)\hfxn_{n+1}\nonumber \\
&& +\frac{1}{t_1+2t_2} \sum_{q=3}^k ((q-1)t_{q-1} - qt_q)\hfxn_{n+q}+ \frac{k t_k}{t_1+2t_2} \hfxn_{n+1+k}\nonumber\\
&&-\frac{1}{t_{1}+2t_{2}}\sum_{q=2}^{k}\frac{qt_{q}^{R_{q}+1}}{N^{2R_{q}}R_{q}!}(\hfxn_{n+qR_{q}+q+1}-\hfxn_{n+qR_{q}+q})\Big|_{t_{q}=0}.
\eea
\end{prop}
\begin{proof}
    The result is a mere consequence of the following identities for the confluent hypergeometric functions (see, e.g., \cite[p.257 formula (5) and p.258 formula (10)]{hammer1953higher})

    \beas
(b-a-1)U(a,b-1;t)-(b-1+t)U(a,b;t)+tU(a,b+1;t)=0,
\eeas
and
\beas
\frac{\mathrm{d}U(a,b;t)}{\mathrm{d}t}=U(a,b;t)-U(a,b+1;t).
\eeas
\end{proof}
\begin{lem}\label{generalexpansiondependingonrecursive}
Let $k\geq 2$, $m\geq 3$, $R_2,\ldots,R_k\in \mathbb{N}$, and fix a  partition $\boldsymbol{\lambda}$. Then, we have that
\bea\label{generalizationiterationformula}
\Hfxn_{N,\boldsymbol{\lambda},m}
&=& \frac{(a+2N+b-1)N^2}{t_1+2t_2} \Hfxn_{N,\boldsymbol{\lambda},m-2}
- \frac{N^2}{t_1+2t_2}\Hfxn_{N,\boldsymbol{\lambda},m-2}^{(w)}\nonumber \\
&& +\, \frac{N^2(-a-2N+1)}{t_1+2t_2}\Hfxn_{N,\boldsymbol{\lambda},m-1}
+ \frac{N^2}{t_1+2t_2}\Hfxn_{N,\boldsymbol{\lambda},m-1}^{(w)}\nonumber\\
&& +\, \frac{t_1}{t_1+2t_2} \Hfxn_{N,\boldsymbol{\lambda},m-1}
+ \frac{kt_k}{t_1+2t_2} \Hfxn_{N,\boldsymbol{\lambda}, m+k-1}\nonumber\\
&&+\frac{1}{t_{1}+2t_{2}}\sum_{q=3}^{k}\big((q-1)t_{q-1}-qt_{q}\big)\Hfxn_{N,\boldsymbol{\lambda},m+q-2}\nonumber\\
&+&\frac{1}{t_{1}+2t_{2}}\sum_{q=2}^{k}\frac{t_{q}^{R_{q}+1}}{N^{2R_{q}}}\mathcal{E}_{q,\boldsymbol{\lambda},m}^{(1)}(t_{1},\ldots,t_{k};R_{q}),
\eea
where
\bea\label{definitionoferror1}
\mathcal{E}_{q,\boldsymbol{\lambda},m}^{(1)}(t_{1},\ldots,t_{k};R_{q})=-\frac{q}{R_{q}!}\mathrm{Tr}\left(\mathrm{adj}(\mathbf{A}_{N,\boldsymbol{\lambda}})\left(\mathbf{A}_{N,\mathcal{S}_{qR_{q}+q-1+m}\boldsymbol{\lambda}}\Big|_{t_{q}=0}-\mathbf{A}_{N,\mathcal{S}_{qR_{q}+q-2+m}\boldsymbol{\lambda}}\Big|_{t_{q}=0}\right)\right),
\eea
with $\mathbf{A}_{N,\boldsymbol{\lambda}}$ given in (\ref{definitionofANlambda}).
\end{lem}
\begin{proof}
This is a direct consequence of the recursive identity (\ref{genralrecursiverelationfortextcrg1}) for $\hfxn_{i+j+\lambda_{N-j}}$.
\end{proof}
\subsubsection{Connecting $\Hfxn_{N,\boldsymbol{\lambda},q}$ to the partial derivatives of $\Hfxn_{N,\boldsymbol{\lambda}}$}
\begin{prop}
\label{generalderivative and translations}
Let $N\geq 1$, $k\geq 2$ be integers.  Let $R_{2},\ldots, R_{k}$ be positive integers. Let $\boldsymbol{\lambda}$ be an integer partition. Let $\Hfxn_{N,\boldsymbol{\lambda}}(t_{1},\ldots,t_{k})$ and $\Hfxn_{N,\boldsymbol{\lambda},m}(t_{1},\ldots,t_{k})$ be given in 
(\ref{generaldefinition of mathcalG}) and (\ref{defn of translation operator2}), respectively, then we have for $q=2,\ldots,k$,
\bea\label{corresponding relation}
\frac{\partial \Hfxn_{N,\boldsymbol{\lambda}}}{\partial t_{q}}=\frac{1}{N^2}\Hfxn_{N,\boldsymbol{\lambda},q}+\frac{t_{q}^{R_{q}}}{N^{2R_{q}+2}}
\mathcal{E}_{q,\boldsymbol{\lambda}}^{(2)}(t_{1},\ldots,t_{k};R_{q}),
\eea
where
\bea\label{definitionoferror2}
\mathcal{E}_{q,\boldsymbol{\lambda}}^{(2)}(t_{1},\ldots,t_{k};R_{q})=-\frac{1}{R_{q}!}\mathrm{Tr}\left(\mathrm{adj}(\mathbf{A}_{N,\boldsymbol{\lambda}})\left(\mathbf{A}_{N,\mathcal{S}_{qR_{q}+q}\boldsymbol{\lambda}}\Big|_{t_{q}=0}\right)\right),
\eea
and
\bea\label{t1operateongeneralG}
\frac{\partial \Hfxn_{N,\boldsymbol{\lambda}}}{\partial{t_{1}}}=-\frac{1}{N^2}
\Hfxn_{N,\boldsymbol{\lambda},1}.
\eea
\end{prop}

\begin{proof}
Write $\boldsymbol{\lambda}$ as $(\lambda_{1},\ldots,\lambda_{m})$ with $\lambda_{1}\geq \ldots \geq \lambda_{m}\geq 1$, and set $\lambda_{j}=0$ for $m+1\leq j\leq N$ when $m<N$. Note that for $q=2,\ldots,k$,
\begin{align*}
\frac{\partial \fxn_{n}(R_{2},\ldots,R_{k};t_{1},\ldots,t_{k})}{\partial t_{q}}
=&\fxn_{n+q}(R_{2},\ldots,R_{k};t_{1},\cdots,t_{k})\\
&-\frac{t_{q}^{R_{q}}}{R_{q}!}\fxn_{n+qR_{q}+q}(R_{2},\ldots,R_{k};t_{1},\ldots,t_{q-1},0,t_{q+1},\ldots,t_{k}).
\end{align*}
Then by the above, claim (\ref{corresponding relation}) comes from the following identity:
\bea
\frac{\partial \Hfxn_{N,\boldsymbol{\lambda}}}{\partial t_{q}}
&=& \frac{1}{N^2}\sum_{\substack{s_0+\cdots+s_{N-1}=1 \\s_0,\ldots,s_{N-1} \geq 0}} \det\left(\frac{\partial ^{s_i}\fxn_{i+j+\lambda_{N-j}}}{\partial t_{q}^{s_i}} \left(R_{2},\ldots,R_{k};\frac{t_{1}}{N^2},\ldots,\frac{t_{k}}{N^2}\right)\right)_{i,j=0,\ldots,N-1}. 
\eea
Claim (\ref{t1operateongeneralG}) comes from the recursive relation $(\ref{generalrecursiverelationfortextcrg2})$.
\end{proof}
\begin{prop}\label{expressionsforlength2}
We have
\beas
\Hfxn_{N,\boldsymbol{\lambda}_{2,1}}&=&\frac{N^4}{2}\frac{\partial^2\Hfxn_{N}}{\partial t_{1}^2}+\frac{N^2}{2}\frac{\partial \Hfxn_{N}}{\partial t_{2}}-\frac{t_{2}^{R_{2}}}{2N^{2R_{2}}}
\mathcal{E}_{2,\emptyset}^{(2)}(t_{1},\ldots,t_{k};R_{2}),\\
\Hfxn_{N,\boldsymbol{\lambda}_{2,2}}&=&\frac{N^4}{2}\frac{\partial^2\Hfxn_{N}}{\partial t_{1}^2}-\frac{N^2}{2}\frac{\partial \Hfxn_{N}}{\partial t_{2}}+\frac{t_{2}^{R_{2}}}{2N^{2R_{2}}}\mathcal{E}_{2,\emptyset}^{(2)}(t_{1},\ldots,t_{k};R_{2}),
\eeas
where $\boldsymbol{\lambda}_{n,q}$ and $\mathcal{E}_{2,\emptyset}^{(2)}(t_{1},\ldots,t_{k};R_{2})$ are defined in Definition \ref{defofhook} and equation (\ref{definitionoferror2}),respectively.
\end{prop}
\begin{proof}
Immediately from definitions, we have
\[
\begin{pmatrix}
\Hfxn_{N,\boldsymbol{\lambda}_{2,1}} \\
\Hfxn_{N,\boldsymbol{\lambda}_{2,2}}
\end{pmatrix} = 
\begin{pmatrix}
1/2 & 1/2 \\
1/2 & -1/2
\end{pmatrix}
\begin{pmatrix}
\Hfxn_{N,\boldsymbol{\lambda}_{1,1},1} \\
\Hfxn_{N,\emptyset,2}
\end{pmatrix}.
\]
Now, making use of Proposition \ref{generalderivative and translations}, while also noting that $\Hfxn_{N,\emptyset,1}=\Hfxn_{N,\blambda_{1,1}}$,
we have 
\beas
\Hfxn_{N,\lambda_{1,1},1}=N^4\frac{\partial^2\Hfxn_{N}}{\partial t_{1}^2},
\eeas
and
\beas
\Hfxn_{N,\emptyset,2}=N^2\frac{\partial \Hfxn_{N}}{\partial t_{2}}-\frac{t_{2}^{R_{2}}}{N^{2R_{2}}}
\mathcal{E}_{2,\emptyset}^{(2)}(t_{1},\ldots,t_{k};R_{2}).
\eeas
The claim now follows using these in combination.
\end{proof}

\subsubsection{A recursive  relation for $\Hfxn_{N,\boldsymbol{\lambda}_{n,j}}$}
Next we deduce a recursive formula for $\Hfxn_{N,\boldsymbol{\lambda}_{n,j}}$ for $n\geq 3$ and $j=1,\ldots,n$.
Before stating the result, we introduce some notation.
\begin{defn}
    For integers $m \geq 0$, and $q,z$ with $1\leq q\leq z$, we define
\begin{equation*}
    \Hfxn_N^{(m)}[z;q]\overset{\textnormal{def}}{=}\left(\Hfxn_{N,\boldsymbol{\lambda}_{z,q}}, \Hfxn_{N,\boldsymbol{\lambda}_{z,q+1}},\ldots, \Hfxn_{N,\boldsymbol{\lambda}_{z,z}},0,\ldots,0 \right)^\textnormal{T},
\end{equation*}
where the number of $0$'s are determined by $m$.
\end{defn}
The recursive identities for $\Hfxn_{N,\boldsymbol{\lambda}_{n,j}}$ take a more concise form when represented in terms of the vectors given as above, where certain vectors will be given as linear combinations of other such vectors upon multiplication by appropriate matrices. Explicit expressions for these matrices are given in the Appendix. In what follows, for the sake of brevity, we write $\partial_m$ to mean $\frac{\partial}{\partial t_m}$.

\begin{prop}\label{prop:recursivevector}
Let $N,k,l\in \mathbb{N}$ be integers such that $l\geq 3, k\geq 2, N\geq 1$. Let $R_{2},\ldots,R_{k}\in \mathbb{N}$. Then, there exist matrices $\mathbf{A}^{(l)}$ and $\mathbf{P}_m^{(l)}$ for $m=0,1,\ldots,k+1$, among which only $\mathbf{P}_2^{(l)}$ has $N$-dependent entries, such that the following holds:
    \begin{align*}
        \Hfxn_{N}^{(0)}&[l;1]= N^2 \frac{k t_k (-1)^{k} \mathbf{A}^{(l)}}{t_1+2t_2}\partial_1\Hfxn_{N}^{(1)}[l+k-2;k]+ \frac{1}{t_1+2t_2} \sum_{m=1}^{k-2}  \mathsf{d}_{m+1,m+2}  (-1)^m \mathbf{P}_{m+2}^{(l)}   \Hfxn_{N}^{(0)}[l+m;1]  \\
        &+\mathbf{A}^{(l)} \Bigg(N^2\partial_1+\frac{N^2}{t_1+2t_2}\Bigg(
    \sum_{m=3}^k  \mathsf{d}_{m-1,m} \partial_{m-1} + kt_k \partial_k \Bigg) \Bigg) \Hfxn_{N}^{(1)}[l-1;1]
   \\
    &+ \left( \frac{t_1}{t_1+2t_2} \mathbf{P}_1^{(l)} + \frac{N^2}{t_1+2t_2} {\mathbf{P}}_0^{(l)}  \right) \Hfxn_{N}^{(0)}[l-1;1]
+ \frac{N^2}{t_1+2t_2} \mathbf{P}_2^{(l)} \Hfxn_{N}^{(0)}[l-2;1] \\ &+\frac{\mathbf{A}^{(l)}}{t_1+2t_2} \sum_{m=0}^{k-3}\mathsf{d}_{m+2,m+3}(-1)^m (N^2 \partial_1 ) \Hfxn_{N}^{(1)}[l+m;m+2]\\ & + \frac{1}{t_1+2t_2} (-1)^{k-1} k t_k \mathbf{P}_{k+1}^{(l)}\Hfxn_{N}^{(0)}[l+k-1;1]\\& +\frac{\mathbf{A}^{(l)} N^2kt_{k}}{t_{1}+2t_{2}}\sum_{h=2}^{k-1} (-1)^{k+h} \partial_h \Hfxn_{N}^{(1)}[l+k-1-h;k+1-h] \\& + \frac{\mathbf{A}^{(l)} N^2}{t_1+2t_2} \sum_{m=4}^k (-1)^{m-1} \mathsf{d}_{m-1,m} \sum_{h=2}^{m-2} (-1)^h \partial_h \Hfxn_{N}^{(1)}[l+m-2-h;m-h]\\
&+\mathbf{A}^{(l)}\frac{1}{t_{1}+2t_{2}}\sum_{m=2}^{k}\frac{t_{m}^{R_{m}+1}}{N^{2R_{m}}}\tilde{\mathcal{R}}_m^{(1)}-\mathbf{A}^{(l)}\frac{kt_{k}}{t_{1}+2t_{2}}\sum_{h=2}^{k}(-1)^{k+h}\frac{t_{h}^{R_{h}}}{N^{2R_{h}}}\tilde{\mathcal{R}}_{k+1,h}^{(2)}\\
&-\mathbf{A}^{(l)} \Bigg(\frac{1}{t_1+2t_2}\Big(
    \sum_{m=3}^k  (-1)^{m-1}\mathsf{d}_{m-1,m}\sum_{h=2}^{m-1}(-1)^{h} \frac{t_{h}^{R_{h}}}{N^{2R_{h}}}\Big) \Bigg)\tilde{\mathcal{R}}_{m,h}^{(2)},
    \end{align*}
    where we denote
    \begin{equation*}
        \mathsf{d}_{p,q}= \mathsf{d}_{p,q}(t_1,\ldots,t_k)\overset{\textnormal{def}}{=} pt_p-qt_q,
    \end{equation*}
and for $m=2,\ldots,k$,
\beas
\tilde{\mathcal{R}}_m^{(1)} =
(0,\dots,\sum_{h=2}^{j}(-1)^{h}\mathcal{E}^{(1)}_{m,\blambda_{l-h,j+1-h},h},\dots,\mathcal{E}^{(1)}_{m,\emptyset,l})^{T}_{j=2,\ldots,l-1}\in \mathbb{R}^l
\eeas
with $\mathcal{E}_{m,\boldsymbol{\lambda},h}^{(1)}$ given as (\ref{definitionoferror1}),
and 
\beas
\tilde{\mathcal{R}}_{m,h}^{(2)}=(\mathcal{E}^{(2)}_{h,\blambda_{l+m-2-h,m-h}},\dots,\mathcal{E}^{(2)}_{h,\blambda_{l+m-2-h,j+m-2-h}},\ldots,0)_{j=2,\ldots,l}^{T}\in \mathbb{R}^l
\eeas
with
$\mathcal{E}_{h,\boldsymbol{\lambda}}^{(2)}$ given as (\ref{definitionoferror2}).
\end{prop}
\begin{proof}
Applying Lemma \ref{generalexpansiondependingonrecursive} and \cite[Theorem 17]{keating-fei}, 
    \begin{align*}
         \Hfxn_{N}^{(0)}[l;1]=&\mathbf{A}^{(l)}\Bigg(-\mathcal{W}_{0}+\frac{N^2(a+2N+b-1)}{t_{1}+2t_{2}}\mathcal{W}_{2}+\frac{N^2(-a-2N+1)+t_{1}}{t_{1}+2t_{2}}\mathcal{W}_{1}\\&-\frac{N^2}{t_{1}+2t_{2}}\mathcal{W}_{2}^{(w)}\nonumber+\frac{N^2}{t_{1}+2t_{2}}\mathcal{W}_{1}^{(w)}
+\frac{1}{t_1+2t_2} \sum_{m=3}^k \mathsf{d}_{m-1,m} \mathcal{W}_m +
\frac{kt_k}{t_1+2t_2} \mathcal{W}_{k+1}\Bigg)\\
&+\mathbf{A}^{(l)}\frac{1}{t_{1}+2t_{2}}\sum_{m=2}^{k}\frac{t_{m}^{R_{m}+1}}{N^{2R_{m}}}\tilde{\mathcal{R}}_m^{(1)}
    \end{align*}
    where $\mathcal{W}_{0}\in \mathbb{R}^{l}$ is given by $(\Hfxn_{N,\blambda_{l-1,1},1},\dots,\Hfxn_{N,\blambda_{l-1,l-1},1},0)^{T}$, and the vectors $\mathcal{W}_m \in \mathbb{R}^l$, defined whenever $m=1,2,\ldots,k+1$, are given by: 
    \begin{equation*}
        (\mathcal{W}_m)_i \overset{\textnormal{def}}{=} 
    \begin{cases}

        0,  & i=1, \\
             \noalign{\vskip4pt}
         \displaystyle\sum_{h=2}^{i} (-1)^h \Hfxn_{N,\blambda_{l-h,i+1-h},h+\beta(m)-2} &  2\leq i\leq l-1, \\ 
         \noalign{\vskip4pt}
         \Hfxn_{N,\emptyset, l+\beta(m)-2} & i=l,
       \end{cases}
    \end{equation*}
    where $\beta(m)=0$ if $m=2$, and $\beta(m)=m$ otherwise.
Moreover, we define ${\mathcal{W}}_1^{(w)}, {\mathcal{W}}_2^{(w)} \in \mathbb{R}^l$ analogously to $\mathcal{W}_1, \mathcal{W}_2$, where we swap $\Hfxn$ with  ${\Hfxn}^{(w)}$. Moving on, we write each $\mathcal{W}_m$, for $m=3, 4,\ldots,k+1$, as a sum $\mathcal{W}_m = \mathcal{W}_m'+\mathcal{W}_m^* $ where $\mathcal{W}_m'$ and $\mathcal{W}_m^* $ are given by, 
    \begin{align*}
        (\mathcal{W}_m')_j &\overset{\textnormal{def}}{=} 
    \begin{cases}
         \displaystyle\sum_{h=1}^{j+m-2} (-1)^{h+m-2} \Hfxn_{N,\blambda_{l+m-2-h,j+m-1-h},h} &  j=1,\ldots, l-1, \\ 
         \noalign{\vskip4pt}
         \Hfxn_{N,\emptyset, l+m-2} & j=l,
       \end{cases}
       \\
        (\mathcal{W}_m^*)_j &\overset{\textnormal{def}}{=} 
        \begin{cases}
         \displaystyle\sum_{h=1}^{m-1} (-1)^{h+m-1} \Hfxn_{N,\blambda_{l+m-2-h,j+m-1-h},h} &  j=1,\ldots, l-1, \\ 
         \noalign{\vskip4pt}
         0 & j=l.
       \end{cases}
    \end{align*}
Next, we make use of (\ref{t1operateongeneralG}) for $\mathcal{W}_{0}$ and (\ref{corresponding relation}) for $\mathcal{W}_m^*$; while for the other terms, we apply \cite[Theorem 17 and Proposition 20]{keating-fei}, leading to the identities:
\begin{align*}
    \mathbf{A}^{(l)}\mathcal{W}_1=\mathbf{P}_1^{(l)} \Hfxn_{N}^{(0)}[l-1;1]\: &,\;\;\;\;  \mathbf{A}^{(l)}{\mathcal{W}}_1^{(w)}=\widetilde{\mathbf{P}}_1^{(l)} \Hfxn_{N}^{(0)}[l-1;1],\\
    \mathbf{A}^{(l)}\mathcal{W}_2=\widehat{\mathbf{P}}_2^{(l)} \Hfxn_{N}^{(0)}[l-2;1]\: &,\;\;\;\;  \mathbf{A}^{(l)}{\mathcal{W}}_2^{(w)}=\widetilde{\mathbf{P}}_2^{(l)} \Hfxn_{N}^{(0)}[l-2;1],\\ \noalign{\vskip4pt}
    \mathbf{A}^{(l)}\mathcal{W}_m'=  \: \mathbf{P}_m^{(l)} & \Hfxn_{N}^{(0)}[l+m-2;1],\\
    \intertext{for matrices $\widetilde{\mathbf{P}}_{1}^{(l)},\widetilde{\mathbf{P}}_{2}^{(l)},\widehat{\mathbf{P}}_{2}^{(l)}$ whose explicit formulae can also be found in the Appendix. In particular, letting}
    {\mathbf{P}}_{0}^{(l)}\defeq (-a-2N+1)\mathbf{P}_{1}^{(l)}+\widetilde{\mathbf{P}}_{1}^{(l)},& \;\;\; \textnormal{and} \;\;\; \mathbf{P}_{2}^{(l)}\defeq (a+2N+b-1)\widehat{\mathbf{P}}_{2}^{(l)}-\widetilde{\mathbf{P}}_{2}^{(l)},
\end{align*}
we arrive at the desired identities. Here, it is important to observe that although 
 $-2N\mathbf{P}_{1}^{(l)}$ and $\widetilde{\mathbf{P}}_{1}^{(l)}$ both contain linear terms in $N$, the linear combination that gives ${\mathbf{P}}_{0}^{(l)}$ allows for cancellations so that there is no $N$-dependence in the final expression for ${\mathbf{P}}_{0}^{(l)}$.
\end{proof}
At this point, it is worth mentioning that \( \mathbf{P}^{(l)}_m \), for \( m = 0, 1, \ldots, k+1 \) (with \( m \neq 0, 2 \)), are identical to those (denoted by \( \mathbf{Q}^{(l)}_m \)) appearing in \cite[Proposition 4.11]{assiotis2024exchangeablearraysintegrablesystems}.
This is due to our use of results by the third and fourth authors, namely \cite[Theorem 17 and Proposition 20]{keating-fei}, applied to shifts of general Hankel determinants, which hold independently of the specific entries.\\

\noindent Moving on, we will expand $\frac{\partial \Hfxn_N}{\partial t_{q}}$ as a power function of the variables $t_{q},\ldots,t_{k}$, for $q=2,\ldots,k$. The main idea is to repeatedly use (\ref{generalizationiterationformula}) for the right-hand side of (\ref{corresponding relation}), mainly for terms with an increased sub-index every time. For the purposes of writing such formulae in a more concise manner, we additionally introduce:
\begin{align}\label{defofdi}
  &\mathsf{Z}_{i}(R_{2},\ldots,R_{k};t_{1},\ldots, t_{k})\nonumber\\
  =&N^2\left\{(a+2N+b-1) \Hfxn_{N,\emptyset,i}-{\Hfxn}^{(w)}_{N,\emptyset,i} +(1-2N-a)\Hfxn_{N,\emptyset, i+1}
+{\Hfxn}^{(w)}_{N,\emptyset, i+1}+\frac{t_1}{N^2}\Hfxn_{N,\emptyset,i+1}\right\}.
\end{align}
Here, $a,b$ are the same parameters that appeared in the definition of
$\Hfxn_{N,\blambda,h}$ and $\Hfxn^{(w)}_{N,\blambda,h}$.
\begin{lem}\label{generaliedexpandingofpartialderivatives}
Let $N,R_{2},\ldots,R_{k}\in \mathbb{N}$, and suppose $\mathcal{E}_{u,\boldsymbol{\lambda},m}^{(1)}$ and $\mathcal{E}_{q,\boldsymbol{\lambda}}^{(2)}$ are given as in (\ref{definitionoferror1}) and (\ref{definitionoferror2}), respectively. For $k=2$, $m\in \mathbb{N}$, we have that
\bea\label{special case k=2}
N^2\frac{\partial \Hfxn_{N}}{\partial t_{2}}
&=&\frac{\mathsf{Z}_{0}}{t_{1}+2t_{2}}+
\sum_{i=2}^{m}\frac{\mathsf{Z}_{i-1}(2t_{2})^{i-1}}{(t_{1}+2t_{2})^{i}}+\frac{(2t_{2})^m}{(t_{1}+2t_{2})^{m}}\Hfxn_{N,\emptyset,m+2}\nonumber\\&&+\frac{t_{2}^{R_{2}}}{N^{2R_{2}}}\mathcal{E}_{2,\emptyset}^{(2)}+\frac{t_{2}^{R_{2}+1}}{N^{2R_{2}}}\sum_{i=1}^{m}\frac{(2t_{2})^{i-1}\mathcal{E}^{(1)}_{2,\emptyset,i+1}}{(t_{1}+2t_{2})^{i}},
\eea
and for $k\geq 3$, $2\leq q\leq k$ and $m\in \mathbb{N}$, we have that
\begin{align*}
&N^2\frac{\partial \Hfxn_{N}}{\partial t_{q}}=\sum_{i=2}^{q-1}\frac{t_{i}}{(t_1+2t_2)^m}V_{i}(t_{1},\ldots,t_{k})+\frac{\mathsf{Z}_{q-2}}{t_1+2t_2}+\sum_{i=2}^m \frac{1}{(t_1+2t_2)^i} \sum_{j=0}^{i-1} \binom{i-1}{j} (kt_k)^j \nonumber \\
&
\sum_{\substack{h_{2}+\cdots+h_{k}=i-1-j\\h_{2}=\cdots=h_{q-1}=0\\h_{q}\geq 0,\ldots,h_{k}\geq 0}}\prod_{n=q}^{k}(nt_{n})^{h_{n}} 
\Bigg(\sum_{\substack{h_3' \leq h_3 \\ \ldots \\ h_{k-1}'\leq h_{k-1}}}c_{h_{2},h_{3}',\ldots,h_{k-1}',h_{k}}^{(i,j)}\Big(\mathsf{Z}_{q+\sum_{n=1}^{k-2}nh_{n+1}+(k-2)h_{k}-\sum_{n=3}^{k-1}h_{n}'+(k-1)j-2}\\
&\;\;\;\; +\sum_{u=2}^{k}\frac{t_{u}^{R_{u}+1}}{N^{2R_{u}}}\mathcal{E}^{(1)}_{u,\emptyset,q+\sum_{n=1}^{k-2}nh_{n+1}+(k-2)h_{k}-\sum_{n=3}^{k-1}h_{n}'+(k-1)j}\Big)\Bigg) + \frac{1}{(t_1+2t_2)^m} \nonumber \\
 & \;\;\;\; \sum_{j=0}^{m} \binom{m}{j} (kt_k)^j
\sum_{\substack{h_{2}+\cdots+h_{k}=m-j\\h_{2}=\cdots=h_{q-1}=0\\h_{q}\geq 0,\ldots,h_{k}\geq 0}}\prod_{n=q}^{k}(nt_{n})^{h_{n}}\Bigg(\sum_{\substack{h_3' \leq h_3 \\ \ldots \\ h_{k-1}'\leq h_{k-1}}}c_{h_{2},h_{3}',\ldots,h_{k-1}',h_{k}}^{(m+1,j)}\Hfxn_{N,\emptyset, q+\sum_{n=1}^{k-2}nh_{n+1}+(k-2)h_{k}-\sum_{n=3}^{k-1}h_{n}'+(k-1)j}\Bigg)\\
& \;\;\;\; +\frac{t_{q}^{R_{q}}}{N^{2R_{q}}}\mathcal{E}_{q,\emptyset}^{(2)}+\frac{1}{t_{1}+2t_{2}}\sum_{u=2}^{k}\frac{t_{u}^{R_{u}+1}}{N^{2R_{u}}}\mathcal{E}_{u,\emptyset,q}^{(1)}.
\end{align*}
Here, $V_{2},\ldots,V_{q-1}$ are $C^\infty$ functions in the variables $t_{1},\ldots,t_{k}$, explicit expressions for which are ommitted here; and 
\bea\label{coefficientsbeforegeneralD}
&&c_{h_{2},h_{3}',\ldots,h_{k-1}',h_{k}}^{(i,j)}=\frac{(i-1-j)!(-1)^{h_{3}'+\cdots+h_{k-1}'+h_{k}}}{h_2! h_k! \prod_{n=3}^{k-1}(h_{n}-h_{n}')!h_n'!}.
\eea
\end{lem}

\begin{proof}
First, we let $\boldsymbol{\lambda}=\emptyset$, and using \eqref{generalizationiterationformula}, we repeatedly plug in the expression therein for $\Hfxn_{N,\boldsymbol{\lambda},l+q-2}$ as $l$ ranges from $3$ to $k+1$. Using this in combination with Proposition \ref{generalderivative and translations}, we obtain for $q=2,\ldots,k, m\geq 1$,
\beas
N^2\frac{\partial \Hfxn_{N}}{\partial t_{q}}
&=& \Hfxn_{N,\emptyset,q}+\frac{t_{q}^{R_{q}}}{N^{2R_{q}}}\mathcal{E}_{q,\emptyset}^{(2)}\\
&=&\frac{\mathsf{Z}_{q-2}}{t_1+2t_2}+\frac{t_{q}^{R_{q}}}{N^{2R_{q}}}\mathcal{E}_{q,\emptyset}^{(2)}
+\sum_{i=2}^m \frac{1}{(t_1+2t_2)^i} \sum_{j=0}^{i-1} \binom{i-1}{j} (kt_k)^j\\ &&\sum_{l_1=1}^{k-2} \sum_{l_2=1}^{k-2} \cdots \sum_{l_{i-1-j}=1}^{k-2}\Big(\prod_{n=1}^{i-1-j}\big((l_{n}+1)t_{l_{n}+1}-(l_{n}+2)t_{l_{n}+2}\big)\Big)\mathsf{Z}_{q+l_1+\cdots+l_{i-1-j}+(k-1)j-2} \\
&&+\sum_{u=2}^{k}\frac{t_{u}^{R_{u}+1}}{N^{2R_{u}}}\sum_{i=2}^m \frac{1}{(t_1+2t_2)^i} \sum_{j=0}^{i-1} \binom{i-1}{j} (kt_k)^j\sum_{l_1=1}^{k-2}\sum_{l_2=1}^{k-2}\\ &&  \cdots \sum_{l_{i-1-j}=1}^{k-2}\Big(\prod_{n=1}^{i-1-j}\big((l_{n}+1)t_{l_{n}+1}-(l_{n}+2)t_{l_{n}+2}\big)\Big)\mathcal{E}_{u,\emptyset,q+l_1+\cdots+l_{i-1-j}+(k-1)j}^{(1)} \\
&&+\frac{1}{t_{1}+2t_{2}}\sum_{u=2}^{k}\frac{t_{u}^{R_{u}+1}}{N^{2R_{u}}}\mathcal{E}_{u,\emptyset,q}^{(1)}+\frac{1}{(t_1+2t_2)^m}  \sum_{j=0}^{m} \binom{m}{j} (kt_k)^j\\
&&\sum_{l_1=1}^{k-2}\sum_{l_2=1}^{k-2}\cdots \sum_{l_{m-j}=1}^{k-2} \Big(\prod_{n=1}^{m-j}\big((l_{n}+1)t_{l_{n}+1}-(l_{n}+2)t_{l_{n}+2}\big)\Big)\Hfxn_{N,\emptyset,q+l_1+\cdots+l_{m-j}+(k-1)j }.
\eeas
In the equation above, to avoid further notational complexity, we assume that when $i=2,\ldots,m+1$, for $j=i-1$, we have
\beas
\sum_{l_1=1}^{k-2}\sum_{l_2=1}^{k-2}\cdots \sum_{l_{i-1-j}=1}^{k-2} \Big(\prod_{n=1}^{i-1-j}\big((l_{n}+1)t_{l_{n}+1}-(l_{n}+2)t_{l_{n}+2}\big)\Big)\mathsf{Z}_{q+l_1+\cdots+l_{i-1-j}+(k-1)j-2}=\mathsf{Z}_{q+(k-1)(i-1)-2},
\eeas
and for $k=2,$ $j<i-1$,
\beas
\sum_{l_1=1}^{k-2}\sum_{l_2=1}^{k-2}\cdots \sum_{l_{i-1-j}=1}^{k-2} \prod_{n=1}^{i-1-j}\big((l_{n}+1)t_{l_{n}+1}-(l_{n}+2)t_{l_{n}+2}\big)\mathsf{Z}_{q+l_1+\cdots+l_{i-1-j}+(k-1)j-2}=0.
\eeas
The above notation also applies to sums involving $\mathcal{E}_{u,\emptyset,l}^{(1)}$ and $\Hfxn_{N,\emptyset,l}$.
Moving on, upon computation, one can easily show that
\begin{align}
&\sum_{l_1=1}^{k-2} \sum_{l_2=1}^{k-2} \cdots \sum_{l_{i-1-j}=1}^{k-2}\Big(\prod_{n=1}^{i-1-j}\big((l_{n}+1)t_{l_{n}+1}-(l_{n}+2)t_{l_{n}+2}\big)\Big)\mathsf{Z}_{q+l_1+\cdots+l_{i-1-j}+(k-1)j-2}\nonumber\\
&=\sum_{\substack{h_{2}+\cdots+h_{k}=i-1-j\\h_{2}\geq 0,\ldots,h_{k}\geq 0}}\prod_{n=2}^{k}(nt_{n})^{h_{n}}\Bigg(\sum_{\substack{h_3' \leq h_3 \\ \ldots \\ h_{k-1}'\leq h_{k-1}}}c_{h_{2},h_{3}',\ldots,h_{k-1}',h_{k}}^{(i,j)}\mathsf{Z}_{q+\sum_{n=1}^{k-2}nh_{n+1}+(k-2)h_{k}-\sum_{n=3}^{k-1}h_{n}'+(k-1)j-2}\Bigg)\nonumber,
\end{align}
and
\begin{align*}
&\sum_{l_1=1}^{k-2}\sum_{l_2=1}^{k-2}\cdots \sum_{l_{m-j}=1}^{k-2} \Big(\prod_{n=1}^{m-j}\big((l_{n}+1)t_{l_{n}+1}-(l_{n}+2)t_{l_{n}+2}\big)\Big)\Hfxn_{N,\emptyset,q+l_1+\cdots+l_{m-j}+(k-1)j}\\
&=\sum_{\substack{h_{2}+\cdots+h_{k}=m-j\\h_{2}\geq 0,\ldots,h_{k}\geq 0}}\prod_{n=2}^{k}(nt_{n})^{h_{n}}\Bigg(\sum_{\substack{h_3' \leq h_3 \\ \ldots \\ h_{k-1}'\leq h_{k-1}}}c_{h_{2},h_{3}',\ldots,h_{k-1}',h_{k}}^{(m+1,j)}\Hfxn_{N,\emptyset, q+\sum_{n=1}^{k-2}nh_{n+1}+(k-2)h_{k}-\sum_{n=3}^{k-1}h_{n}'+(k-1)j}\Bigg),
\end{align*}
where $c_{h_{2},h_{3}',\ldots,h_{k-1}',h_{k}}^{(i,j)}$ was defined in (\ref{coefficientsbeforegeneralD}).
Even though the quantities above are not well defined for $k=2$, the argument goes through almost verbatim. Hence this case is stated separately in the statement of the result. Finally, observe that the terms in the identities above, other than the term where $h_2=0,\ldots,h_{q-1}=0$, can be written in the form of the term involving $V_i(t_1,\ldots,t_k)$ in the statement of the result. Using these in combination, the desired result now follows.
\end{proof}
\subsection{Recursive relations for the partial derivatives of $\Hfxn_{N,\boldsymbol{\lambda}}$.}\label{recursiveformula}
In this subsection, we derive recursive relations for the partial derivatives of $\Hfxn_{N,\boldsymbol{\lambda}}$ with respect to $t_2, \ldots, t_k$ evaluated at $(0, \ldots, 0)$. The following observation is crucial for translating the results concerning $\Hfxn_{N,\boldsymbol{\lambda}}$ from Subsection \ref{truncation results} to those obtained in the present subsection.
\begin{lem}\label{usefulobservation}
Let $k\geq 2,N\geq 1$ be integers. Let $n_{2},\ldots,n_{k}$ be non-negative integers. Let $\boldsymbol{\lambda}$ be an integer partition. Suppose $\boldsymbol{\lambda}=(\lambda_1,\ldots, \lambda_m)$, and set $\lambda_{m+1}=\cdots=\lambda_N=0$ if $m\leq N-1$. Then the following identity holds for any $(R_{2},\ldots, R_{k})$ with $R_{i}\geq n_{i}$, $i=2,\ldots,k$,
\bea
&&\frac{1}{N^{\sum_{m=2}^{k}2n_{m}}}\sum_{\substack{\sum_{j=1}^{N}u_{m,j}=n_{m}\\m=2,\ldots,k}}\prod_{m=2}^{k}\binom{n_{m}}{u_{m,1},\ldots,u_{m,N}}\det_{0\leq i,j\leq N-1}\left(g_{i+j+\lambda_{N-j}+\sum_{m=2}^k m u_{m,j+1}}\big(\frac{t_1}{N^2}\big)\right)\label{thegoalobject}\\
=&&\frac{\partial^{n_{2}}}{\partial t_{2}^{n_{2}}}\frac{\partial^{n_{3}}}{\partial t_{3}^{n_{3}}}\cdots\frac{\partial^{n_{k}}}{\partial t_{k}^{n_{k}}}\Hfxn_{N,\boldsymbol{\lambda}}\Bigg|_{t_{2}=\cdots=t_{k}=0}\nonumber,
\eea
where $g_{n}(t_{1})$ is given in 
(\ref{definition of am}), and $\Hfxn_{N,\boldsymbol{\lambda}}$ is introduced in (\ref{generaldefinition of mathcalG}) (depending on $R_{2},\ldots,R_{k}$).
\end{lem}
\begin{proof}
Firstly, given $R_{2},\ldots,R_{k}$ with $R_{i}\geq n_{i}$, $i=2,\ldots,k$, we construct $\Hfxn_{N, \boldsymbol{\lambda}}$ as in (\ref{generaldefinition of mathcalG}). Then
\beas
\frac{\partial^{n_{2}}}{\partial t_{2}^{n_{2}}}\frac{\partial^{n_{3}}}{\partial t_{3}^{n_{3}}}\cdots\frac{\partial^{n_{k}}}{\partial t_{k}^{n_{k}}}\Hfxn_{N,\boldsymbol{\lambda}}\Bigg|_{t_{2}=\cdots=t_{k}=0}=&&
\sum_{\substack{\sum_{j=1}^{N}u_{m,j}=n_{m}\\m=2,\ldots,k}}\prod_{m=2}^{k}\binom{n_{m}}{u_{m,1},\ldots,u_{m,N}}\\
&&\det_{0\leq i,j\leq N-1}\left(\prod_{m=2}^{k}\frac{\partial^{u_{m,j+1}}}{\partial{t_{m}^{u_{m,j+1}}}}\hfxn_{i+j+\lambda_{N-j}}\Big|_{t_{2}=\cdots=t_{k}=0}\right).
\eeas
Also note, that by the definition of $\hfxn_{n}$,
\beas
\hfxn_{n}(R_{2},\ldots,R_{k};t_1,\ldots,t_k)=\sum_{m_2=0}^{R_{2}}\cdots\sum_{m_{k}=0}^{R_{k}}\left(\prod_{q=2}^{k}\frac{1}{N^{2m_{q}}}\right)\frac{t_2^{m_2}\ldots t_k^{m_k}}{m_2 !\ldots m_k!} g_{n+\sum_{q=2}^k q m_q}\big(\frac{t_1}{N^2}\big).
\eeas
Since $R_{i}\geq n_{i}$, and $0\leq u_{i,j}\leq n_{i}$ for $i=2,\ldots,k$ and $j=1,\ldots,N$, we then get that
\beas
\prod_{m=2}^{k}\frac{\partial^{u_{m,j+1}}}{\partial{t_{m}^{u_{m,j+1}}}}\hfxn_{i+j+\lambda_{N-j}}\Big|_{t_{2}=\cdots=t_{k}=0}=\frac{1}{N^{\sum_{m=2}^{k}2u_{m,j+1}}}g_{i+j+\lambda_{N-j}+\sum_{m=2}^k m u_{m,j+1}}\big(\frac{t_1}{N^2}\big).
\eeas
Combining these, the desired conclusion follows.
\end{proof}
Next, we proceed by obtaining recurrence relations for the quantities in Lemma \ref{usefulobservation}. In order to do so, we define, for any non-negative integers $n_{2},\ldots,n_{k}$ and any integer $l,q$ with $l\geq 1$ and $1\leq q\leq l$,
\bea\label{definitionforgeneralmathcalFlq}
&&\mathcal{G}_{N,l,q}^{(n_{2},n_{3},\ldots,n_{k})}(t_{1})\nonumber\\
\defeq&&\frac{1}{N^{\sum_{m=2}^{k}2n_{m}}}\sum_{\substack{\sum_{j=1}^{N}u_{m,j}=n_{m}\\m=2,\ldots,k}}\prod_{m=2}^{k}\binom{n_{m}}{u_{m,1},\ldots,u_{m,N}}\det_{0\leq i,j\leq N-1}\left(g_{i+j+\lambda_{N-j}^{(l,q)}+\sum_{m=2}^k m u_{m,j+1}}\big(\frac{t_1}{N^2}\big)\right),
\eea
where $\boldsymbol{\lambda}_{l,q}=(\lambda_{1}^{(l,q)},\ldots,\lambda_{N}^{(l,q)})$ is defined as in Definition \ref{defofhook}, and
\bea\label{definitionforgeneralmathcalFlqv2}
&&\mathcal{G}_{N}^{(n_{2},n_{3},\ldots,n_{k})}(t_{1})\nonumber\\
\defeq&&\frac{1}{N^{\sum_{m=2}^{k}2n_{m}}}\sum_{\substack{\sum_{j=1}^{N}u_{m,j}=n_{m}\\m=2,\ldots,k}}\prod_{m=2}^{k}\binom{n_{m}}{u_{m,1},\ldots,u_{m,N}}\det_{0\leq i,j\leq N-1}\left(g_{i+j+\sum_{m=2}^k m u_{m,j+1}}\big(\frac{t_1}{N^2}\big)\right).
\eea
\begin{prop}\label{recursionfirstprop}
Suppose $l\geq 3$, $1\leq q\leq l$. Let $k\geq 2$, $i_{2}, n_{2},n_{3},\ldots,n_{k},m \in \mathbb{N}\cup \{0\}$.
Define the vectors
\beas
\textbf{G}_{l}^{(n_{2},n_{3},\ldots,n_{k})} =\begin{pmatrix}
\mathcal{G}_{N,l,1}^{(n_{2},n_{3},\ldots,n_{k})}(t_{1}),
\ldots,
\mathcal{G}_{N,l,l}^{(n_{2},n_{3},\ldots,n_{k})}(t_{1}) 
\end{pmatrix}^\textnormal{T}
\eeas
and
\beas
\hat{\textbf{G}}_{l,m}^{(n_{2},n_{3},\ldots,n_{k})} = \begin{pmatrix}
\mathcal{G}_{N,l+m,m+2}^{(n_{2},n_{3},\ldots,n_{k})}(t_{1}), 
\ldots ,
\mathcal{G}_{N,l+m,l+m}^{(n_{2},n_{3},\ldots,n_{k})}(t_{1}),\: 0 
\end{pmatrix}^\textnormal{T}.
\eeas
Then, we have that
\bea\label{generalrecursiveconditionformathcalG}
\textbf{G}_{l}^{(i_{2},n_{3},\ldots,n_{k})}
=\mathbf{A}^{(l)}N^2\frac{\mathrm{d}}{\mathrm{d} t_1}
\begin{pmatrix}
\textbf{G}_{l-1}^{(i_{2},n_{3},\ldots,n_{k})} \\
0
\end{pmatrix}
+\sum_{n_2=0}^{i_2} \frac{(-2)^{i_2-n_2} \binom{i_2}{i_2-n_2} (i_2-n_2)! }{t_1^{i_2-n_2+1}}\textbf{H}_{l}^{(n_{2},n_{3},\ldots,n_{k})},
\eea
where $\textbf{H}_{l}^{(n_{2},n_{3},\ldots,n_{k})}$ admits the following explicit expression,
\begin{align*}
&\textbf{H}_{l}^{(n_{2},n_{3},\ldots,n_{k})}
= \mathbf{A}^{(l)}N^2 \Bigg( 
\sum_{m=3}^k (m-1) n_{m-1} \begin{pmatrix}
\textbf{G}_{l-1}^{(n_{2},n_{3},\ldots,n_{k})} \\
0
\end{pmatrix}
- \sum_{m=3}^k m n_{m}
\begin{pmatrix}
\textbf{G}_{l-1}^{(n_2,\ldots,n_{m-1}+1,n_m-1,\ldots,n_k)} \\
0
\end{pmatrix}\\
&+ k n_k \begin{pmatrix}
\textbf{G}_{l-1}^{(n_{2},n_{3},\ldots,n_{k})} \\
0
\end{pmatrix}
\Bigg) + (t_1 \mathbf{P}_1^{(l)} + N^2 {\mathbf{P}}_0^{(l)})
\textbf{G}_{l-1}^{(n_{2},n_{3},\ldots,n_{k})}
 + N^2\mathbf{P}_2^{(l)} 
\textbf{G}_{l-2}^{(n_{2},n_{3},\ldots,n_{k})}\\
&+ \sum_{m=1}^{k-2} (m+1) n_{m+1} (-1)^m \mathbf{P}_{m+2}^{(l)} 
\textbf{G}_{l+m}^{(n_2,\ldots,n_{m+1}-1,\ldots,n_k)} + (-1)^{k-1} k n_k \mathbf{P}_{k+1}^{(l)}\textbf{G}_{l+k-1}^{(n_2,n_{3},\ldots,n_k-1)}\\
&- \sum_{m=1}^{k-2} (m+2) n_{m+2} (-1)^m \mathbf{P}_{m+2}^{(l)}
\textbf{G}_{l+m}^{(n_2,\ldots,n_{m+2}-1,\ldots,n_k)}- \mathbf{A}^{(l)} \Bigg(-N^2 \frac{\mathrm{d}}{\mathrm{d} t_1}\Bigg) \Bigg\{\sum_{m=0}^{k-3} \Bigg[(m+2) n_{m+2} (-1)^m \\ & 
\hat{\textbf{G}}_{l,m}^{(n_2,\ldots,n_{m+2}-1,\ldots,n_k)} 
- (m+3) n_{m+3} (-1)^m
\hat{\textbf{G}}_{l,m}^{(n_2,\ldots,n_{m+3}-1,\ldots,n_k)} \Bigg]
+(-1)^k k n_k \hat{\textbf{G}}_{l,k-2}^{(n_2,n_3,\ldots,n_{k}-1)}
\Bigg\} 
\\ &+\mathbf{A}^{(l)} N^2 \sum_{m=4}^k \Bigg[(-1)^{m-1} (m-1) n_{m-1} \sum_{h=2}^{m-2} (-1)^h 
\hat{\textbf{G}}_{l,m-2-h}^{(n_2,\ldots,n_{h}+1,\ldots,n_{m-1}-1,\ldots,n_k)} 
\\
&-(-1)^{m-1} m n_{m} \sum_{h=2}^{m-2} (-1)^h 
\hat{\textbf{G}}_{l,m-2-h}^{(n_2,\ldots,n_{h}+1,\ldots,n_{m}-1,\ldots,n_k)}  \Bigg] 
 + \mathbf{A}^{(l)}N^2kn_k \sum_{h=2}^{k-1} (-1)^{k+h} 
\hat{\textbf{G}}_{l,k-1-h}^{(n_2,\ldots,n_{h}+1,\ldots,n_k-1)}.
\end{align*}
Moreover, the initial conditions for the recurrence above are given by (here, we suppress $t_1$ from notation):
\begin{align}\label{initialconditionforthegeneralcase}
\mathcal{G}_{N,1,1}^{(n_2,\ldots,n_k)}
&= -N^2 \frac{\mathrm{d}\mathcal{G}_{N}^{(n_2,\ldots,n_k)}}{\mathrm{d} t_1} , \nonumber\\
\mathcal{G}_{N,2,1}^{(n_2,\ldots,n_k)}
&= \frac{N^4}{2} \frac{\mathrm{d}^2\mathcal{G}_{N}^{(n_2,\ldots,n_k)}}{\mathrm{d}t_1^2} + \frac{N^2\mathcal{G}_{N}^{(n_2+1,\ldots,n_k)}}{2} , \nonumber\\
\mathcal{G}_{N,2,2}^{(n_2,\ldots,n_k)}
&= \frac{N^4}{2} \frac{\mathrm{d}^2\mathcal{G}_{N}^{(n_2,\ldots,n_k)}}{\mathrm{d}t_1^2}- \frac{N^2\mathcal{G}_{N}^{(n_2+1,\ldots,n_k)}}{2} .
\end{align}
\end{prop}
\begin{proof}
For any given $i_{2},n_{3},\ldots,n_{k}$ non-negative integers, set $R_{2}=i_{2}+1$ and $R_{i}=n_{i}+1$ for $i=3,\ldots,k$, and  then construct $\Hfxn_{N, \boldsymbol{\lambda}}$ as in (\ref{generaldefinition of mathcalG}). By Lemma \ref{usefulobservation}, for any nonnegative integers $q_{i}\leq R_{i}$, $ i=2,\ldots,k$, we have 
\beas
\frac{\partial^{q_{2}}}{\partial t_{2}^{q_{2}}}\frac{\partial^{q_{3}}}{\partial t_{3}^{q_{3}}}\cdots\frac{\partial^{q_{k}}}{\partial
t_{k}^{q_{k}}}\Hfxn_{N,\boldsymbol{\boldsymbol{\lambda}_{l,q}}}\Bigg|_{t_{2}=\cdots=t_{k}=0}=\mathcal{G}_{N,l,q}^{(q_{2},q_{3},\ldots,q_{k})}(t_{1}).
\eeas
Then (\ref{generalrecursiveconditionformathcalG}) follows from the differentiation of both sides of the recursion in Proposition \ref{prop:recursivevector}. For the initial conditions given in (\ref{initialconditionforthegeneralcase}), by an argument similar to the one above, we take $R_{i}=n_{i}+1$ for $i=2,\ldots,k$. Then differentiate both sides of \eqref{t1operateongeneralG} and the equations in Proposition \ref{expressionsforlength2} with
respect to $t_2,\ldots, t_k$ and let $t_2=\cdots=t_k=0$.
\end{proof}
\noindent As remarked earlier, our end goal is to establish a connection between $\mathcal{G}_{N}^{(n_{2},n_{3},\ldots,n_{k})}(t_{1})$ and a solution of the $\sigma$-Painlev\'e V equation. To carry out the induction process with this goal in mind, in addition to the the recursive relation on $\mathcal{G}_{N,l,q}^{(n_{2},n_{3},\ldots,n_{k})}(t_{1})$ described above, we now show that $\mathcal{G}_{N}^{(n_{2},n_{3},\ldots,n_{k})}(t_{1})$ admits an expression in terms of $\mathcal{G}_{N,l,q}^{(n_{2},n_{3},\ldots,n_{k})}(t_{1})$. To this end, we define for any non-negative integers $n_{2},\ldots,n_{k}$,
\begin{align}\label{F0repinL}
\mathsf{F}_0^{(n_{2},\ldots,n_k)}(t_{1}) =
N^2(N^2+(a+b)N)\mathcal{G}_{N}^{(n_{2},n_{3},\ldots,n_{k})}(t_{1})-(N^2a-t_{1})\mathcal{G}_{N,1,1}^{(n_{2},n_{3},\ldots,n_{k})}(t_{1}), 
\end{align}
and, for $j\geq 1$,
\begin{align}\label{formulaforFj}
\mathsf{F}_j^{(n_{2},\ldots,n_k)} (t_{1})=\sum_{i=1}^{j+1}(-1)^{i-1}\Bigg( N^2(2-2i+j-a) \mathcal{G}_{N,j+1,i}^{(n_{2},\ldots,n_{k})}(t_{1})+t_{1}\mathcal{G}_{N,j+1,i}^{(n_{2},\ldots,n_k)}(t_{1})\Bigg)& \nonumber
    \\+\sum_{i=1}^{j}(-1)^{i-1} {N^2(a+b-1+2i-j)} \mathcal{G}_{N,j,i}^{(n_{2},\ldots,n_k)}(t_{1}),&
\end{align}
where $\mathcal{G}_{N}^{(n_{2},n_{3},\ldots,n_{k})}(t_{1})$ and $\mathcal{G}_{N,l,q}^{(n_{2},n_{3},\ldots,n_{k})}(t_{1})$ are defined as in (\ref{definitionforgeneralmathcalFlqv2}) and (\ref{definitionforgeneralmathcalFlq}), respectively.
\begin{prop}\label{generalexpressionformathcalFn2nk}
Let $k\geq 2$, $N\geq 1$ be integers. Let $n_{2},\ldots,n_{k} \in \mathbb{N}\cup\{0\}$. Let $\mathsf{F}_0^{(n_{2},\ldots,n_k)}(t_{1})$ and $\mathsf{F}_j^{(n_{2},\ldots,n_k)}(t_{1})$ ($j\geq 1$) be defined as in (\ref{F0repinL}) and (\ref{formulaforFj}), respectively. 
Let $c_{h_{2},h_{3}',\ldots,h_{k-1}',h_{k}}^{(i,j)}$ be given as in (\ref{coefficientsbeforegeneralD}). Then, whenever
 $k=2$ and $n_2\geq 1$, we have that
\bea\label{expressionsforfi2special}
\mathcal{G}_{N}^{(n_{2})}(t_{1})=\frac{1}{N^2} \sum_{i=0}^{n_{2}-1} (n_{2}-1)! \sum_{j=0}^{n_{2}-1-i} \binom{n_2-1-j}{i} \frac{(-1)^{n_{2}-1-i-j}2^{n_{2}-1-j}}{j!}\frac{\mathsf{F}_{i}^{(j)}(t_{1})}{t_{1}^{n_{2}-j}} ,
\eea
whereas if  $k\geq 3$, $n_2\geq 1$, we have,
\begin{multline}\label{prop415first}
    \mathcal{G}_{N}^{(n_2,\ldots,n_k)}(t_1)=\frac{1}{N^2} \sum_{i_{2}=0}^{n_2-1} (-2)^{n_2-1-i_{2}}\frac{(n_{2}-1)!}{i_{2}!}
\frac{\mathsf{F}_0^{(i_{2},n_3,\ldots,n_k)}(t_{1})}{t_1^{n_2-i_{2}}}  \\ 
+  \sum_{i=1}^{n_2+\cdots+n_k-1} \frac{1}{N^2}\sum_{j=0}^{i} \binom{i}{j} 
\sum_{h_2+\cdots+h_k=i-j}  \binom{n_2-1}{h_2} \binom{n_k}{h_k+j} \binom{n_3}{h_3} \ldots \binom{n_{k-1}}{h_{k-1}} (h_{k}+j)!\\k^{j+h_{k}}
\prod_{n=2}^{k-1}n^{h_{n}}h_{n}! \sum_{i_{2}=0}^{n_2-1-h_2}\binom{n_2-1-h_2}{i_{2}} \frac{(i+n_2-1-h_2-i_{2})! (-2)^{n_2-1-h_2-i_{2}}}{t_1^{i+n_2-h_2-i_{2}} i!} \\
 \sum_{\substack{h_3' \le h_3 \\ \ldots \\ h_{k-1}' \le h_{k-1}}} c_{h_2,h_3',\ldots,h_{k-1}',h_{k}}^{(i+1,j)} \mathsf{F}_{\sum_{n=1}^{k-2}nh_{n+1}+(k-2)h_{k}-\sum_{n=3}^{k-1}h_{n}'+(k-1)j}^{(i_{2},n_3-h_3,\ldots,n_{k-1}-h_{k-1}, n_k-j-h_k)}(t_{1}).
\end{multline}
On the other hand, if there exists $q=3,\ldots,k-1$ such that $n_l=0$ for $l=2,\ldots,q-1$ and  $n_{q}\geq 1$, then
\begin{align*}
\mathcal{G}_N^{(0,\ldots,0,n_{q},\ldots, n_k)}(t_{1})
= &\frac{\mathsf{F}_{q-2}^{(0,\ldots,0,n_q-1,\ldots,n_k)}(t_{1})}{N^2t_1}
+\sum_{i=1}^{n_q+\cdots+n_k-1} \frac{1}{t_1^{i+1}N^2}\sum_{j=0}^i \binom{i}{j} \nonumber\\
& \sum_{h_q+\cdots+h_k=i-j} \binom{n_q-1}{h_q}\binom{n_k}{h_k+j} \binom{n_{q+1}}{h_{q+1}} \cdots \binom{n_{k-1}}{h_{k-1}} k^{j+h_{k}}(h_{k}+j)!\prod_{n=q}^{k-1}n^{h_{n}}h_{n}! \nonumber\\
& \;\;\;\;\;\;\;\;\;\;\;\;\; \sum_{\substack{h_q'\le h_q \\ \cdots \\ h_{k-1}' \le h_{k-1}}} c_{0,\ldots,0,h_q',\ldots,h_{k-1}',h_{k}}^{(i+1,j)} \mathsf{F}_{q-2+\sum_{n=q-1}^{k-2}nh_{n+1}+(k-2)h_{k}-\sum_{n=q}^{k-1}h_{n}'+(k-1)j}^{(0,\ldots,0,n_q-1-h_q,n_{q+1}-h_{q+1},\ldots,n_{k-1}-h_{k-1}, n_k-j-h_k)}(t_{1}). 
\end{align*}
Finally, if $n_l=0$ for $2\leq l\leq k-1$, and $n_{k}\geq 1$, then
\begin{align*}
\mathcal{G}_N^{(0,\ldots,0, n_k)}(t_{1})= \frac{\mathsf{F}_{k-2}^{(0,\ldots,0,n_k-1)}(t_{1})}{N^2t_1}
+\sum_{i=1}^{n_k-1}\frac{k^i}{t_{1}^{i+1}N^2}\frac{(n_{k}-1)!}{(n_{k}-1-i)!}
\sum_{j=0}^i \binom{i}{j} \: c_{0,\ldots,0,i-j}^{(i+1,j)}\mathsf{F}_{k-2+(k-1)i-(i-j)}^{(0,\ldots,0,n_{k}-1-i)}(t_{1}).
\end{align*}
\end{prop}
\begin{proof}
Firstly, for any given $n_{2},n_{3},\ldots,n_{k}$ non-negative integers, we choose $R_{i}=n_{i}+1$ for $i=2,\ldots,k$, and  then construct $\Hfxn_{N, \boldsymbol{\lambda}}$ as in (\ref{generaldefinition of mathcalG}).
Then, by means of the following relations that are consequences of their definitions,
\beas
&& \Hfxn_{N,\emptyset,h} = \sum_{j=1}^{h}(-1)^{j-1} \Hfxn_{N,\blambda_{h,j}},  \quad
\Hfxn_{N,\emptyset,h}^{(w)} = \sum_{j=1}^{h} (-1)^{j-1} (2N-2j+h) \Hfxn_{N,\blambda_{h,j}}, \quad \;\;\;\;\;\;\;\;\;\;\;(h\geq 1) , \\
&& \Hfxn_{N,\blambda,0} = N \Hfxn_{N,\blambda},   \quad \;\;\;\;\;\;\;\;\;\;\;\;\;\;\;
\Hfxn_{N,\blambda,0}^{(w)} = (N(N-1) + n) \Hfxn_{N,\blambda},
\quad \;\;\;\; (\blambda \text{ is a partition of } n) ,
\eeas
and by the definition of $\mathsf{Z}_{i}(R_{2},\ldots,R_{k};t_{1},\ldots, t_{k})$ given in  (\ref{defofdi}),
we have
\begin{align*}
\mathsf{Z}_{0}=&N^2(N^2+(a+b)N)\Hfxn_{N}-(N^2a-t_{1}) \Hfxn_{N,\blambda_{1,1}}\\
\mathsf{Z}_{j}=&\sum_{i=1}^{j+1}(-1)^{i-1}\Bigg( N^2(2-2i+j-a) \Hfxn_{N,\blambda_{j+1,i}}+t_{1}\Hfxn_{N,\blambda_{j+1,i}}\Bigg) \nonumber
    \\&+\sum_{i=1}^{j}(-1)^{i-1} {N^2(a+b-1+2i-j)} \Hfxn_{N,\blambda_{j,i}}, \,\,j\geq 1.
\end{align*}
Next, using Lemma \ref{usefulobservation}, we have that whenever $q_{i}\leq R_{i}$ for $i=2,\ldots,k$, and $j\geq 0$,
\beas
\mathsf{F}_j^{(q_{2},\ldots,q_k)} (t_{1})=\frac{\partial^{q_{2}}}{\partial t_{2}^{q_{2}}}\cdots\frac{\partial^{q_{k}}}{\partial t_{k}^{q_{k}}}\mathsf{Z}_{j}(R_{2},\ldots,R_{k};t_{1},\ldots,t_{k})\Big|_{t_{2}=\cdots=t_{k}=0}.
\eeas
Hence, we apply Lemma \ref{generaliedexpandingofpartialderivatives} by letting $m=n_q+\cdots+n_k$, where $q$ is the smallest index such that $n_q>0$. The desired conclusion then follows by differentiating both sides of the expressions in Lemma~\ref{generaliedexpandingofpartialderivatives} at $t_2=\cdots=t_k=0$.
\end{proof}
\subsection{Proofs of Theorems \ref{structureforgeneralfinitesizeforfiniteN} and \ref{mainpainlevethm}}\label{proofofmainTheorem}

We are now in position to prove Theorem \ref{structureforgeneralfinitesizeforfiniteN} by inducting on $k$ and $(n_{2},\ldots,n_{k})$.

\begin{proof}[Proof of Theorem \ref{structureforgeneralfinitesizeforfiniteN}]
By Proposition \ref{proprepresented as a Hankel} and the definition of $\mathcal{G}_{N}^{(n_{2},n_{3},\ldots,n_{k})}(t_{1})$ given in (\ref{definitionforgeneralmathcalFlqv2}),
\beas
\mathbb{E}_{N}^{(a,b)}\left[\mathrm{e}^{-t_{1}\sum_{j=1}^{N}\frac{1}{\mathsf{x}_{j}^{(N)}}}\prod_{q=2}^{k}\left(\sum_{j=1}^{N}(\mathsf{x}_{j}^{(N)})^{-q}\right)^{n_{q}}\right] =C_{N}^{(a,b)}N^{2\sum_{q=2}^{k}n_{q}}\mathcal{G}_{N}^{(n_{2},n_{3},\ldots,n_{k})}(t_{1})\nonumber,
\eeas
whenever  $a>\sum_{q=2}^{k}qn_{q}-1$.
Thus, (\ref{general structure1}) is equivalent to showing that
\begin{align}\label{general structure1equivalent}
\mathcal{G}_{N}^{(n_{2},n_{3},\ldots,n_{k})}(t_{1})=\frac{1}{t_{1}^{\sum_{q=2}^{k}qn_{q}-1}}
\sum_{m=0}^{\sum_{q=2}^{k}(q-1)n_{q}}t_{1}^{m-1}P_{m}^{(a,b)}(N;t_{1})\frac{\mathrm{d}^{m}}{\mathrm{d}t_{1}^{m}}\mathcal{G}_{N}^{(0,0,\ldots,0)}(t_{1})
\end{align}
for polynomials $P_{m}^{(a,b)}(N;t_{1})$ as specified in the theorem's formulation. Observe first that $\mathcal{G}_{N}^{(n_{2},n_{3},\ldots,n_{k})}$  is continuous in $t_1$, at the point $t_{1}=0$, so that it suffices to show the validity of (\ref{general structure1equivalent}) for all $t_1>0$. One can then prove the result from an inductive argument similar to that of  \cite[Theorem 1.9]{assiotis2024exchangeablearraysintegrablesystems}, with the main difference being the definitions of $\mathfrak{S}^{(k)}$ and $\mathfrak{S}_{l,j}^{(k)}$ in \cite[Definition 4.16]{assiotis2024exchangeablearraysintegrablesystems}. To be more precise, we introduce the modified versions of these sets for this particular case; $\mathfrak{C}^{(k)},\mathfrak{C}_{l,j}^{(k)}$, by defining them as follows. Let $\mathfrak{C}^{(k)} \subset (\mathbb{N}\cup \{0\})^{k-1}$ denote the collection of all tuples $(n_2,\ldots,n_k)$ such that \eqref{general structure1equivalent} holds, also including the $0$-tuple $(0,\ldots,0)$. In accordance with this, we then define $\mathfrak{C}_{l,j}^{(k)}$ to be the collection of tuples $(n_2,\ldots, n_k) $ such that
\begin{equation}\label{inductiveassumptiononGnlj}
\mathcal{G}_{N,l,j}^{(n_{2},n_{3},\ldots,n_{k})}(t_{1})=\frac{1}{t_{1}^{l+\sum_{q=2}^{k}qn_{q}-1}}
\sum_{m=0}^{l+\sum_{q=2}^{k}(q-1)n_{q}} t_1^{m-1} P_{l,j,m}^{(a,b)}(N;t_{1})\frac{\mathrm{d}^{m}}{\mathrm{d}t_{1}^{m}}\mathcal{G}_{N}^{(0,0,\ldots,0)}(t_{1})
\end{equation}
for polynomials $P_{l,j,m}^{(a,b)}(N;t_{1})$, $m\geq 1$, and $t_{1}^{-1}P_{l,j,0}^{(a,b)}(N;t_{1})$ in $t_1$ whose degrees do not exceed $l+\sum_{q=2}^{k}qn_{q}-m$ and $l+\sum_{q=2}^{k}qn_{q}-1$, respectively, subject to the additional constraint that they are also polynomials
in $N, a, b$, whose degrees in $N$ and in $a, b$ do not exceed $2l+\sum_{q=2}^{k}2(q-1)n_{q}$ and $l+\sum_{q=2}^{k}(q-1)n_{q}-1$, respectively. With this setting in mind, the result then follows by the induction scheme from \cite{assiotis2024exchangeablearraysintegrablesystems}, where we make use of Propositions \ref{recursionfirstprop} and \ref{generalexpressionformathcalFn2nk}.
\end{proof}
\noindent Before turning to the proof of Theorem \ref{mainpainlevethm}, we demonstrate how Theorem \ref{structureforgeneralfinitesizeforfiniteN}, in conjunction with Theorem \ref{thm:jointmom} and Proposition \ref{thm:laguerre}, allows one to deduce Corollary \ref{cor:integerhj}.
\begin{proof}[Proof of Corollary \ref{cor:integerhj}]
We first note that when $(h_{1},\ldots,h_{m})= (0,\ldots,0)$, the conclusion in this corollary follows from \cite{KeatingSnaithLfunctions}. We
now fix $h_1,\ldots, h_m\in \mathbb{N}\cup\{0\}$ and $(h_{1},\ldots,h_{m})\neq (0,\ldots,0)$, and consider the function,
\begin{align*}
     \varphi : \mathbb{R}_{\geq 0} &\to \mathbb{R}_+\\
    h_0 &\mapsto \lim_{N\to \infty}\mathbb{E}_{N}^{(h_0+\sum_{k=1}^m h_j-\frac{1}{2},-\frac{1}2{})}\left[\prod_{k=1}^m \mathfrak{e}_{N,k}^{h_k}\right],
\end{align*}
where $\mathfrak{e}_{N,k}$ is an abbreviation for 
$\mathfrak{e}_{N,k}^{(a,-\frac{1}{2})}$, 
as defined in (\ref{definitionofeNk}) with $a\defeq \sum_{k=0}^m h_k -\frac{1}{2}$.
This is initially well-defined only on $h_0\in [\frac{1}{2},\infty)$, thanks to the proof of Theorem \ref{thm:jointmom}. We claim that it is in fact a well-defined rational function in $h_0$ on all of $\mathbb{R}_{\geq 0}$. 
By Theorem \ref{structureforgeneralfinitesizeforfiniteN}, we have for $(n_{1},\ldots,n_{m})\neq (0,\ldots,0)$ and $a>\sum_{k=1}^{m}kn_{k}-1$,
\begin{align}\label{formula721}
\mathbb{E}_{N}^{(a,-\frac{1}{2})}\left[({\mathfrak{p}_{1,N}})^{n_{1}}\prod_{k=2}^{m}(\mathfrak{p}_{q,N})^{n_{k}}\right]=&\Bigg(\frac{1}{t^{\sum_{k=1}^{m}kn_{k}-1}}
\sum_{q=0}^{\sum_{k=2}^{m}(k-1)n_{k}}t^{q-1}P_{q,1}^{(a,-\frac{1}{2})}(N;t)\frac{\mathrm{d}^{q}}{\mathrm{d}t^{q}}\mathbb{E}_{N}^{(a,-\frac{1}{2})}\left[\mathrm{e}^{-t\mathfrak{p}_{1,N}}\right]\nonumber\\
&+\frac{1}{t^{\sum_{k=2}^{m}kn_{k}-1}}
\sum_{q=0}^{\sum_{k=2}^{m}(k-1)n_{k}}t^{q-1}P_{q,2}^{(a,-\frac{1}{2})}(N;t)\frac{\mathrm{d}^{q+n_{1}}}{\mathrm{d}t^{q+n_{1}}}\mathbb{E}_{N}^{(a,-\frac{1}{2})}\left[\mathrm{e}^{-t\mathfrak{p}_{1,N}}\right]\Bigg)\Bigg|_{t=0},
\end{align}
where $P_{m,1}^{(a,-\frac{1}{2})}(t)$ and $P_{m,2}^{(a,-\frac{1}{2})}(t)$ are polynomials of $t$. Moreover, the coefficients of these polynomials are themselves polynomials in $N,a$. Note that, $\mathbb{E}_{N}^{(a,-\frac{1}{2})}\left[\mathrm{e}^{-t\mathfrak{p}_{1,N}}\right]$ is $\sum_{k=1}^{m}kn_{k}$-times differentiable and it can be represented as $\exp\left(\int_{0}^{t}\frac{\sigma_{N}(t)}{t}\mathrm{d}t\right)$,
where $\sigma_{N}(t)$ satisfies 
$\sigma$-Painlev\'e V equation (\ref{equationforfiniteN}). Hence, by examining the asymptotic expansion of $\sigma_{N}(t)$ in powers of $t$ at $0$, up to $t^{\sum_{k=1}^{m}kn_{k}}$, we conclude that the coefficients of these terms can be obtained recursively and are rational functions of $a$, by an argument similar to that of \cite[Section 4.2]{Basor_2019}). Moreover, this implies that the first $\sum_{k=1}^{m}kn_{k}+1$ Taylor coefficients of $\mathbb{E}_{N}^{(a,-\frac{1}{2})}\left[\mathrm{e}^{-t\mathfrak{p}_{1,N}}\right]$   are of the form $\frac{\mathcal{Q}_2(N;a)}{\mathcal{Q}_1(a)}$, 
where $\mathcal{Q}_1(a)$ is a polynomial in $a$ and $\mathcal{Q}_2(N;a)$ is a polynomial in $a$, 
whose coefficients themselves are polynomials in $N$. After plugging this into the right-hand side of equation (\ref{formula721}), while noting that 
the coefficients of the singular terms $\frac{1}{t^{l}}$ 
for $l=1,\dots,\sum_{k=1}^m k n_{k}$ vanish, since the value of the left-hand side of (\ref{formula721}) is finite. We see that right-hand side of equation (\ref{formula721}) must be a rational function in $a$ 
of the form described above; and therefore, so does the left-hand side of equation (\ref{formula721}).  Hence, by Newton's formula, we readily conclude that $\mathbb{E}_{N}^{(a,-\frac{1}2{})}\left[\prod_{k=1}^m \mathfrak{e}_{N,k}^{h_k}\right]$ has the rational function form $\frac{Q_2(N;a)}{Q_1(a)}$, 
for certain polynomials  $Q_1(a)$ and $Q_2(N;a)$ in $a$, 
whose coefficients of the powers of $a$ in $Q_2(N;a)$ are polynomials in $N$ whenever $\Re(a)>\sum_{k=1}^{m}kh_{k}-1$. On the other hand, the map $ a\mapsto 
\mathbb{E}_{N}^{(a,-\frac{1}{2})}\left[\prod_{k=1}^m \mathfrak{e}_{N,k}^{h_k}\right]
$ is an analytic function by its definition on the set $\{a\in \mathbb{C}:\Re(a)> \sum_{k=1}^{m} h_k-1\}$.
Thus, by the identity theorem, it must equal 
$
\frac{Q_2(N;a)}{Q_1(a)}
$ whenever $\Re(a) > \sum_{k=1}^{m} h_k - 1$. Since $a=\sum_{k=0}^{m} h_k-\frac{1}{2}>\sum_{k=1}^{m} h_k-1$ when $h_{0}\geq 0$, $
\frac{Q_2(N;a)}{Q_1(a)}
$ is a rational function of $h_{0}$ on all of $h_{0}\geq 0$. Note that, when $h_{0}>\frac{1}{2}$,  $
N^{-2\sum_{k=1}^{m}{kh_{k}}}\frac{Q_2(N;a)}{Q_1(a)}
$ converges as $N\rightarrow \infty$. Since this is a rational function, this convergence holds for any $h_{0}\geq 0$ and the limit is a rational function of $h_{0}$, which is the desired  claim. 

To complete the proof, we observe that thanks to Proposition \ref{thm:laguerre}, the function
\begin{align*}
    \widetilde{\varphi} : \mathbb{R}_{\geq0} &\to \mathbb{R_+}\\
    h_0&\mapsto \mathbb{E}\left[\prod_{k=1}^m \mathfrak{e}_k \big(h_0+\sum_{k=1}^m h_j-\frac{1}{2}\big)^{h_k}\right]
\end{align*}
is well-defined, and rational in $h_0$ on $\mathbb{R}_{\geq 0}$, since all finite-dimensional averages taken with respect to the Laguerre ensembles therein can be reduced to finite linear combinations of ratios of Gamma functions, which simplify to rational functions. Combining this with the earlier claim, we see that under the given assumptions, both sides of \eqref{jointmomorth} (modulo a common factor) define a rational function on $h_0\in [0,\infty)$. Moreover, by Theorem \ref{thm:jointmom}, these two functions agree on $h_0>\frac{1}{2}$, and hence, they must agree everywhere. This completes the proof of the desired result.
\end{proof}

\begin{proof}[Proof of Theorem \ref{mainpainlevethm}]
First, we claim that whenever $t_1> 0$, we have
\begin{align}\label{eq:convergenceforpainleve0328}
       \frac{1}{N^{2\sum_{m=2}^k mn_m}}\mathbb{E}_{N}^{(a,b)}\left[\mathrm{e}^{-\frac{t_1}{N^2}\sum_{j=1}^N \frac{1}{\mathsf{x}^{(N)}_j}} \prod_{m=2}^k\left(\sum_{j=1}^N (\mathsf{x}^{(N)}_j)^{-m}\right)^{n_m}\right] &\xrightarrow[]{N\to \infty} \mathbb{E}\left[\mathrm{e}^{-t_1 \mathfrak{e}_1(a)} 
        \prod_{m=2}^k \left(\mathfrak{p}_m(a)\right)^{n_m}
        \right]. 
\end{align}
Indeed, an application of the elementary estimates $\mathrm{e}^{-x}\leq C_rx^{-r}$ for $x> 0$, and $\sum_{j=1}^{N}y_{j}^{-l}\leq \left(\sum_{j=1}^{N}y_{j}^{-1}\right)^{l}$  for $l\in \mathbb{N}$ and $y_{1},\ldots,y_{N}>0$ gives that,
\begin{equation*}
   {\mathrm{e}^{-\frac{t_1}{N^2}\mathfrak{p}_{1,N}^{(a,b)}}\prod_{m=2}^k  (\mathfrak{p}_{m,N}^{(a,b)})^{n_m} \leq  \frac{N^{2r}C_r}{t_1^r (\mathfrak{p}_{1,N}^{(a,b)})^r} \prod_{m=2}^k  (\mathfrak{p}_{m,N}^{(a,b)})^{n_m}\leq \frac{C_r N^{2r}}{t_1^r} (\mathfrak{p}_{1,N}^{(a,b)}})^{\sum_{m=2}^{k} m n_m -r}.
\end{equation*}
Hence, picking an appropriate $r\in \mathbb{N}$ such that the parameter  restriction in Proposition \ref{thm:main} holds, we conclude the the uniform integrability of the sequence
\beas
\left\{N^{-2\sum_{m=2}^k m n_m}\mathrm{e}^{-\frac{t_1}{N^2}\mathfrak{p}_{1,N}^{(a,b)}}\prod_{m=2}^k  (\mathfrak{p}_{m,N}^{(a,b)})^{n_m}\right\}_{N\geq 1}.
\eeas
Then, using the distributional convergence of the above sequence in Proposition \ref{thm:main} and  Skorokhod's representation theorem, we obtain (\ref{eq:convergenceforpainleve0328}).
Similarly, we have that for any $m$ with $m<a+1$,
\bea\label{0328afternoon1}
\frac{\mathrm{d}^m}{\mathrm{d}t_1^m} \mathbb{E}_{N}^{(a,b)}\left[\mathrm{e}^{-\frac{t_1}{N^2}\sum_{j=1}^N \frac{1}{\mathsf{x}_{j}^{(N)}}}\right] &\xrightarrow[]{N\to \infty} \frac{\mathrm{d}^m}{\mathrm{d}t_{1}^m} \mathbb{E}\left[\mathrm{e}^{-t_1\mathfrak{e}_1(a)}\right].
\eea
Secondly, we claim that the convergence in (\ref{eq:convergenceforpainleve0328}) and (\ref{0328afternoon1}) also hold at $t_1=0$. We only present here the argument for (\ref{eq:convergenceforpainleve0328}); which goes through mutatis mutandis for the convergence in (\ref{0328afternoon1}). On one hand, by an argument similar to that of Corollary \ref{cor:integerhj}, for fixed $b>-1$, the left-hand side of (\ref{eq:convergenceforpainleve0328}) at $t_1 = 0$ can be expressed as a rational function of $a$ and $N$ of the form
$
\frac{\tilde{Q}_2(N;a)}{\tilde{Q}_1(a)}$, for $\Re(a) > \sum_{m=2}^{k} m n_m - 1$, where $\tilde{Q}_1(a)$ is a polynomial in $a$, and $\tilde{Q}_2(N;a)$ is a polynomial in $a$ whose coefficients are themselves polynomials in $N$. Meanwhile, the right-hand side of (\ref{eq:convergenceforpainleve0328}) at $t_1=0$ is also a rational function of $a$ for $a > \sum_{m=2}^{k} m n_m - 1$.
By Proposition \ref{thm:main}, when $a > \sum_{m=2}^{k} m n_m$, the convergence of (\ref{eq:convergenceforpainleve0328}) at $t=0$ holds. Hence, by the properties of rational functions, this convergence extends to the full range $a > \sum_{m=2}^{k} m n_m - 1$.

We now prove equation (\ref{formula1}) of the theorem. Together with Theorem \ref{structureforgeneralfinitesizeforfiniteN}, if one defines the polynomials $\mathcal{P}_m^{(a)}(n_{2},\ldots,n_{k};t_{1})$ by,
    \begin{equation}
    \label{eq:limitingpoly}
        \mathcal{P}_m^{(a)}(n_{2},\ldots,n_{k};t_{1})=\lim_{N\to \infty}\frac{P_{m}^{(a,b)}(N;t_{1})}{N^{\sum_{m=2}^k 2(m-1)n_m}},
    \end{equation}
the result then follows, modulo the validity of the properties of  $\mathcal{P}_m^{(a)}(n_{2},\ldots,n_{k};t_{1})$ stated in the theorem. Firstly, we claim that the $t_1$-degree of the polynomial $t_{1}^{m-1} \mathcal{P}_m^{(a)}(n_{2},\ldots,n_{k};t_{1})$ is $\sum_{q=2}^k (q-1)n_q-1$, which is in fact less than the \textit{a priori} degree that Theorem \ref{structureforgeneralfinitesizeforfiniteN} yields. The precise reason for this becomes more clear if one decomposes
    \begin{equation*}
        \mathcal{G}_{N,l,q}^{(n_2,\ldots,n_k)}(t_1) = N^{2l+\sum_{m=2}^k 2(m-1)n_m}  \mathcal{M}_{N,l,q}^{(n_2,\ldots,n_k)}(t_1) + O(N^{2l+\sum_{m=2}^k 2(m-1)n_m-1}),
    \end{equation*}
for every $1\leq q\leq l$, (while doing the same for $ \mathcal{G}_{N}^{(n_2,\ldots,n_k)}$), and establishes equivalent statements for $\mathcal{M}_{N,l,q}^{(n_2,\ldots,n_k)}$ as in Propositions \ref{recursionfirstprop} and \ref{generalexpressionformathcalFn2nk} by matching the leading-order coefficients therein.
Indeed, doing so and 
carefully carrying out the same inductive process as in the proof of Theorem \ref{structureforgeneralfinitesizeforfiniteN}, we find that the polynomials in $t_{1}$ appearing in the terms $\mathcal{M}_{N,l,q}^{(n_2,\ldots,n_k)}$ have the desired degrees as stated in Theorem \ref{mainpainlevethm}. 

The fact that the limit in (\ref{eq:limitingpoly}) is independent of the parameter $b$ can be seen similarly based on the following observations.  On one hand, we observe that in Proposition \ref{recursionfirstprop}, 
   the only term involving $b$ that appears in the expression of $\mathbf{H}_{l}^{(n_{2},\ldots,n_{k})} $ is the term $N^2\mathbf{P}_2^{(l)} 
\textbf{G}_{l-2}^{(n_{2},\ldots,n_{k})}$. Moreover, using the explicit expression for $\mathbf{P}_2^{(l)}$, it is easily seen that the highest-order contribution in $N$ coming from this particular term is the same as that of 
$N^2\mathbf{M}_2^{(l)}
\textbf{G}_{l-2}^{(n_{2},\ldots,n_{k})}$,
where $\mathbf{M}_2^{(l)}$ is given in the Appendix, and in particular, is independent of $b$. On the other hand, the only terms involving $b$ in Proposition \ref{generalexpressionformathcalFn2nk} appear in the expressions for $\mathsf{F}_0^{(n_{2},\ldots,n_k)}(t_{1})$ and $\mathsf{F}_j^{(n_{2},\ldots,n_k)}(t_{1})$, $(j\geq 1)$, with the $b$-dependent terms contributing in a strictly smaller order of $N$.  Hence, proceeding with the inductive process, where the inductive assumptions are in accordance with the formula  (\ref{inductiveassumptiononGnlj}), while keeping track of the coefficients of $N^{2l+\sum_{m=2}^k 2(m-1)n_m}$ by means of these two observations, we arrive at the desired conclusion regarding $b$-independence of the polynomials in \eqref{eq:limitingpoly}. This concludes the proof of the theorem.
\end{proof}
\noindent An inspection of the proof of Theorem \ref{mainpainlevethm} reveals that we have in fact produced a recursive algorithm for computing
$\mathcal{M}_{N,l,q}^{(n_2,\ldots,n_k)}$, and consequently, for computing the leading coefficients of the joint moments.
\begin{proof}[Proof of Corollary \ref{some examples}]
By (\ref{initialconditionforthegeneralcase}) for $n_{2}=0$ and $n_{2}=1$, and by (\ref{expressionsforfi2special}) for $n_{2}=1$ and $n_{2}=2$, we obtain, for $b>-1$ and $t\geq 0$, when $a>1$,
\beas
& & \E_{N}^{(a,b)} \left[ \mathrm{e}^{-t \sum_{j=1}^N \frac{1}{\mathsf{x}_j^{(N)}}  } \left( \sum_{j=1}^N \frac{1}{(\mathsf{x}_j^{(N)})^2}\right) \right] \\
&=& \left(\frac{a}{t}-1\right) \frac{\mathrm{d}}{\mathrm{d}t} \E_{N}^{(a,b)}\left[  \mathrm{e}^{-t \sum_{j=1}^N \frac{1}{\mathsf{x}_j^{(N)}}  }  \right] + \frac{(a+b+N)N}{t} \E_{N}^{(a,b)}\left[  \mathrm{e}^{-t \sum_{j=1}^N \frac{1}{\mathsf{x}_j^{(N)}}  }  \right]
\eeas
and when $a>3$,
\beas
& & \E_{N}^{(a,b)} \left[ \mathrm{e}^{-t \sum_{j=1}^N \frac{1}{\mathsf{x}_j^{(N)}}  } \left( \sum_{j=1}^N \frac{1}{(\mathsf{x}_j^{(N)})^2}\right)^2 \right] \\
&=& \frac{(a-t)^2+2}{t^2} \frac{\mathrm{d}^2}{\mathrm{d}t^2} \E_{N}^{(a,b)}\left[\mathrm{e}^{-t \sum_{j=1}^N \frac{1}{\mathsf{x}_j^{(N)}}  }  \right] \\
&& + \frac{N (a+b+N) ((1+N^2+(a+b)N) t-3a)}{t^3} \E_{N}^{(a,b)}\left[ \mathrm{e}^{-t \sum_{j=1}^N \frac{1}{\mathsf{x}_j^{(N)}}  }  \right] \\
&&+ \frac{-3a^2+(2N^2a+2a(a+b)N+a-2b )t-(2N^2+2(a+b)N ) t^2 }{t^3}
\frac{\mathrm{d}}{\mathrm{d}t} \E_N^{(a,b)}\left[  \mathrm{e}^{-t \sum_{j=1}^N \frac{1}{\mathsf{x}_j^{(N)}}  }  \right].
\eeas
The result then follows by doing the rescaling $t \mapsto  \frac{t}{N^2}$, dividing by an appropriate power of $N$, and and taking the limit as $N \to \infty$.

\end{proof}

\addcontentsline{toc}{section}{Appendix}
\section*{Appendix}
\label{app:scripts}

We now turn to defining explicitly the matrices that featured in the proofs of Theorems \ref{structureforgeneralfinitesizeforfiniteN} and \ref{mainpainlevethm}. Firstly, as mentioned previously, for $m\neq 0,2$, the matrices $\mathbf{P}_m^{(l)}$ are identical to the matrices $\mathbf{Q}_m^{(l)}$ appearing in \cite[Proposition 4.11]{assiotis2024exchangeablearraysintegrablesystems}. As for the remaining matrices, they are given as follows. 

\begin{align*}
  \left(\mathbf{A}^{(l)}\right)_{ij} &=\begin{cases}
\frac{(-1)^{i+j-1}}{j(j+1)} & j\geq i; \\
-\frac{1}{i} & j=i-1; \\
0  & j<i-1; \\
\frac{(-1)^{i-1}}{l} & j=l.
\end{cases}\\
\left(\mathbf{P}_{0}^{(l)}\right)_{ij}&=
\begin{cases}\label{definitionofC1}
(-1)^{i+j}\frac{j(l-j-2)+1-a}{j(j+1)}& \text{if } i\leq j\leq l-1;\\
\frac{j(j+a)-l+1}{j+1}& \text{if } j=i-1; \\
  0 & \text{if } j<i-1.
\end{cases}
\\
\left(\widetilde{\mathbf{P}}_1^{(l)}\right)_{ij} &= \begin{cases}
\frac{(-1)^{i+j}}{j+1} (\frac{2N}{j}-j-2+l) & i\leq j \leq l-1; \\
-\frac{j(2N-j)+l-j-1}{j+1} & j=i-1; \\
0 & 1\leq j \leq i-2.
\end{cases}
\\
\left(\mathbf{P}_{2}^{(l)}\right)_{ij}  &= \begin{cases}
(-1)^{i+j}\frac{N^2+(a+b)N+j(j+3)-(j+2)(l-2)+2(a+b-1)}{(j+1)(j+2)}  & i-1\leq j \leq l-2; \\
\frac{-N^2-(a+b)N+j(a+b+j)}{j+2}& j=i-2; \\
0 & 1\leq j <  i-2.
\end{cases}
\\
\left(\widehat{\mathbf{P}}_{2}^{(l)}\right)_{ij} &= \begin{cases}
(-1)^{i+j}\frac{N+2}{(j+1)(j+2)}  & i-1\leq j \leq l-2; \\
\frac{j-N}{j+2} & j=i-2; \\
0 & 1\leq j <  i-2.
\end{cases}
\\
\left(\widetilde{\mathbf{P}}_2^{(l)}\right)_{ij} &= \begin{cases}
(-1)^{i+j}\frac{N^2+3N-j(j+3)+(j+2)(l-2)}{(j+1)(j+2)}  & i-1\leq j \leq l-2; \\
\frac{-N^2+N(2j+1)-j(j+1)}{j+2} & j=i-2; \\
0 & 1\leq j < i-2.
\end{cases}
\\
\left(\mathbf{M}_{2}^{(l)}\right)_{ij}  &= \begin{cases}
(-1)^{i+j}\frac{N^2}{(j+1)(j+2)}  & i-1\leq j \leq l-2; \\
\frac{-N^2}{j+2}& j=i-2; \\
0 & 1\leq j <  i-2.
\end{cases}
\end{align*}

\bibliographystyle{plain}
\bibliography{main}

\begin{thebibliography}{10}

\bibitem{altugetal}
S.~Ali Altuğ, Sandro Bettin, Ian Petrow, Rishikesh, and Ian Whitehead.
\newblock {A recursion formula for moments of derivatives of random matrix polynomials}.
\newblock {\em The Quarterly Journal of Mathematics}, 65(4):1111--1125, 2014.

\bibitem{alvarez2023asymptotics}
Emilia Alvarez, Pierre Bousseyroux, and Nina~C Snaith.
\newblock Asymptotics of non-integer moments of the logarithmic derivative of characteristic polynomials over {$SO (2N+ 1) $}.
\newblock {\em arXiv preprint arXiv:2303.07813}, 2023.

\bibitem{AlvarezSnaith}
Emilia Alvarez and Nina~C Snaith.
\newblock Moments of the logarithmic derivative of characteristic polynomials from {$SO(N)$} and {$USp(2N)$}.
\newblock {\em J. Math. Phys.}, 61(10):103506, 32, 2020.

\bibitem{andrade2024joint}
Julio Andrade and Christopher Best.
\newblock Joint moments of derivatives of characteristic polynomials of random symplectic and orthogonal matrices.
\newblock {\em Journal of Physics A: Mathematical and Theoretical}, 2024.

\bibitem{assiotiswishart}
Theodoros Assiotis.
\newblock Ergodic decomposition for inverse {W}ishart measures on infinite positive-definite matrices.
\newblock {\em Symmetry, Integrability and Geometry: Methods and Applications}, 2019.

\bibitem{RandomAnalytic}
Theodoros Assiotis.
\newblock Random entire functions from random polynomials with real zeros.
\newblock {\em Adv. Math.}, 410:Paper No. 108701, 28, 2022.

\bibitem{ABGS}
Theodoros Assiotis, Benjamin Bedert, Mustafa~Alper Gunes, and Arun Soor.
\newblock On a distinguished family of random variables and {P}ainlevé equations.
\newblock {\em Probability and Mathematical Physics}, 2(3):613–642, 2021.

\bibitem{assiotis2024exchangeablearraysintegrablesystems}
Theodoros Assiotis, Mustafa~Alper Gunes, Jonathan~P. Keating, and Fei Wei.
\newblock Exchangeable arrays and integrable systems for characteristic polynomials of random matrices.
\newblock {\em arXiv preprint arXiv:2407.19233}, 2024.

\bibitem{CircJacobiBeta}
Theodoros Assiotis, Mustafa~Alper Gunes, and Arun Soor.
\newblock Convergence and an explicit formula for the joint moments of the circular {J}acobi {$\beta$}-ensemble characteristic polynomial.
\newblock {\em Math. Phys. Anal. Geom.}, 25(2):Paper No. 15, 24, 2022.

\bibitem{assiotis2022joint}
Theodoros Assiotis, Jonathan~P Keating, and Jon Warren.
\newblock On the joint moments of the characteristic polynomials of random unitary matrices.
\newblock {\em International Mathematics Research Notices}, 2022(18):14564--14603, 2022.

\bibitem{assiotis2020boundary}
Theodoros Assiotis and Joseph Najnudel.
\newblock The boundary of the orbital beta process.
\newblock {\em Moscow Mathematical Journal}, 21:659–694, 2021.

\bibitem{Bailey_2019}
Emma~C Bailey, Sandro Bettin, Gordon Blower, J~Brian Conrey, Andrei Prokhorov, Michael~O Rubinstein, and Nina~C Snaith.
\newblock Mixed moments of characteristic polynomials of random unitary matrices.
\newblock {\em Journal of Mathematical Physics}, 60(8), 2019.

\bibitem{barhoumi2020new}
Yacine Barhoumi-Andr{\'e}ani.
\newblock A new approach to the characteristic polynomial of a random unitary matrix.
\newblock {\em arXiv preprint arXiv:2011.02465}, 2020.

\bibitem{Basor_2019}
Estelle Basor, Pavel Bleher, Robert Buckingham, Tamara Grava, Alexander Its, Elizabeth Its, and Jonathan~P Keating.
\newblock A representation of joint moments of {CUE} characteristic polynomials in terms of {P}ainlevé functions.
\newblock {\em Nonlinearity}, 32(10):4033–4078, 2019.

\bibitem{BorodinDet}
Alexei Borodin.
\newblock Determinantal point processes.
\newblock In {\em The {O}xford handbook of random matrix theory}, pages 231--249. Oxford Univ. Press, Oxford, 2011.

\bibitem{Borodin_2001}
Alexei Borodin and Grigori Olshanski.
\newblock Infinite random matrices and ergodic measures.
\newblock {\em Communications in Mathematical Physics}, 223(1):87–123, 2001.

\bibitem{SingularWeightPainleve2}
L~Brightmore, Francesco Mezzadri, and Man~Y Mo.
\newblock A matrix model with a singular weight and {P}ainlev\'e{} {III}.
\newblock {\em Comm. Math. Phys.}, 333(3):1317--1364, 2015.

\bibitem{ChenChenFan}
Min Chen, Yang Chen, and En-Gui Fan.
\newblock The {R}iemann-{H}ilbert analysis to the {P}ollaczek-{J}acobi type orthogonal polynomials.
\newblock {\em Stud. Appl. Math.}, 143(1):42--80, 2019.

\bibitem{ChenDai}
Yang Chen and Dan Dai.
\newblock Painlev\'e{} {V} and a {P}ollaczek-{J}acobi type orthogonal polynomials.
\newblock {\em J. Approx. Theory}, 162(12):2149--2167, 2010.

\bibitem{Conrey}
J~Brian Conrey.
\newblock The fourth moment of derivatives of the {R}iemann zeta-function.
\newblock {\em Quart. J. Math. Oxford Ser. (2)}, 39(153):21--36, 1988.

\bibitem{CFKRS1}
J~Brian Conrey, David~W Farmer, Jon~P Keating, Michael~O Rubinstein, and Nina~C Snaith.
\newblock Autocorrelation of random matrix polynomials.
\newblock {\em Comm. Math. Phys.}, 237(3):365--395, 2003.

\bibitem{CFKRS2}
J~Brian Conrey, David~W Farmer, Jon~P Keating, Michael~O Rubinstein, and Nina~C Snaith.
\newblock Integral moments of {$L$}-functions.
\newblock {\em Proc. London Math. Soc. (3)}, 91(1):33--104, 2005.

\bibitem{ConreyGhosh}
J~Brian Conrey and Anish Ghosh.
\newblock On mean values of the zeta-function. {II}.
\newblock {\em Acta Arith.}, 52(4):367--371, 1989.

\bibitem{conreyetal}
J~Brian Conrey, Michael~O Rubinstein, and Nina~C Snaith.
\newblock Moments of the derivative of characteristic polynomials with an application to the {R}iemann zeta function.
\newblock {\em Communications in Mathematical Physics}, 267(3):611--629, 2006.

\bibitem{Hurwitz}
Fabio~Deelan Cunden, Antoine Dahlqvist, and Neil O'Connell.
\newblock Integer moments of complex {W}ishart matrices and {H}urwitz numbers.
\newblock {\em Ann. Inst. Henri Poincar\'e{} D}, 8(2):243--268, 2021.

\bibitem{MomentsHypergeometric}
Fabio~Deelan Cunden, Francesco Mezzadri, Neil O'Connell, and Nick Simm.
\newblock Moments of random matrices and hypergeometric orthogonal polynomials.
\newblock {\em Comm. Math. Phys.}, 369(3):1091--1145, 2019.

\bibitem{Wigner-Smith}
Fabio~Deelan Cunden, Francesco Mezzadri, Nick Simm, and Pierpaolo Vivo.
\newblock Correlators for the {W}igner-{S}mith time-delay matrix of chaotic cavities.
\newblock {\em J. Phys. A}, 49(18):18LT01, 20, 2016.

\bibitem{BallisticChaotic}
Fabio~Deelan Cunden, Francesco Mezzadri, Nick Simm, and Pierpaolo Vivo.
\newblock Large-{$N$} expansion for the time-delay matrix of ballistic chaotic cavities.
\newblock {\em J. Math. Phys.}, 57(11):111901, 16, 2016.

\bibitem{Dehaye2008}
Paul-Olivier Dehaye.
\newblock Joint moments of derivatives of characteristic polynomials.
\newblock {\em Algebra Number Theory 2, 31–68}, 2008.

\bibitem{Dehaye2010note}
Paul-Olivie\vspace{0mm}r Dehaye.
\newblock A note on moments of derivatives of characteristic polynomials.
\newblock {\em 22nd International Conference on Formal Power Series and Algebraic Combinatorics, 681-692, Discrete Math. Theor. Comput. Sci. Proc. AN}, 2010.

\bibitem{forrester_importanceofselberg}
Peter Forrester and Ole Warnaar.
\newblock The importance of the selberg integral.
\newblock {\em Bulletin of the American Mathematical Society}, 45:489–534, 2008.

\bibitem{ForresterBessel}
Peter~J Forrester.
\newblock The spectrum edge of random matrix ensembles.
\newblock {\em Nuclear Phys. B}, 402(3):709--728, 1993.

\bibitem{ForresterBook}
Peter~J Forrester.
\newblock {\em Log-gases and random matrices}, volume~34 of {\em London Mathematical Society Monographs Series}.
\newblock Princeton University Press, Princeton, NJ, 2010.

\bibitem{forrester2022joint}
Peter~J Forrester.
\newblock Joint moments of a characteristic polynomial and its derivative for the circular $\beta$-ensemble.
\newblock {\em Probability and Mathematical Physics}, 3(1):145--170, 2022.

\bibitem{ForresterWittePainleve1}
Peter~J Forrester and Nicholas~S Witte.
\newblock Application of the {$\tau$}-function theory of {P}ainlev\'{e} equations to random matrices: {PIV}, {PII} and the {GUE}.
\newblock {\em Comm. Math. Phys.}, 219(2):357--398, 2001.

\bibitem{ForresterWittePainleve2}
Peter~J Forrester and Nicholas~S Witte.
\newblock Application of the {$\tau$}-function theory of {P}ainlev\'{e} equations to random matrices: {$\rm P_V$}, {$\rm P_{III}$}, the {LUE}, {JUE}, and {CUE}.
\newblock {\em Comm. Pure Appl. Math.}, 55(6):679--727, 2002.

\bibitem{ForresterWittePainleveIII}
Peter~J Forrester and Nicholas~S Witte.
\newblock Boundary conditions associated with the {P}ainlev\'e{} {III{$'$}} and {V} evaluations of some random matrix averages.
\newblock {\em J. Phys. A}, 39(28):8983--8995, 2006.

\bibitem{gharakhloo2023modulated}
Roozbeh Gharakhloo and Nicholas~S Witte.
\newblock Modulated bi-orthogonal polynomials on the unit circle: The $2j-k$ and $j-2k$ systems.
\newblock {\em Constructive Approximation}, 58(1):1--74, 2023.

\bibitem{gunes2022characteristic}
Mustafa~Alper Gunes.
\newblock Characteristic polynomials of orthogonal and symplectic random matrices, {J}acobi ensembles \& {L}-functions.
\newblock {\em Random Matrices: Theory and Applications}, 13(02):2450006, 2024.

\bibitem{hammer1953higher}
Carl Hammer.
\newblock Higher transcendental functions, volume {I}: by {Harry Bateman} (compiled by the staff of the {Bateman Manuscript Project}). 302 pages, {New York, McGraw-Hill Book Co., Inc., 1953}, 1953.

\bibitem{dianejacobi}
Diane Holcomb and Gregorio~R. Moreno~Flores.
\newblock Edge scaling of the $\beta$-jacobi ensemble.
\newblock {\em Journal of Statistical Physics}, 149(6):1136--1160, 2012.

\bibitem{Hughes}
Christopher~P Hughes.
\newblock On the characteristic polynomial of a random unitary matrix and the {R}iemann zeta function.
\newblock {\em PhD Thesis, University of Bristol}, 2001.

\bibitem{Ingham}
Albert~E Ingham.
\newblock Mean-{V}alue {T}heorems in the {T}heory of the {R}iemann {Z}eta-{F}unction.
\newblock {\em Proc. London Math. Soc. (2)}, 27(4):273--300, 1927.

\bibitem{JohanssonDet}
Kurt Johansson.
\newblock Random matrices and determinantal processes.
\newblock In {\em Mathematical statistical physics}, pages 1--55. Elsevier B. V., Amsterdam, 2006.

\bibitem{KatzSarnak}
Nicholas~M. Katz and Peter Sarnak.
\newblock {\em Random matrices, {F}robenius eigenvalues, and monodromy}, volume~45 of {\em American Mathematical Society Colloquium Publications}.
\newblock American Mathematical Society, Providence, RI, 1999.

\bibitem{KeatingSnaithLfunctions}
Jon~P Keating and Nina~C Snaith.
\newblock Random matrix theory and {$L$}-functions at {$s=1/2$}.
\newblock {\em Comm. Math. Phys.}, 214(1):91--110, 2000.

\bibitem{keatingsnaith}
Jon~P Keating and Nina~C Snaith.
\newblock Random matrix theory and $\zeta$(1/2+it).
\newblock {\em Communications in Mathematical Physics}, 214(1):57--89, 2000.

\bibitem{keating-fei}
Jon~P Keating and Fei Wei.
\newblock {Joint moments of higher order derivatives of CUE characteristic polynomials II: Structures, recursive relations, and applications}.
\newblock {\em Nonlinearity}, 37:085009, 2024.

\bibitem{keatingwei}
Jonathan~P Keating and Fei Wei.
\newblock Joint moments of higher order derivatives of {CUE} characteristic polynomials {I}: asymptotic formulae.
\newblock {\em International Mathematics Research Notices}, 2024:9607--9632, 2024.

\bibitem{valkoli}
Yun Li and Benedek Valk\'{o}.
\newblock Operator level limit of the circular {J}acobi $\beta$-ensemble.
\newblock {\em Random Matrices: Theory and Applications}, 11(04):2250043, 2022.

\bibitem{SingularWeightPainleve1}
Francesco Mezzadri and Man~Yue Mo.
\newblock On an average over the {G}aussian unitary ensemble.
\newblock {\em Int. Math. Res. Not.}, (18):3486--3515, 2009.

\bibitem{MezzadriSimm}
Francesco Mezzadri and Nick Simm.
\newblock Tau-function theory of chaotic quantum transport with {$\beta=1,2,4$}.
\newblock {\em Comm. Math. Phys.}, 324(2):465--513, 2013.

\bibitem{RamirezRider}
Jos{e}~A Ram{i}rez and Brian Rider.
\newblock Diffusion at the random matrix hard edge.
\newblock {\em Comm. Math. Phys.}, 288(3):887--906, 2009.

\bibitem{StochasticAiry}
Jose{}~A Ramirez, Brian Rider, and B\'alint Vir\'ag.
\newblock Beta ensembles, stochastic {A}iry spectrum, and a diffusion.
\newblock {\em J. Amer. Math. Soc.}, 24(4):919--944, 2011.

\bibitem{simm2024moments}
Nick Simm and Fei Wei.
\newblock On moments of the derivative of cue characteristic polynomials and the riemann zeta function.
\newblock {\em arXiv preprint arXiv:2409.03687}, 2024.

\bibitem{ValkoViragOperator}
Benedek Valk\'o and B\'alint Vir\'ag.
\newblock The {$\rm Sine_\beta$} operator.
\newblock {\em Invent. Math.}, 209(1):275--327, 2017.

\bibitem{ValkoViragSecular}
Benedek Valk\'o and B\'alint Vir\'ag.
\newblock The many faces of the stochastic zeta function.
\newblock {\em Geom. Funct. Anal.}, 32(5):1160--1231, 2022.

\bibitem{winn2012derivative}
Brian Winn.
\newblock Derivative moments for characteristic polynomials from the {CUE}.
\newblock {\em Communications in Mathematical Physics}, 315(2):531–562, 2012.

\bibitem{ForresterWitteCauchy}
Nicholas~S Witte and Peter~J Forrester.
\newblock Gap probabilities in the finite and scaled {C}auchy random matrix ensembles.
\newblock {\em Nonlinearity}, 13(6):1965--1986, 2000.

\end{thebibliography}

\noindent{\sc School of Mathematics, University of Edinburgh, James Clerk Maxwell Building, Peter Guthrie Tait Rd, Edinburgh EH9 3FD, U.K.}\newline
\href{mailto:theo.assiotis@ed.ac.uk}{\small theo.assiotis@ed.ac.uk}

\bigskip
\noindent
{\sc Department of Mathematics, Princeton University, 
Fine Hall, 304 Washington Rd, Princeton, NJ 08544, USA.}\newline
\href{mailto:magunes@princeton.edu}{\small magunes@princeton.edu}

\bigskip
\noindent
{\sc Mathematical Institute, University of Oxford, Andrew Wiles Building, Woodstock Road, Oxford, OX2 6GG, UK. }\newline
\href{keating@maths.ox.ac.uk}{\small keating@maths.ox.ac.uk}

\bigskip
\noindent
{\sc Department of Mathematics, University of Sussex, Brighton, BN1 9RH, UK. }\newline
\href{weif0831@gmail.com}{\small weif0831@gmail.com}

\end{document}